\documentclass{lmcs} 
\pdfoutput=1

\usepackage{lastpage}
\lmcsdoi{15}{3}{25}
\lmcsheading{}{\pageref{LastPage}}{}{}%
{Jul.~07,~2014}{Aug.~30,~2019}{}

\usepackage[utf8]{inputenc}
\usepackage[english]{babel}
\usepackage{pdfsync,hyperref}
\usepackage{verbatim}
\usepackage{amsmath,amsthm,amssymb,bussproofs,tikz,stmaryrd,mathtools}
\usepackage[emu]{cmll}
\usepackage{hyperref,doi}
\usepackage{subfigure}
\usepackage{enumitem}
\usepackage{aliascnt}
\usepackage{xspace}

\usepackage{fullmacros}

\newcommand{\goi}{{\scshape goi}\xspace}
\newcommand{\nat}{{\tt nat}\xspace}
\newcommand{\measurable}[1]{\mathcal{M}(\measure{#1})}


\usepackage{color}

\begin{document}

\title{Interaction Graphs: Exponentials}

\author{Thomas Seiller }
\address{CNRS and University Paris 13}
\email{seiller@lipn.univ-paris13.fr}
\thanks{This work was partially supported by the Marie Sk\l odowska-Curie Individual Fellowship (H2020-MSCA-IF-2014) 659920 - ReACT and the ANR project ANR-10-BLAN-0213 LOGOI.}

\begin{abstract}
This paper is the fourth of a series \cite{seiller-goim, seiller-goiadd, seiller-goig} exposing a systematic combinatorial approach to Girard's Geometry of Interaction (\goi) program \cite{towards}. The \goi program aims at obtaining particular realisability models for linear logic that accounts for the dynamics of cut-elimination. This fourth paper tackles the complex issue of defining exponential connectives in this framework. For that purpose, we use the notion of \emph{graphings}, a generalisation of graphs which was defined in earlier work \cite{seiller-goig}. We explain how to define a \goi for Elementary Linear Logic (\ELL) with second-order quantification, a sub-system of linear logic that captures the class of elementary time computable functions.
\end{abstract}

\maketitle

\renewcommand{\sectionautorefname}{Section}
\renewcommand{\subsectionautorefname}{Section}

\tableofcontents

\section{Introduction}

\subsection{Geometry of Interaction}

A Geometry of Interaction (\goi) model, i.e.\ a construction that fulfils the GoI research program \cite{towards}, is in a first approximation a representation of linear logic proofs that accounts for the dynamics of cut-elimination. A proof is no longer a morphism from $A$ to $B$ -- a function from $A$ into $B$ -- but an operator acting on the space $A\oplus B$. As a consequence, the modus ponens is no longer represented by composition. The operation representing cut-elimination, i.e.\ the obtention of a cut-free proof of $B$ from a cut-free proof of $A$ and a cut-free proof of $A\multimap B$, consists in constructing the solution to an equation called the \emph{feedback equation} (illustrated in \autoref{retroaction}). A \goi model hence represents both the proofs and the dynamics of their normalisation. Contrarily to denotational semantics, a proof $\pi$ and its normalised form $\pi'$ are not represented by the same object. However, those remain related, as the normalisation procedure has a semantical counterpart -- the execution formula $\Ex(\cdot)$ -- which satisfies $\Ex(\pi)=\pi'$. This essential difference between denotational semantics and \goi is illustrated in \autoref{denotgoi}. 

\begin{figure}
\centering
\subfigure[Denotational Semantics]{
\centering
\begin{tikzpicture}[x=1.5cm,y=1.5cm]
	\node (A) at (0,0) {$\pi$};
	\node (B) at (2,0) {$\Int{\pi}{}$};
	\node (C) at (0,-1.5) {$\rho$};
	\node (D) at (2,-1.5) {$\Int{\rho}{}$};
	
	\draw[->] (A) -- (B) node [midway,above] {$\Int{\cdot}{}$};
	\draw[->] (C) -- (D) node [midway,above] {$\Int{\cdot}{}$};`
	\draw[->] (A) -- (C) node [midway,above,sloped] {\small{cut}} node [midway,below,sloped] {\small{elimination}};
	\draw[red,double] (B) -- (D) {};

\end{tikzpicture}
}
\subfigure[Geometry of Interaction]{
\centering
\begin{tikzpicture}[x=1.5cm,y=1.5cm]
	\node (A) at (0,0) {$\pi$};
	\node (B) at (2,0) {$\Int{\pi}{}$};
	\node (C) at (0,-1.5) {$\rho$};
	\node (D) at (2,-1.5) {$\Int{\rho}{}$};
	
	\draw[->] (A) -- (B) node [midway,above] {$\Int{\cdot}{}$};
	\draw[->] (C) -- (D) node [midway,above] {$\Int{\cdot}{}$};`
	\draw[->] (A) -- (C) node [midway,above,sloped] {\small{cut}} node [midway,below,sloped] {\small{elimination}};
	\draw[->,red] (B) -- (D) node [midway,above,sloped] {$\Ex(\cdot)$};

\end{tikzpicture}
}
\caption{Denotational Semantics vs Geometry of Interaction}\label{denotgoi}
\end{figure}
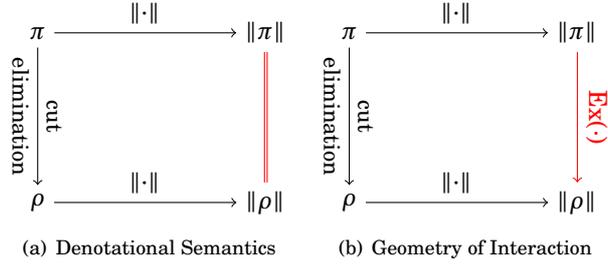

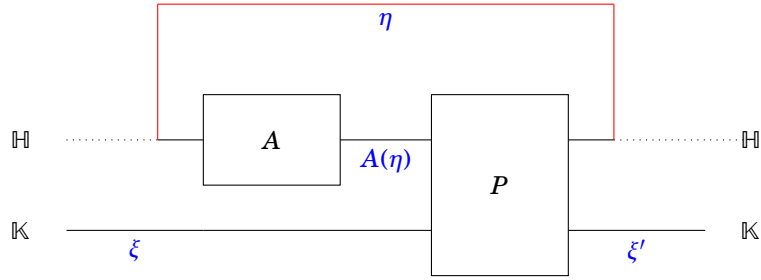
\begin{figure}
\centering
\subfigure[Formal statement]{
\centering
\framebox{
\begin{minipage}{10cm}
$P\in\B{\hil{H\oplus K}}$ represents\footnote{\kern1pt Here, $\hil{H}$ and $\hil{K}$ are separable infinite-dimensional Hilbert spaces, and $\B{\star}$ denotes the set of operators acting on the Hilbert space $\star$: bounded (or, equivalently, continuous) linear maps from $\star$ to $\star$.} a program/proof of implication\\
$A\in\B{\hil{H}}$ represents an argument.\\
$R\in\B{\hil{K}}$ represents the result of the computation if:
$$R(\xi)=\xi'\Leftrightarrow\exists \eta,\eta'\in\hil{H}, \left\{\begin{array}{lcl}P(\eta\oplus \xi)&=&\eta'\oplus \xi'\\A(\eta')&=&\eta\end{array}\right.$$
\vspace{-1em}
\end{minipage}
}
}
\subfigure[Illustration of the equation]{
\centering
\begin{tikzpicture}[x=0.6cm,y=0.6cm]
	\node (aaa) at (1,1.5) {};
	\node (eee) at (1,-5.5) {};
	\draw (1,-2) -- (2,-2) {};
	\draw[dotted] (-1,-2) -- (1,-2) {};
		\node (HHi) at (-2,-2) {$\hil{H}$};
	\draw (-1,-4) -- (2,-4) node [midway,below,blue] {$\xi$};
	\draw (2,-4) -- (7,-4) {};
		\node (HHHi) at (-2,-4) {$\hil{K}$};
	\draw (2,-1) -- (2,-3) -- (5,-3) -- (5,-1) -- (2,-1);
		\node (U) at (3.5,-2) {$A$};
	\draw (7,-1) -- (7,-5) -- (10,-5) -- (10,-1) -- (7,-1);
		\node (V) at (8.5,-3) {$P$};
	\draw (5,-2) -- (7,-2) node [midway,below,blue] {$\eta'$};
	\draw (10,-2) -- (11,-2) {};
	\draw[dotted] (11,-2) -- (14,-2) {};
		\node (HHo) at (14,-2) {$\hil{H}$};
	\draw (10,-4) -- (13,-4) node [midway,below,blue] {$\xi'$};
		\node (HHHo) at (14,-4) {$\hil{K}$};
	\draw[red] (11,-2) -- (11,1) {};
	\draw[red] (11,1) -- (1,1) node [midway,below,blue] {$\eta'$};	
	\draw[red] (1,1) -- (1,-2) {};
\end{tikzpicture}
}
\caption{The Feedback Equation}\label{retroaction}
\end{figure}

The objects under study in a \goi construction are in fact generalisations of (sequent calculus) proof -- paraproofs --, in the same sense the notion of proof structure generalises that of sequent calculus proof\footnote{We recall that all sequent calculus proof gives rise to a proof structure, but not all proof structures are obtained in this way. In particular, those proof structures that are representations of sequent calculus proofs -- called \emph{proof nets} -- are characterised by topological/geometrical properties -- or \emph{correctness criterions}.}. This is reminiscent of game semantics \cite{hylandong,abramsky94full} where not all strategies are interpretations of programs, or Krivine's classical realisability \cite{krivine1,krivine2} where terms containing continuation constants are distinguished from \enquote{proof-like terms}. Let us note here that formal connections between the interpretations of proofs in \goi and game semantics have been studied in the past \cite{baillot-phd}.

Beyond the dynamic interpretation of proofs, the \goi program has a second aim: define by realisability techniques a reconstruction of logical operations from the dynamical model just exposed. This point of view allows a reconstruction of logic as a description of how paraproofs interact; in this aspect \goi differs from game semantics tradition, but seems very close to Krivine's classical realisability constructions (although formal connections are inexistent at this point in time). In some ways a \goi model can be understood as ``discursive syntax'', in which paraproofs are opposed one to another, argue together, each one trying to make its case, to prove the other wrong. Such an argument then terminates if one of them gives up. The discussion itself corresponds to the execution formula, which describes the solution to the feedback equation and generalises the cut-elimination procedure to this generalised notion of proofs. The \enquote{result} of the argument is then described by a notion of \emph{orthogonality}: two paraproofs are orthogonal -- denoted by the symbol $\poll$ -- when this discussion (takes place and) terminates. A notion of formula is then drawn from this notion of orthogonality: a formula is a set of paraproofs $A$ equal to its bi-orthogonal closure $A^{\pol\pol}$ or, equivalently, a set of paraproofs $A=B^{\pol}$ which is the orthogonal to a given set of paraproofs $B$. This realisability construction is an example of Hyland and Schalk's (orthogonality) double-glueing \cite{doubleglueing}.

Drawing some intuitions from the Curry-Howard correspondence, one may propose an alternative reading to this construction in terms of programs. Since proofs correspond to well-behaved programs, paraproofs are a generalisation of those, representing somehow \emph{badly-behaved programs}. If the orthogonality relation represents negation from a logical point of view, it represents a notion of \emph{testing} from a computer science point of view. The notion of formula defined from it corresponds to a notion of type, defined interactively from how (para)programs behave. This point of view is still natural when thinking about programs: a program is of type $\nat\rightarrow\nat$ because it produces a natural number when given a natural number as an argument. On the logical side, this change may be more radical: a proof is a proof of the formula $\Nat\Rightarrow\Nat$ because it produces a proof of $\Nat$ each time it is cut (applied) to a proof of $\Nat$.

Once the notion of type/formula defined, one can reconstruct the connectives: from a ``low-level'' -- between paraproofs -- definition, one obtains a ``high-level'' definition -- between types. For instance, the connective $\otimes$ is first defined between any two paraproofs $\de{a,b}$, and this definition is then extended to types by defining $A\otimes B=\{\de{a}\otimes\de{b}~|~\de{a}\in A,\de{b}\in B\}^{\pol\pol}$. As a consequence, the connectives are not defined in an \emph{ad hoc} way, but their definition is a consequence of their computational meaning: the connectives are defined on proofs/programs and their definition at the level of types is just the reflection of the interaction between the execution -- the dynamics of proofs -- and the low-level definition on paraproofs. Logic thus arises as generated by computation, by the normalisation of proofs: types/formulas are not there to tame the programs/proofs but only to describe their behaviour. This is reminiscent of realisability in the sense that a type is defined as the set of its (para-)proofs. Of course, the fact that we consider a generalised notion of proofs from the beginning has an effect on the construction: contrarily to usual realisability models (except from classical realisability in the sense of Krivine \cite{krivine1,krivine2}), the types $A$ and $A^{\pol}$ (the negation of $A$) are in general both non-empty. This is balanced by the fact that one can define a notion of successful paraproofs, similar to the notion of winning strategy in game semantics. This notion on paraproofs then yields a high-level definition: a formula/type is \emph{true} when it contains a successful paraproof.

\subsection{Interaction Graphs and Graphings}

Throughout the years, a number of works have proposed more combinatorial approached to Girard's operator-based \goi models \cite{RegnierPhD,AspertiDanosLaneveRegnier94,Haghverdi_phd,LaurentTokenGoI,seiller-phd}. It should be noted however that all these approaches except for the last dealt with the proof interpretation aspect of the geometry of interaction program, but did not account for the reconstruction of models of linear logic from the dynamic. The latest of such approaches was introduced by the author under the name ``Interaction Graphs'' \cite{seiller-goim} and defined a combinatorial approach to Girard's geometry of interaction in the hyperfinite factor \cite{goi5}. The initial motivation and the main results of this first paper where that the execution formula -- the counterpart of the cut-elimination procedure -- can be computed as the set of alternating paths between graphs, and that the measurement of interaction defined by Girard using the Fuglede-Kadison determinant \cite{FKdet} can be computed as a measurement of a set of cycles. 

The setting was then extended to deal with additive connectives \cite{seiller-goiadd}, showing by the way that the constructions were a combinatorial approach not only to Girard's hyperfinite \goi construction but also to all the earlier constructions \cite{multiplicatives,goi1,goi2,goi3}. This result could be obtained by unveiling a single geometrical property, which we called the \emph{trefoil property}, upon which all the constructions of geometry of interaction introduced by Girard are built.

In a third paper, we explored a wide generalisation of the graph framework\footnote{This generalisation, or more precisely a fragment of it, already appeared in the author's PhD thesis \cite{seiller-phd}.}. We introduced the notion of \emph{graphing} which we now informally describe. If $(X,\mathcal{B},\mu)$ is a measured space and $\mathfrak{m}$ is a monoid of measurable maps\footnote{For technical reasons, we in fact consider monoids of measurable-preserving non-singular maps \cite{seiller-goig}.} $X\rightarrow X$ (the internal law is composition), then a graphing in $\mathfrak{m}$ is a countable family of restrictions of elements of $\mathfrak{m}$ to measurable subsets. These restrictions of elements of $\mathfrak{m}$ are regarded as edges of a graph \emph{realised} as measurable (partial) maps. We showed that a notion of (alternating) path in a graphing could be defined. As a consequence, one can define the \emph{execution} as the set of alternating paths between graphings, mimicking the corresponding operation of graphs. On the other hand, a more complex argument shows that one can define appropriate measures of cycles in order to insure that the trefoil property holds. From this, we obtained whole hierarchies of models of multiplicative-additive linear logic. The purpose of this paper is to exhibit a family of such models in which one can interpret Elementary Linear Logic \cite{LLL,danosjoinet} with second-order quantification.

It should be noted that since this paper was submitted, the author obtained further results using these techniques, mainly in two directions. On one hand, the author extended the constructions just mentioned in order to obtain models of full linear logic with second-order quantification \cite{seiller-goif}. The latter work however depends on the present paper, and in particular on the careful analysis of the interpretation of types when considering exponential connectives (i.e.\ (co-)perennial, or positive/negative conducts; \autoref{sec_Exp} and \autoref{polarise}).  On the other hand, as envisioned in the previous interaction graphs paper \cite{seiller-goig}, the author obtained first results using graphings-based interaction graphs constructions to characterise a family of non-deterministic sub-linear complexity classes \cite{seiller-goinda}. 

\subsection{Contributions}

The main contribution of this paper in terms of semantics is the definition of the first dynamic semantics of Elementary Linear Logic (\ELL). Although several denotational semantics were defined in the past \cite{} for this fragment of linear logic, the author is not aware of any \emph{dynamic models} being defined, i.e.\ game semantics model or geometry of interaction model. What makes the models defined in this paper particularly interesting is the lack of (explicit\footnote{Due to the natural stratification arising from \ELL syntax, one could probably stratify \emph{a posteriori} the models defined in this paper. However, no a priori stratification is considered in the construction.}) stratification.

It is also important to stress here that the purpose of this work is \emph{not} to obtain a complete model. On the contrary, the specific model described in this paper is an attempt to obtain the \emph{biggest} model of \ELL one can obtain using techniques from interaction graphs. What is shown in fact is that the model of graphings realised by arbitrary affine maps still is a model of \ELL only; recent work \cite{seiller-goif} defining models of full linear logic uses continuous dialect instead of discrete ones, and it is the author belief that this change is necessary to model exponential connectives that are strictly more expressive than the elementary exponentials. Negative results presented in this paper showing that the additional principles  (i.e.\ digging and dereliction) needed to make the model in this paper a model fo full linear logic cannot be interpreted in a satisfactory way (\autoref{nodereliction} and \autoref{nodigging}).

Moreover, the exponential connectives are constructed here in a manner which appears quite different from usual \goi and game semantics definitions. Indeed, while an exponential $\oc a$ is interpreted in geometry of interaction, game semantics, or even denotational semantics, by introducing an infinite number of copies of $a$ through a product with the natural number object $\naturalN$, we here use a product with the unit interval $[0,1]$. While the reader may argue that this also introduces an infinite number of copies of $a$ (and moreover an even greater infinity), let us stress that we are working with measure-theoretic notions. Hence while $\naturalN$ is isomorphic to $\naturalN\times\naturalN$, the space $[0,1]$ is not isomorphic (as a measure space) to $\naturalN\times [0,1]$ although it is isomorphic to products $F\times [0,1]$ where $F$ is a finite set. Therefore, while the former approach can formally be understood as collecting a (countable) infinity of copies of $a$, this is not true of the latter. This difference is essential, and should be understood as a difference between actual infinity and potential infinity: while the product with $\naturalN$ creates an infinite number of copies of $a$, performing the product with $[0,1]$ creates only a potential infinite: one can understand the result as any arbitrary finite number of copies of $a$, but shall never view it as a concrete infinite product\footnote{Of course, one could argue that $[0,1]$ can be expressed as an infinite disjoint union of intervals. Those intervals, however, cannot have all the same measure since the sun of those should converge to $1$.}.

This difference in the definition of exponential connectives provides grounds for another key aspect of the construction, namely its quantitative nature. This aspect appears formally as the fact that all the construction is parametrised by a monoid of weights $\Omega$ and an associated map $m: \Omega\rightarrow \realposN\cup\{\infty\}$. In a recent work \cite{seiller-fock}, the author showed how the monoid $\Omega$ can be formally related to the complete semiring in Laird and al. \cite{quantdenot} weighted relational models, thus allowing to model quantitative aspects. Our models are thus the first quantitative semantics for the fragment \ELL. Moreover, it is known that a crucial step in defining quantitative models is the definition of exponentials, as it usually introduces divergences. Here, divergences are avoided because of the use of measure-theoretic notions: while performing the product with $[0,1]$ (with its usual Lebesgue measure) creates a potential infinity, it does not change the overall measure of the underlying space since $[0,1]$ is of measure $1$.

Lastly, let us mention another contribution of this paper which is specific to geometry of interaction constructions, but important nonetheless. We provide a detailed study of the type constructions induced by exponential connectives, and how those interact with \emph{behaviours} -- types for multiplicative additive linear logic. These results are not related to the specific exponential connectives studied in this paper (and for which we show soundness results), but motivate the definition of the sequent calculi used in this paper. It should be noted that, in particular, the recent construction of models for full linear logic \cite{seiller-goif} makes an essential use of those results.

\subsection{Outline of the paper}

In a first section, we recall some important definitions and properties on directed weighted graphs. This allows us to introduce notation that will be used later on. We then sketch the definition of models of multiplicative-additive linear logic based on graphs \cite{seiller-goiadd}, as well as recall some properties that will be of use in the following sections. These properties are essential to the understanding of the construction of the multiplicative-additive fragment of linear logic in the setting of interaction graphs.

In \autoref{sec_Thick}, we define and study the notion of \emph{thick graphs}, and show how it can be used to interpret the contraction $\cond{\oc A\multimap \oc A\otimes \oc A}$ for some specific formulas $\cond{A}$. This motivates the definition of a \emph{perennisation} $\Omega$ from which one can define an exponential $\cond{A}\mapsto\cond{\oc_{\Omega} A}$. We then explain why it we chose to work with a generalisation of graphs, namely graphings, in order to define perennisations.

In \autoref{sec_Graphings}, we recall the notion of graphing, a generalisation of graphs introduced in an earlier paper \cite{seiller-goig} that allows for the definition of richer models of \mall~ -- allowing in particular the interpretation of second-order quantification. Following the introduction of thick graphs to model contraction as discussed in the previous section, we proceed to study the notion of \emph{thick graphings} that will be used to interpret exponential connectives. We then study the type constructions induced from the consideration of exponential connectives in interaction graphs models. 

Based on this, we introduce and study the specific model of graphings realised by affine maps on the real line (in \autoref{sec_Exp}), and give the definition of an exponential connective defined from a suitable notion of perennisation. We show for this a result which allows us to encode any bijection over the natural numbers as a measure-preserving map over the unit interval of the real line. This result is then used to encode some combinatorics as measure-preserving maps and show that functorial promotion can be implemented for the exponential connective we just defined. We end the section with negative results concerning the interpretation of dereliction and digging in this model.

We then prove a soundness result for a variant (in \autoref{ellbehav}) of Elementary Linear Logic (\ELL) in which one can only write proofs that are somehow ``intuitionnistic''. Indeed, for technical reasons explained later on, the introduction of exponentials (though so-called \emph{perennial conducts}) cannot be performed without being associated to a tensor product. Since the interpretation of elementary time functions in \ELL relies heavily on those proofs that are not intuitionnistic in this sense\footnote{This fact was pointed out to the author by Damiano Mazza.}, this result, though interesting, is not ideal.

Consequently, we introduce (in \autoref{polarise}) a notion of polarities which generalise the notion of \emph{perennial/co-perennial} formulas defined before. The discussion on polarities, which again applies to any interaction graphs models, leads to a refinement of the sequent calculus considered in the previous section which does not suffer from the drawbacks explained above. We then prove a soundness result for the specific model considered in \autoref{sec_Exp} with respect to this calculus.

\section{Interaction Graphs}\label{sec_IG}

\subsection{Basic Definitions}

Departing from the realm of infinite-dimensional vector spaces and linear maps between them, we proposed in previous work \cite{seiller-goim,seiller-goiadd} a graph-theoretical construction of \goi models. We give here a brief overview of the main definitions and results. The graphs we consider are directed and weighted, where the weights are taken in a \emph{weight monoid} $(\Omega,\cdot)$. 

\begin{definition}
A \emph{directed weighted graph} is a tuple $G$, where $V^{G}$ is the set of vertices, $E^{G}$ is the set of edges, $s^{G}$ and $t^{G}$ are two functions from $E^{G}$ to $V^{G}$, the \emph{source} and \emph{target} functions, and $\omega^{G}$ is a function $E^{G} \rightarrow \Omega$.
\end{definition}

In order to ease some later arguments, we first define the \emph{plugging} of two graphs; alternating paths between graphs $F$ and $G$ are then particular paths in the plugging of $F$ and $G$. The plugging is a simple union of graphs together with a colouring function on edges that is used to remember to which graph the edge initially belonged. Here and in the rest of the paper, we denote by $\cup$ the usual set-theoretic union, and by $\disjun$ the disjoint union of sets which, as a categorical coproduct, gives rise to a \emph{copairing} operation $\disjun$ on maps.

\begin{definition}
The \emph{plugging} $F\bicol G$ of two directed weighted graphs $F,G$ is defined as the graph
\[ (V^{F}\cup V^{G},E^{F}\disjun E^{G},s^{F}\disjun s^{G},t^{F}\disjun t^{G},\omega^{F}\disjun \omega^{G}) \] 
together with a colouring function $\delta: E^{F}\disjun E^{G}\rightarrow \{0,1\}$ defined by $\delta(e)=0$ if and only if $e\in E^{F}$.
\end{definition}

The construction is centred around the notion of alternating paths. Given two graphs $F$ and $G$, an alternating path is a path $e_{1}\dots e_{n}$ such that $e_{i}\in E^{F}$ if and only if $e_{i+1}\in E^{G}$. The set of alternating paths will be used to define the interpretation of cut-elimination in the framework, i.e.\ the graph $F\plug G$ -- the \emph{execution of $F$ and $G$} -- is defined as the graph of alternating paths between $F$ and $G$ whose source and target are in the symmetric difference $V^{F}\Delta V^{G}$. The weight of a path is naturally defined as the product of the weights of the edges it contains. 

\begin{definition}
Let $F,G$ be directed weighted graphs. The set of alternating paths between $F$ and $G$ is the set of paths $e_{0},e_{1},\dots,e_{n}$ in $F\bicol G$ such that $\forall i\in\{0,\dots,n-1\}$, $\delta(e_{i})\neq\delta(e_{i+1})$. We write $\enpaths{F,G}$ the set of such paths, and $\enpaths{F,G}_{V}$ the subset of $\enpaths{F,G}$ containing the paths whose source and target lie in $V$.
\end{definition}

\begin{definition}
The execution $F\plug G$ of two directed weighted graphs $F$ and $G$ is the graph defined by:
\begin{eqnarray*}
V^{F\plug G}&=&V^{F}\Delta V^{G}\\
E^{F\plug G}&=&\enpaths{F,G}_{V^{F\plug G}}
\end{eqnarray*}
where the source and target maps are naturally defined, and the weight of a path is the product of the weights of the edges it is composed of.
\end{definition}

As it is usual in mathematics, this notion of paths cannot be considered without the associated notion of cycle: an \emph{alternating cycle} between two graphs $F$ and $G$ is a cycle, i.e.\ a path $e_{0},e_{1},\dots,e_{n}$ such that $s(e_{0})=t(e_{n})$, which is furthermore alternating, i.e.\ $e_{0},e_{1},\dots,e_{n}$ is an alternating path and $\delta(e_{0})\neq\delta(e_{n})$. For technical reasons, we actually consider the notion of circuit, and the related notion of $1$-circuit. A \emph{circuit} is an equivalence class of cycles w.r.t. the action of translations. i.e.\ if $\pi=e_{0}\dots e_{n-1}$ is a cycle, then a translation by an integer $k$ is given as $\sigma(\pi)=e_{0\bar{+}k}\dots e_{n-1\bar{+}k}$ (the symbol $\bar{+}$ denotes addition in $\integerN/n\integerN$), which is itself a cycle. A circuit is the (finite) collection of all such cycles, in some ways a geometric notion of cycle in which the starting edge is not taken into account. A $1$-circuit is then a circuit which is not a proper power of a smaller circuit. 

\begin{definition}
A \emph{$1$-circuit} is an alternating circuit $\pi=e_{0}\dots e_{n-1}$ which is not a proper power of a smaller circuit. In mathematical terms, there do not exists a circuit $\rho$ and an integer $k$ such that $\pi=\rho^{k}$, where the power represents iterated concatenation.
\end{definition}

We denote by $\uncircuits{F,G}$ the set of $1$-circuits in the following. We showed that these notions of paths and cycles satisfy a property we call the \emph{trefoil property} which turned out to be fundamental. The (geometric) trefoil property states that there exists weight-preserving bijections:
\begin{equation*}
\uncircuits{F\plug G,H}\cup\uncircuits{F,G}\cong\uncircuits{G\plug H,F}\cup\uncircuits{G,H}\cong\uncircuits{H\plug F,G}\cup\uncircuits{H,F}
\end{equation*}

\begin{remark}
Intuitively, the trefoil property states the equivalence between three ways of defining the set of all alternating circuits between three graphs $F,G,H$. For instance, the left-hand expression above defines the set of all alternating circuits between $F$ and $G$ together with the set of all alternating circuits between $H$ and the graph of alternating paths between $F$ and $G$. The right-hand expression defines the set of all alternating circuits between $H$ and $F$ together with the alternating circuits between $G$ and the graph of alternating paths between $H$ and $F$. These two expressions are in fact equal to the set of all alternating circuits between the three graphs $F,G,H$ (that can be formally defined through a 3-coloured graph and a similar notion of alternation as in the binary case). One essential ingredient for this property to hold is however the consideration of \emph{circuits} instead of cycles, as the size of a given circuit may vary depending on the way one counts it, i.e.\ a cycle may be of length $2$ in the left-hand expression, but of length $4$ in the right-hand one (this is because alternating paths in the execution are considered as a single edge). Hence while the corresponding sets of cycles are not in bijection, the sets of circuits are.
\end{remark}

Based only on the trefoil property, we were able to define models of the multiplicative fragment of Linear Logic, models furthermore fulfilling the \goi research program. This construction is moreover parametrised by a map from the set $\Omega$ to $\mathbb{R}_{\geqslant 0}\cup\{\infty\}$, and therefore yields not only one but a whole family of models. This parameter is introduced to define the notion of orthogonality in our setting, a notion that account for linear negation. Indeed, given a tracial map\footnote{We will call those \emph{circuit-quantifying maps} in this paper. Although this notion is quite simple in the case of graphs and may not require a specific terminology, the notion of circuit-quantifying maps for graphings is quite involved and we opted for an homogeneous presentation.} $m: \Omega\rightarrow \mathbb{R}_{\geqslant 0}\cup\{\infty\}$  (i.e.\ $m(ab)=m(ba)$) and two graphs $F,G$ we define $\meas{F,G}$ as the sum $\sum_{\pi\in\uncircuits{F,G}} m(\omega(\pi))$, where $\omega(\pi)$ is the weight of the cycle $\pi$. The orthogonality is then constructed from this measurement, and the geometric trefoil property gives rise to the \emph{numerical} trefoil property:
\[ \meas{F\plug G,H}+\meas{F,G} = \meas{G\plug H,F}+\meas{G,H} = \meas{H\plug F,G}+\meas{H,F}  \]

We moreover showed how, from any of these constructions, one can obtain a $\ast$-autonomous category \catmll{} with $\parr\not\cong\otimes$ and $1\not\cong\bot$, i.e.\ a non-degenerate denotational semantics for Multiplicative Linear Logic (\MLL). This model can then be extended to deal with additive connectives by considering a generalisation of graphs named \emph{sliced graphs}. 

\begin{definition}
A \emph{sliced (directed weighted) graph} $A$ of carrier $V^{A}$ is a finite formal weighted sum $\sum_{i\in I^{A}} \alpha^{A}_{i} A_{i}$, where each $\alpha_{i}^{A}$ is a real number, and $A_{i}$ is a graph such that $V^{A_{i}}=V^{A}$. We define $\unit{A}=\sum_{i\in I^{A}}\alpha^{A}_{i}$. We will refer to the subgraphs $A_{i}$ as \emph{slices} following terminology from additive proof nets \cite{proofnets}, and we identify a graph $A$ with the sliced graph $1.A$.
\end{definition}

\noindent The measurement defined above on graphs is extended as follows to sliced graphs:
\[ \meas{A,B}=\sum_{i\in I^{A}}\sum_{j\in I^{B}} \alpha_{i}^{A}\alpha^{B}_{j}\meas{A_{i},B_{j}} \]
The corresponding numerical trefoil property is then expressed as follows:
\[ \meas{F\plug G,H}+\unit{H}\meas{F,G} = \meas{G\plug H,F}+\unit{F}\meas{G,H} = \meas{H\plug F,G}+\unit{G}\meas{H,F}  \]

The formal sum is then used to define additive connectives. However, as in all the versions of \goi dealing with those connectives, our construction of additives does not define a categorical product. We solve this issue by introducing a notion of \emph{observational equivalence} within the model. We are then able to define a categorical product from our additive connectives when considering classes of observationally equivalent objects, thus obtaining a denotational semantics for Multiplicative Additive Linear Logic (\MALL).

\subsection{Models of \MALL in a Nutshell}\label{recallmall}

We recall the basic definitions of projects, and behaviours, which will be respectively used to interpret proofs and formulas, as well as the definition of connectives.

\begin{itemize}[noitemsep,nolistsep]
\item a \emph{project} of carrier $V^{A}$ is a triple $\de{a}=(a,V^{A},A)$, where $a$ is a real number, and $A$ is a sliced graph of carrier $V^{A}$; 
\item for all set $V$, we denote by $\de{0}_{V}$ the project $(0,V,\emptyset_{V})$ where $\emptyset_{V}$ denotes the empty graph of support $V$;
\item given a project $(a,A)$ and a real number $\lambda$, we define the project $\lambda\de{a}$ as the project $(\lambda a,\lambda A)$ where $\lambda A$ stands for $\sum_{i\in I^{A}}\lambda \alpha_{i}^{A}A_{i}$;
\item given two projects $(a,A)$ and $(b,B)$ of equal support $V^{A}$, we define their sum $\de{a+b}$ as the project $(a+b,A+B)$ (the sum $A+B$ of formal sums is defined in the obvious way);
\item two projects $\de{a,b}$ are \emph{orthogonal}, noted $\de{a}\poll\de{b}$, when:
$$\sca{a}{b}=a\unit{B}+b\unit{A}+\meas{A,B}\neq 0,\infty$$
\item the orthogonal of a set of project $E$ is defined as $E^{\pol}=\{\de{a}~|~\forall \de{e}\in E, \de{a}\poll\de{e}\}$; we write $E^{\pol\pol}=(E^{\pol})^{\pol}$ the \emph{bi-orthogonal} closure of $E$;
\item the \emph{execution} of two projects $\de{a,b}$ is defined as:
$$\de{a\plug b}=(\sca{a}{b},V^{A}\Delta V^{B},\sum_{i\in I^{A}}\sum_{j\in I^{B}} \alpha^{A}_{i}\alpha^{B}_{j} A_{i}\plug B_{j})$$
\item if $\de{a}$ is a project and $V$ is a measurable set such that $V^{A}\subset V$, we define the extension $\de{a}_{\uparrow V}$ as the project $(a,V,A)$;
\item a \emph{conduct} $\cond{A}$ of carrier $V^{A}$ is a set of projects of carrier $V^{A}$ which is equal to its bi-orthogonal $\cond{A}^{\pol\pol}$;
\item a \emph{behaviour} $\cond{A}$ of carrier $V^{A}$ is a conduct such that for all $\lambda\in\realN$, 
$$\begin{array}{rcl}\de{a}\in\cond{A} &\Rightarrow& \de{a+\lambda 0}_{V^{A}}\in\cond{A}\\
\de{b}\in\cond{A}^{\pol} &\Rightarrow& \de{b+\lambda 0}_{V^{A}}\in\cond{A}^{\pol}\end{array}$$
\item we define, for every measurable set the \emph{empty} behaviour of carrier $V$ as the empty set $\cond{0}_{V}$, and the \emph{full behaviour} of carrier $V$ as its orthogonal $\cond{T}_{V}=\{\de{a}~|~\de{a}\text{ of support }V\}$;
\item if $\cond{A,B}$ are two behaviours of disjoint carriers, we define:
\begin{eqnarray*}
\cond{A \otimes B}\footnotemark&=&\{\de{a\plug b}~|~\de{a}\in\cond{A},\de{b}\in\cond{B}\}^{\pol\pol}\\
\cond{A \multimap B}&=&\{\de{f}~|~\forall\de{a}\in\cond{A}, \de{f\plug a}\in\cond{B}\}\\
\cond{A \oplus B}&=&(\{\de{a}_{\uparrow V^{A}\cup V^{B}}~|~\de{a}\in\cond{A}\}^{\pol\pol}\cup \{\de{b}_{\uparrow V^{A}\cup V^{B}}~|~\de{b}\in\cond{B}\}^{\pol\pol})^{\pol\pol}\\
\cond{A \with B}&=&\{\de{a}_{\uparrow V^{A}\cup V^{B}}~|~\de{a}\in\cond{A^{\pol}}\}^{\pol}\cap \{\de{b}_{\uparrow V^{A}\cup V^{B}}~|~\de{b}\in\cond{B}^{\pol}\}^{\pol}
\end{eqnarray*}
\footnotetext{Notice that the point-wise operation in the tensor product is shown as $\de{a\plug b}$ but since we are in the specific case where $\de{a}$ and $\de{b}$ have disjoint supports, the project $\de{a\plug b}$ is equal to a \enquote{slice-wise disjoint sum}: $\de{a\otimes b}=(a+b,\sum_{i\in I^{A}}\sum_{j\in I^{B}} \alpha_{i}^{A}\alpha^{B}_{j}A_{i}\disjun B_{j})$. We thus recover the usual definition of the tensor product of Geometry of Interaction. For practical purposes, and to avoid proving twice the same results -- for $\plug$ and for $\otimes$ --, it is however useful to work with this definition and consider the operation $\otimes$ as a notation for $\plug$ in the case the projects are of disjoint supports.}
\item two elements $\de{a,b}$ of a conduct $\cond{A}$ are \emph{observationally equivalent} when:
\begin{equation*}
\forall \de{c}\in\cond{A}^{\pol},~\sca{a}{c}=\sca{b}{c}
\end{equation*} 
\end{itemize}

One important point in this work is the fact that all results rely on a single geometric property, namely the previously introduced \emph{trefoil property} which describes how the sets of $1$-circuits evolve during an execution. This property insures on its own the four following facts:
\begin{itemize}[noitemsep,nolistsep]
\item we obtain a $\ast$-autonomous category \catmll{} whose objects are conducts and morphisms are projects;
\item the observational equivalence is a congruence on this category;
\item the quotient category \concat{} inherits the $\ast$-autonomous structure;
\item the quotient category \concat{} has a full subcategory \behcat{} with products whose objects are behaviours.
\end{itemize}
This can be summarised in the following two theorems.

\begin{theorem}
For any map $m:\Omega\rightarrow\realN\cup\{\infty\}$, the categories \concat and \catmll{} are non-degenerate categorical models of Multiplicative Linear Logic with multiplicative units.
\end{theorem}

\begin{theorem}
For any map $m:\Omega\rightarrow\realN\cup\{\infty\}$, the full subcategory \behcat{} of \concat{} is a non-degenerate categorical model of Multiplicative-Additive Linear Logic with additive units.
\end{theorem}

The categorical model we obtain has two layers (see \autoref{catmodels}). The first layer consists in a non-degenerate (i.e.\ $\otimes\neq\parr$ and $\cond{1}\neq\cond{\bot}$) $\ast$-autonomous category \concat{}\!\!, hence a denotational model for \MLL with units. The second layer is the full subcategory \behcat which does not contain the multiplicative units but is a non-degenerate model (i.e.\ $\otimes\neq\parr$, $\oplus\neq\with$ and $\cond{0}\neq\cond{\top}$) of \MALL with additive units that does not satisfy the mix and weakening rules.

\begin{figure}
\centering
\begin{tikzpicture}[x=1.2cm,y=0.25cm]
	\draw[fill,opacity=0.1] (0,0) .. controls  (0,7.5) and (0.5,8) .. (5,8) .. controls (9.5,8) and (10,7.5) .. (10,0) .. controls (10,-7.5) and (9.5,-8) .. (5,-8) .. controls (0.5,-8) and (0,-7.5) .. (0,0) ;
		\node (A) at (2,0) {\begin{tabular}{c}\small{\concat}\\\small{($\ast$-autonomous)}\end{tabular}};
	\draw[fill,opacity=0.2]  (3.5,0) .. controls (3.5,5.5) and (4,6) .. (5.5,6) .. controls (8,6) and (8.5,5.5) .. (8.5,0) .. controls (8.5,-5.5) and (8,-6) .. (7,-6) .. controls (4,-6) and (3.5,-5.5) .. (3.5,0);
		\node (B) at (6,0) {\begin{tabular}{c}\small{\behcat}\\\small{(closed under $\otimes,\multimap,\with,\oplus,(\cdot)^{\pol}$)}\\[2em]\small{NO weakening, NO mix}\end{tabular}};
	\node (bot) at (1.5,-4) {$\bullet_{\bot}$};
	\node (one) at (2.5,-4) {$\bullet_{\cond{1}}$};
	
	\node (top) at (5.5,-1) {$\bullet_{\cond{T}}$};
	\node (zero) at (6.5,-1) {$\bullet_{\cond{0}}$};
\end{tikzpicture}
\caption{The categorical models}\label{catmodels}
\end{figure}
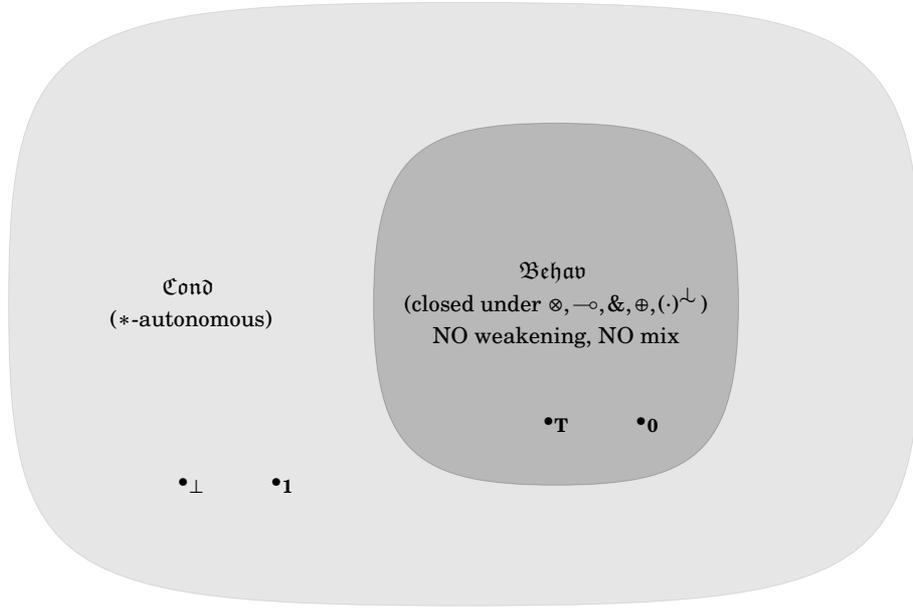

We here recall some technical results obtained in our paper on additives \cite{seiller-goiadd} and that will be useful in the remaining sections.

\begin{proposition}\label{fellorth}
If $A$ is a non-empty set of projects of same carrier $V^{A}$ such that $(a,A)\in A$ implies $a=0$, then $\de{b}\in A^{\pol}$ implies $\de{b}+\lambda\de{0}_{V^{A}}\in A^{\pol}$ for all $\lambda\in\mathbb{R}$.
\end{proposition}

\begin{proposition}\label{wagerfreeorth}
If $A$ is a non-empty set of projects of carrier $V$ such that $\de{a}\in A\Rightarrow \de{a+\lambda 0}_{V}\in A$, then any project in $A^{\pol}$ is wager-free, i.e.\ if $(a,A)\in A^{\pol}$ then $a=0$.
\end{proposition}

\begin{lemma}[Homothety]\label{homothetie}
Conducts are closed under homothety: for all $\de{a}\in\cond{A}$ and all $\lambda\in\mathbf{R}$ with $\lambda\neq 0$, $\lambda\de{a}\in\cond{A}$.
\end{lemma}

\begin{proposition}\label{ethtenscondtens}
We denote by $\cond{A\odot B}$ the set $\{\de{a}\otimes\de{b}~|~\de{a}\in\cond{A},\de{b}\in\cond{B}\}$. Let $E,F$ be non-empty sets of projects of respective carriers $V,W$ with $V\cap W=\emptyset$. Then 
\begin{equation*}
(E\odot F)^{\pol\pol}=(E^{\pol\pol}\odot F^{\pol\pol})^{\pol\pol}
\end{equation*}
\end{proposition}

\begin{proposition}\label{compintoplus}
Let $\cond{A,B}$ be conducts. Then:
$$(\{\de{a}\otimes \de{0}_{\mathnormal{B}}~|~\de{a}\in\cond{A}\}\cup\{\de{0}_{\mathnormal{A}}\otimes\de{b}~|~\de{b}\in\cond{B}\})^{\pol\pol}=\cond{A\oplus B}$$
\end{proposition}

\begin{proposition}[Distributivity]\label{distributivity}
For any behaviours $\cond{A,B,C}$, and delocations\footnote{Let us recall that a \emph{delocation} \cite[Definition ??]{} $\phi((a,A))$ of a project $(a,A)$ is simply given by a bijection $\phi$ from the support of $V^{A}$ of $A$ onto a support $V$ such that $V\cap V^{A}=\emptyset$; $\phi((a,A))$ is then defined as $(a,\phi(A))$ with $\phi(A)$ the graph obtained from $A$ by renaming its vertices according to the bijection $\phi$. This operation naturally extends to conducts by letting $\phi(\cond{A})=\{\phi(\de{a})~|~\de{a}\in\cond{A}\}$. Let us stress here that when considering several delocations -- such as in the statement of \autoref{compintoplus} --, we always suppose that all involved supports are pairwise disjoint.} $\phi,\psi,\theta,\rho$ of $\cond{A},\cond{A},\cond{B},\cond{C}$ respectively, there is a project $\de{distr}$ in the behaviour 
$$\cond{((\phi(A)\!\multimap\! \theta(B))\!\with\! (\psi(A)\!\multimap\! \rho(C)))\!\multimap\! (A\!\multimap\! (B\!\with\! C))}$$
\end{proposition}

\section{Thick Graphs and Contraction}\label{seccontraction}\label{sec_Thick}

In this section, we will define the notion of \emph{thick graphs}, and extend the additive construction defined in our earlier paper \cite{seiller-goiadd} to that setting. The introduction of these objects will be motivated in \autoref{subsec:contraction}, where we will explain how thick graphs allows for the interpretation of the contraction rule. This contraction rule being satisfied only for a certain kind of conducts -- interpretations of formulas, this will justify the definition of the exponentials.

\subsection{Thick Graphs}\label{graphesepais}

\begin{definition}
Let $S^{G}$ and $D^{G}$ be finite sets. A directed weighted \emph{thick graph} $G$ of carrier $S^{G}$ and \emph{dialect} $D^{G}$ is a directed weighted graph over the set of vertices $S^{G}\times D^{G}$.

Following the terminology introduced for sliced graphs, we will call \emph{slices} the set of vertices $S^{G}\times\{d\}$ for $d\in D^{G}$.
\end{definition}

\begin{remark}
When no further precision is given, a thick graph $G$ will always be considered with carrier $S^{G}$, dialect $D^{G}$, with set of vertices $V^{G}$, set of edges $E^{G}$ and associated maps $s^{G}$, $t^{G}$, and $\omega^{G}$.
\end{remark}

\autoref{graphesepaisex} shows two examples of thick graphs. Thick graphs will be represented following a graphical convention very close to the one we used for sliced graphs:
\begin{itemize}[noitemsep,nolistsep]
\item Graphs are once again represented with coloured edges and delimited by hashed lines;
\item Elements of the carrier $S^{G}$ are represented on a horizontal scale, while elements of the dialect $D^{G}$ are represented on a vertical scale;
\item Inside a given graph, slices are separated by a \emph{dotted} line.
\end{itemize}

\begin{figure}
\begin{center}
\begin{tikzpicture}[node distance=2cm]
	\node (G11) at (0,0) {$1_{1}$};
	\node (G21) at (2,0) {$2_{1}$};

	\node (G12) at (0,2) {$1_{2}$};
	\node (G22) at (2,2) {$2_{2}$};

	\node (H11) at (6,0) {$2_{1}$};
	\node (H21) at (8,0) {$3_{1}$};

	\node (H12) at (6,2) {$2_{2}$};
	\node (H22) at (8,2) {$3_{2}$};

	\draw[<->,blue] (G11) .. controls (-1,0.5) and (-1,1.5) .. (G12) {};
	\draw[<->,blue] (G12) .. controls (0.5,3) and (1.5,3) .. (G22) {};
	\draw[<->,blue] (G11) .. controls (0.5,-1) and (1.5,-1) .. (G21) {};
	
	\draw[dashed,blue] (-1.3,-1) -- (-1.3,3) node [sloped,above,near end] {slice $2$} node [sloped,above,near start] {slice $1$};
	\draw[dashed,blue] (-1.3,-1) -- (3,-1) {};
	\draw[dashed,blue] (-1.3,3) -- (3,3) {};
	\draw[dashed,blue] (3,-1) -- (3,3) {};
	\draw[dotted,blue] (-1.3,1) -- (3,1) {};
	\node (G) at (-1.1,2.8) {\textcolor{blue}{G}};
	
	\draw[<->,red] (H12) -- (H21) {};
	\draw[->,red] (H22) .. controls (7.5,3) and (8.5,3) .. (H22) {};
	\draw[->,red] (H11) .. controls (5.5,-1) and (6.5,-1) .. (H11) {};
	
	\draw[-,dashed,red] (4.85,-1) -- (4.85,3) node [sloped,above,near end] {slice $2$} node [sloped,above,near start] {slice $1$};
	\draw[dashed,red] (4.85,-1) -- (9.15,-1) {};
	\draw[dashed,red] (4.85,3) -- (9.15,3) {};
	\draw[dashed,red] (9.15,-1) -- (9.15,3) {};
	\draw[dotted,red] (4.85,1) -- (9.15,1) {};
	\node (G) at (8.8,2.8) {\textcolor{red}{H}};

\end{tikzpicture}
\end{center}
\caption{Two thick graphs $G$ and $H$, both with dialect $\{1,2\}$}\label{graphesepaisex}
\end{figure}

\begin{remark}
If $G=\sum_{i\in I^{G}} \alpha^{G}_{i}G_{i}$ is a sliced graph such that $\forall i\in I^{G}, \alpha^{G}_{i}=1$, then $G$ can be identified with a thick graph of dialect $I^{G}$. Indeed, one can define the thick graph $\{G\}$ by:
\begin{eqnarray*}
V^{\{G\}}&=&V^{G}\times I^{G}\\
E^{\{G\}}&=&\disjun_{i\in I^{G}} E^{G_{i}}\\
s^{\{G\}}&=&e\in E^{G_{i}}\mapsto (s^{G_{i}}(e),i)\\
t^{\{G\}}&=&e\in E^{G_{i}}\mapsto (t^{G_{i}}(e),i)\\
\omega^{\{G\}}&=&e\in E^{G_{i}}\mapsto \omega^{G_{i}}(e)
\end{eqnarray*}
\end{remark}

\begin{definition}[Variants]
Let $G$ be a thick graph and $\phi: D^{G}\rightarrow E$ a bijection. One defines $G^{\phi}$ as the graph:
\begin{eqnarray*}
V^{G^{\phi}}&=&S^{G}\times E\\
E^{G^{\phi}}&=&E^{G}\\
s^{G^{\phi}}&=&(Id_{V^{G}}\times\phi)\circ s^{G}\\
t^{G^{\phi}}&=&(Id_{V^{G}}\times\phi)\circ t^{G}\\
\omega^{G^{\phi}}&=&\omega^{G}
\end{eqnarray*}
If $G$ and $H$ are two thick graphs such that $H=G^{\phi}$ for a bijection $\phi$, then $H$ is called a \emph{variant} of $G$. The relation defined by $G\sim H$ if and only if $G$ is a variant of $H$ can easily be checked to be an equivalence relation.
\end{definition}

\begin{definition}[Dialectal Interaction]
Let $G$ and $H$ be thick graphs. 
\begin{enumerate}
\item We denote by $G^{\dagger_{D^{H}}}$ the thick graph of dialect $D^{G}\times D^{H}$ defined as $\{\sum_{i\in D^{H}} G\}$;
\item We denote by $H^{\ddagger_{D^{G}}}$ the thick graph of dialect $D^{G}\times D^{H}$ defined as $\{\sum_{i\in D^{G}} H\}^{\tau}$ where $\tau$ is the natural bijection $D^{H}\times D^{G}\rightarrow D^{G}\times D^{H}, (a,b)\mapsto(b,a)$.
\end{enumerate}
\end{definition}

\begin{figure}
\begin{center}
\begin{tikzpicture}[scale=0.9]
	\node (G111) at (0,0,0) {$1_{1,1}$};
	\node (G211) at (2,0,0) {$2_{1,1}$};
	\node (G121) at (0,2,0) {$1_{2,1}$};
	\node (G221) at (2,2,0) {$2_{2,1}$};
	\node[opacity=0.6] (G112) at (0,0,-4) {$1_{1,2}$};
	\node[opacity=0.6] (G212) at (2,0,-4) {$2_{1,2}$};
	\node[opacity=0.6] (G122) at (0,2,-4) {$1_{2,2}$};
	\node[opacity=0.6] (G222) at (2,2,-4) {$2_{2,2}$};

	\draw[<->,blue] (G111) .. controls (-1,0.5,0) and (-1,1.5,0) .. (G121) {};
	\draw[<->,blue] (G121) .. controls (0.5,3,0) and (1.5,3,0) .. (G221) {};
	\draw[<->,blue] (G111) .. controls (0.5,-1,0) and (1.5,-1,0) .. (G211) {};
	\draw[<->,blue,dashed] (G112) .. controls (-1,0.5,-4) and (-1,1.5,-4) .. (G122) {};
	\draw[<->,blue,dashed] (G122)  .. controls (0.5,3,-4) and (1.5,3,-4) .. (G222) {};
	\draw[<->,blue,dashed] (G112) .. controls (0.5,-1,-4) and (1.5,-1,-4) .. (G212) {};

	\draw[->,dotted] (0,0,0) -- (5,0,0)  {};
	\draw[->,dotted] (0,0,0) -- (0,5,0)  {};
	\draw[->,dotted] (0,0,0) -- (0,0,-12)  {};
	
	\draw[fill=blue,draw=red,opacity=.2,very thin,line join=round]
 		(3,3,-0) --
		(-1,3,0) --
		(-1,-1,0) --
		(3,-1,0) --
		(3,3,0) node [sloped,above,midway,opacity=0.5] {slices $(\cdot,1)$} ;
	
	\draw[fill=blue,draw=red,opacity=.2,very thin,line join=round]
		(3,3,-4) --
		(-1,3,-4) --
		(-1,-1,-4) --
		(3,-1,-4) --
		(3,3,-4) node [sloped,above,midway,opacity=0.5] {slices $(\cdot,2)$} ;
	
	\node (H111) at (7,0,0) {$2_{1,1}$};
	\node (H211) at (9,0,0) {$3_{1,1}$};
	\node[opacity=0.6] (H112) at (7,0,-4) {$2_{1,2}$};
	\node[opacity=0.6] (H212) at (9,0,-4) {$3_{1,2}$};
	\node (H121) at (7,2,0) {$2_{2,1}$};
	\node (H221) at (9,2,0) {$3_{2,1}$};
	\node[opacity=0.6] (H122) at (7,2,-4) {$2_{2,2}$};
	\node[opacity=0.6] (H222) at (9,2,-4) {$3_{2,2}$};

	\draw[<-,red] (H211) -- (8,0,-2) {};
	\draw[->,red,densely dashed] (8,0,-2) -- (H112) {};
	\draw[->,red,dashed] (H212) .. controls (8.5,0,-6) and (9.5,0,-6) .. (H212) {};
	\draw[->,red] (H111) .. controls (6.5,0,-2) and (7.5,0,-2) .. (H111) {};
	\draw[<-,red] (H221) -- (8,2,-2) {};
	\draw[->,red,densely dashed] (8,2,-2) -- (H122) {};
	\draw[->,red,dashed] (H222) .. controls (8.5,2,-6) and (9.5,2,-6) .. (H222) {};
	\draw[->,red] (H121) .. controls (6.5,2,-2) and (7.5,2,-2) .. (H121) {};
	
	\draw[->,dotted] (7,0,0) -- (12,0,0)  {};
	\draw[->,dotted] (7,0,0) -- (7,5,0)  {};
	\draw[->,dotted] (7,0,0) -- (7,0,-12)  {};

	\draw[fill=red,draw=red,opacity=.2,very thin,line join=round]
		(10,0,-6) --
		(6,0,-6) --
		(6,0,2) --
		(10,0,2) --
		(10,0,-6) node [sloped,above,near start,opacity=0.5] {slices $(1,\cdot	)$} ;
	
	\draw[fill=red,draw=red,opacity=.2,very thin,line join=round]
		(10,2,-6) --
		(6,2,-6) --
		(6,2,2) --
		(10,2,2) --
		(10,2,-6) node [sloped,above,near start,opacity=0.5] {slices $(2,\cdot	)$} ;

\end{tikzpicture}
\end{center}
\caption{The graphs $G^{\dagger_{D^{H}}}$ and $H^{\ddagger_{D^{G}}}$}\label{interactiondialectale}
\end{figure}

We can then define the plugging $F\bicol G$ of two thick graphs as the plugging of the graphs $F^{\dagger_{D^{G}}}$ and $G^{\ddagger_{D^{F}}}$. \autoref{epaisbranchement} shows the result of the plugging of $G$ and $H$, the thick graphs represented in \autoref{graphesepaisex}.

\begin{remark}
For the purpose of this section, we will now introduce separated notations for execution of thick graphs w.r.t. the notion of execution for (non-thick) graphs. This is because many properties of the thick graphs execution can be deduced from computations with the (non-thick) graphs execution. In the following sections, however, we will use the generic notation $\plug$, as the collapse of notations may not induce any misunderstanding.
\end{remark}

\begin{figure}
\begin{center}
\begin{tikzpicture}
	\node (G111) at (0,0,0) {$1_{1,1}$};
	\node (G211) at (2,0,0) {$2_{1,1}$};
	\node (G121) at (0,2,0) {$1_{2,1}$};
	\node (G221) at (2,2,0) {$2_{2,1}$};
	\node[opacity=0.6] (G112) at (0,0,-4) {$1_{1,2}$};
	\node[opacity=0.6] (G212) at (2,0,-4) {$2_{1,2}$};
	\node[opacity=0.6] (G122) at (0,2,-4) {$1_{2,2}$};
	\node[opacity=0.6] (G222) at (2,2,-4) {$2_{2,2}$};

	\draw[<->,blue] (G111) .. controls (-1,0.5,0) and (-1,1.5,0) .. (G121) {};
	\draw[<->,blue] (G121) .. controls (0.5,3,0) and (1.5,3,0) .. (G221) {};
	\draw[<->,blue] (G111) .. controls (0.5,-1,0) and (1.5,-1,0) .. (G211) {};
	\draw[<->,blue,dashed] (G112) .. controls (-1,0.5,-4) and (-1,1.5,-4) .. (G122) {};
	\draw[<->,blue,dashed] (G122)  .. controls (0.5,3,-4) and (1.5,3,-4) .. (G222) {};
	\draw[<->,blue,dashed] (G112) .. controls (0.5,-1,-4) and (1.5,-1,-4) .. (G212) {};

	\draw[->,dotted] (0,0,0) -- (7,0,0) {};
	\draw[->,dotted] (0,0,0) -- (0,5,0) {};
	\draw[->,dotted] (0,0,0) -- (0,0,-12) {};
	
	\draw[fill=blue,draw=red,opacity=.2,very thin,line join=round]
 		(0,0,0) -- 
 		(2,0,0) --
		(2,2,0) --
		(0,2,0) --
		(0,0,0) {} ;
	
	\draw[fill=blue,draw=red,opacity=.2,very thin,line join=round]
 		(0,0,-4) -- 
 		(2,0,-4) --
		(2,2,-4) --
		(0,2,-4) --
		(0,0,-4) {} ;
	
	\node (H211) at (4,0,0) {$3_{1,1}$};
	\node[opacity=0.6] (H212) at (4,0,-4) {$3_{1,2}$};
	\node (H221) at (4,2,0) {$3_{2,1}$};
	\node[opacity=0.6] (H222) at (4,2,-4) {$3_{2,2}$};

	\draw[<-,red] (H211) -- (3,0,-2) {};
	\draw[->,red,densely dashed] (3,0,-2) -- (G212) {};
	\draw[->,red,dashed] (H212) .. controls (3.5,0,-6) and (4.5,0,-6) .. (H212) {};
	\draw[->,red] (G211) .. controls (1.5,0,-2) and (2.5,0,-2) .. (G211) {};
	\draw[<-,red] (H221) -- (3,2,-2) {};
	\draw[->,red,densely dashed] (3,2,-2) -- (G222) {};
	\draw[->,red,dashed] (H222) .. controls (3.5,2,-6) and (4.5,2,-6) .. (H222) {};
	\draw[->,red] (G221) .. controls (1.5,2,-2) and (2.5,2,-2) .. (G221) {};
	
	\draw[fill=red,draw=red,opacity=.2,very thin,line join=round]
 		(2,0,0) -- 
 		(4,0,0) --
		(4,0,-4) --
		(2,0,-4) --
		(2,0,0) {} ;
	
	\draw[fill=red,draw=red,opacity=.2,very thin,line join=round]
 		(2,2,0) -- 
 		(4,2,0) --
		(4,2,-4) --
		(2,2,-4) --
		(2,2,0) {} ;
		
\end{tikzpicture}
\end{center}
\caption{Plugging of the thick graphs $G$ and $H$}\label{epaisbranchement}
\end{figure}	

One can then define the execution $G\plugepais H$ of two thick graphs $G$ and $H$ as the execution of the graphs $G^{\dagger_{D^{H}}}$ and $H^{\ddagger_{D^{G}}}$. \autoref{cheminsaltepais} shows the set of alternating paths in the plugging of the thick graphs $G$ and $H$ introduced in \autoref{graphesepaisex}. \autoref{execepais3d} and \autoref{execepais2d} represent the result of the execution of these two thick graphs, the first is three-dimensional representation which can help make the connection with the set of alternating paths in \autoref{cheminsaltepais}, while the second is a two-dimensional representation of the same graph. In a natural way, the measurement of the interaction between two thick graphs $G,H$ is defined as $\meas{G^{\dagger_{D^{H}}},H^{\ddagger_{D^{G}}}}$.

\begin{definition}
The execution $F\plugepais G$ of two thick graphs $F,G$ is the thick graph of carrier $S^{F}\Delta S^{G}$ and dialect $D^{F}\times D^{G}$ defined as $F^{\dagger_{D^{G}}}\plug G^{\ddagger_{D^{F}}}$.
\end{definition}

\begin{figure}
\begin{center}
\begin{tikzpicture}
	\node (G111) at (0,0,0) {$1_{1,1}$};
	\node (G121) at (0,2,0) {$1_{2,1}$};
	\node[opacity=0.6] (G112) at (0,0,-4) {$1_{1,2}$};
	\node[opacity=0.6] (G122) at (0,2,-4) {$1_{2,2}$};

	\draw[<->,blue] (G111) .. controls (-1,0.5,0) and (-1,1.5,0) .. (G121) {};
	\draw[->,blue] (G121) .. controls (0.5,3,0) and (1.5,3,0) .. (1.8,2,0) {};
	\draw[->,blue] (G111) .. controls (0.5,-1,0) and (1.5,-1,0) .. (1.8,0,0) {};
	\draw[<-,blue] (G121) .. controls (0.5,3,0) and (1.5,3,0) .. (2.2,2,0) {};
	\draw[<-,blue] (G111) .. controls (0.5,-1,0) and (1.5,-1,0) .. (2.2,0,0) {};
	\draw[<->,blue,dashed] (G112) .. controls (-1,0.5,-4) and (-1,1.5,-4) .. (G122) {};
	\draw[->,blue,dashed] (G122)  .. controls (0.5,3,-4) and (1.5,3,-4) .. (1.8,2,-4) {};
	\draw[->,blue,dashed] (G112) .. controls (0.5,-1,-4) and (1.5,-1,-4) .. (1.8,0,-4) {};
	\draw[<-,blue,dashed] (G122)  .. controls (0.5,3,-4) and (1.5,3,-4) .. (2.2,2,-4) {};
	\draw[<-,blue,dashed] (G112) .. controls (0.5,-1,-4) and (1.5,-1,-4) .. (2.2,0,-4) {};

	\draw[->,dotted] (0,0,0) -- (7,0,0) {};
	\draw[->,dotted] (0,0,0) -- (0,5,0) {};
	\draw[->,dotted] (0,0,0) -- (0,0,-12) {};
	
	\draw[fill=blue,draw=red,opacity=.2,very thin,line join=round]
 		(0,0,0) -- 
 		(2,0,0) --
		(2,2,0) --
		(0,2,0) --
		(0,0,0) {} ;
	
	\draw[fill=blue,draw=red,opacity=.2,very thin,line join=round]
 		(0,0,-4) -- 
 		(2,0,-4) --
		(2,2,-4) --
		(0,2,-4) --
		(0,0,-4) {} ;
	
	\node (H211) at (4,0,0) {$3_{1,1}$};
	\node[opacity=0.6] (H212) at (4,0,-4) {$3_{1,2}$};
	\node (H221) at (4,2,0) {$3_{2,1}$};
	\node[opacity=0.6] (H222) at (4,2,-4) {$3_{2,2}$};

	\draw[-,red] (H211) -- (3.1,0,-2) {};
	\draw[-,red,densely dashed] (1.8,0,-4) -- (2.9,0,-2) {};
	\draw[->,red,densely dashed] (3.1,0,-2) -- (2.2,0,-4) {};
	\draw[->,red] (2.9,0,-2) -- (H211) {};
	\draw[->,red,dashed] (H212) .. controls (3.5,0,-6) and (4.5,0,-6) .. (H212) {};
	\draw[->,red] (1.8,0,0) .. controls (1.5,0,-2) and (2.5,0,-2) .. (2.2,0,0) {};
	\draw[-,red] (H221) -- (3.1,2,-2) {};
	\draw[-,red,densely dashed] (1.8,2,-4) -- (2.9,2,-2) {};
	\draw[->,red,densely dashed] (3.1,2,-2) -- (2.2,2,-4) {};
	\draw[->,red] (2.9,2,-2) -- (H221) {};
	\draw[->,red,dashed] (H222) .. controls (3.5,2,-6) and (4.5,2,-6) .. (H222) {};
	\draw[->,red] (1.8,2,0) .. controls (1.5,2,-2) and (2.5,2,-2) .. (2.2,2,0) {};
	
	\draw[fill=red,draw=red,opacity=.2,very thin,line join=round]
		(4,0,-4) --
		(2,0,-4) --
		(2,0,0) --
		(4,0,0) --
		(4,0,-4) {} ;
	
	\draw[fill=red,draw=red,opacity=.2,very thin,line join=round]
 		(2,2,0) -- 
 		(4,2,0) --
		(4,2,-4) --
		(2,2,-4) --
		(2,2,0) {} ;
		
\end{tikzpicture}
\end{center}
\caption{Alternating paths in the plugging of thick graphs $G$ and $H$}\label{cheminsaltepais}
\end{figure}

\begin{figure}
\begin{center}
\begin{tikzpicture}
	\node (G111) at (0,0,0) {$1_{1,1}$};
	\node (G121) at (0,2,0) {$1_{2,1}$};
	\node[opacity=0.6] (G112) at (0,0,-4) {$1_{1,2}$};
	\node[opacity=0.6] (G122) at (0,2,-4) {$1_{2,2}$};
	\node (H211) at (4,0,0) {$3_{1,1}$};
	\node[opacity=0.6] (H212) at (4,0,-4) {$3_{1,2}$};
	\node (H221) at (4,2,0) {$3_{2,1}$};
	\node[opacity=0.6] (H222) at (4,2,-4) {$3_{2,2}$};

	\draw[<->,violet] (G111) .. controls (-1,0.5,0) and (-1,1.5,0) .. (G121) {};
	\draw[<->,violet] (G121) .. controls (-0.5,3,0) and (0.5,3,0) .. (G121) {};
	\draw[<->,violet] (G111) .. controls (-0.5,-1,0) and (0.5,-1,0) .. (G111) {};
	\draw[<->,violet,dashed] (G112) .. controls (-1,0.5,-4) and (-1,1.5,-4) .. (G122) {};
	\draw[<-,violet,densely dashed] (G122) -- (2,2,-2) {};
	\draw[<-,violet,densely dashed] (G112) -- (2,0,-2) {};
	\draw[->,violet] (2,2,-2) -- (H221) {};
	\draw[->,violet] (2,0,-2) -- (H211) {};
	\draw[->,violet,dashed] (H212) .. controls (3.5,0,-6) and (4.5,0,-6) .. (H212) {};
	\draw[->,violet,dashed] (H222) .. controls (3.5,2,-6) and (4.5,2,-6) .. (H222) {};

	\draw[->,dotted] (0,0,0) -- (7,0,0) {};
	\draw[->,dotted] (0,0,0) -- (0,5,0)  {};
	\draw[->,dotted] (0,0,0) -- (0,0,-12)  {};
	
	\draw[fill=blue,draw=red,opacity=.2,very thin,line join=round]
 		(0,0,0) -- 
 		(4,0,0) --
		(4,2,0) --
		(0,2,0) --
		(0,0,0) {} ;
	
	\draw[violet,opacity=0.6] (0,0,0) -- (4,0,0) node [sloped, above, midway, opacity=0.5] {\scriptsize{slice $(1,1)$}};
	\draw[violet,opacity=0.6] (0,2,0) -- (4,2,0) node [sloped, above, midway, opacity=0.5] {\scriptsize{slice $(2,1)$}};
	\draw[violet,opacity=0.6] (0,0,-4) -- (4,0,-4) node [sloped, below, midway, opacity=0.5] {\scriptsize{slice $(1,2)$}};
	\draw[violet,opacity=0.6] (0,2,-4) -- (4,2,-4) node [sloped, below, midway, opacity=0.5] {\scriptsize{slice $(2,2)$}};
	
	\draw[fill=blue,draw=red,opacity=.2,very thin,line join=round]
 		(0,0,-4) -- 
 		(4,0,-4) --
		(4,2,-4) --
		(0,2,-4) --
		(0,0,-4) {} ;
		
	\draw[fill=red,draw=red,opacity=.2,very thin,line join=round]
 		(0,0,0) -- 
 		(4,0,0) --
		(4,0,-4) --
		(0,0,-4) --
		(0,0,0) {} ;
	
	\draw[fill=red,draw=red,opacity=.2,very thin,line join=round]
 		(0,2,0) -- 
 		(4,2,0) --
		(4,2,-4) --
		(0,2,-4) --
		(0,2,0) {} ;
		
\end{tikzpicture}
\end{center}
\caption{Result of the execution of the thick graphs $G$ and $H$}\label{execepais3d}
\end{figure}

\begin{figure}
\begin{center}
\begin{tikzpicture}
	\node (G111) at (0,0) {$1_{1,1}$};
	\node (G121) at (0,2) {$1_{2,1}$};
	\node (G112) at (0,4) {$1_{1,2}$};
	\node (G122) at (0,6) {$1_{2,2}$};
	\node (H211) at (4,0) {$3_{1,1}$};
	\node (H212) at (4,4) {$3_{1,2}$};
	\node (H221) at (4,2) {$3_{2,1}$};
	\node (H222) at (4,6) {$3_{2,2}$};

	\draw[<->,violet] (G111) -- (G121) {};
	\draw[->,violet] (G121) .. controls (-1,1.5) and (-1,2.5) .. (G121) {};
	\draw[->,violet] (G111) .. controls (-1,-0.5) and (-1,0.5) .. (G111) {};
	\draw[<->,violet] (G112) -- (G122) {};
	\draw[<->,violet] (G122) -- (H221) {};
	\draw[<->,violet] (G112) -- (H211) {};
	\draw[->,violet] (H212) .. controls (5,3.5) and (5,4.5) .. (H212) {};
	\draw[->,violet] (H222) .. controls (5,5.5) and (5,6.5) .. (H222) {};
	
	\draw[dashed,violet] (-1,-1) -- (-1,3) node [sloped,above,near end] {\scriptsize{slice $(2,1)$}} node [sloped,above,near start] {\scriptsize{slice $(1,1)$}};
	\draw[dashed,violet] (-1,3) -- (-1,7) node [sloped,above,near end] {\scriptsize{slice $(2,2)$}} node [sloped,above,near start] {\scriptsize{slice $(1,2)$}};
	\draw[dashed,violet] (-1,-1) -- (5,-1) {};
	\draw[dashed,violet] (-1,7) -- (5,7) {};
	\draw[dashed,violet] (5,-1) -- (5,7) {};
	\draw[dotted,violet] (-1,1) -- (5,1) {};
	\draw[dotted,violet] (-1,3) -- (5,3) {};
	\draw[dotted,violet] (-1,5) -- (5,5) {};
	
	\node (graphe) at (-0.4,6.8) {\textcolor{violet}{$G\plugepais H$}};
	
\end{tikzpicture}
\end{center}
\caption{The thick graph $G\plugepais H$ represented in two dimensions.}\label{execepais2d}
\end{figure}

\begin{remark}\label{discussionassocepais} Since we only modified the graphs before plugging them together, we can make the following remark. Let $F,G,H$ be thick graphs. Then the thick graph $F\plugepais(G\plugepais H)$ is defined as 
$$F^{\dagger_{D^{G}\times D^{H}}}\plug (G^{\dagger_{D^{H}}}\plug H^{\ddagger_{D^{G}}})^{\ddagger_{D^{F}}}=F^{\dagger_{D^{G}\times D^{H}}}\plug((G^{\dagger_{D^{H}}})^{\ddagger_{D^{F}}}\plug H^{\ddagger_{D^{F}\times D^{G}}})$$ 
If one supposes that $S^{F}\cap S^{G}\cap S^{H}=\emptyset$, it is clear that $(S^{F}\times D)\cap(S^{G}\times D)\cap (S^{H}\times D)=\emptyset$. We can deduce from the associativity of execution that 
$$F^{\dagger_{D^{G}\times D^{H}}}\plug((G^{\dagger_{D^{H}}})^{\ddagger_{D^{F}}}\plug H^{\ddagger_{D^{F}\times D^{G}}})=(F^{\dagger_{D^{G}\times D^{H}}}\plug(G^{\dagger_{D^{H}}})^{\ddagger_{D^{F}}})\plug H^{\ddagger_{D^{F}\times D^{G}}}$$
But:
$$(F^{\dagger_{D^{G}\times D^{H}}}\plug(G^{\dagger_{D^{H}}})^{\ddagger_{D^{F}}})\plug H^{\ddagger_{D^{F}\times D^{G}}}=((F^{\dagger_{D^{G}}}\plug G^{\ddagger_{D^{F}}})^{\dagger_{D^{H}}})\plug H^{\ddagger_{D^{F}\times D^{G}}}$$
The latter is by definition the thick graph $(F\plugepais G)\plugepais H$. This shows that the associativity of $\plugepais$ on thick graphs is a simple consequence of the associativity of $\plug$ on simple graphs.
\end{remark}

\begin{proposition}[Associativity]
Let $F,G,H$ be thick graphs such that $S^{F}\cap S^{G}\cap S^{H}=\emptyset$. Then:
\begin{equation*}
F\plugepais(G\plugepais H)=(F\plugepais G)\plugepais H
\end{equation*}
\end{proposition}

\begin{definition}
Let $F$ and $G$ be two thick graphs. We define $\circuitsepais{F,G}$ as the set of 1-circuits in $F^{\dagger_{D^{G}}}\bicol G^{\ddagger_{D^{F}}}$. 

We also define, being given a dialect $D^{H}$,
\begin{itemize}[noitemsep,nolistsep]
\item the set $\circuitsepais{F,G}^{\dagger_{D^{H}}}$ of 1-circuits in the graph $(F^{\dagger_{D^{G}}})^{\dagger_{D^{H}}}\bicol (G^{\ddagger_{D^{F}}})^{\dagger_{D^{H}}}$
\item the set $\circuitsepais{F,G}^{\ddagger_{D^{H}}}$ of 1-circuits in the graph $(F^{\ddagger_{D^{G}}})^{\dagger_{D^{H}}}\bicol (G^{\ddagger_{D^{F}}})^{\dagger_{D^{H}}}$
\end{itemize}
\end{definition}

\begin{proposition}
Let $F,G,H$ be thick graphs and $\phi: D^{H}\rightarrow D$ a bijection. Then:
\begin{eqnarray*}
\circuitsepais{F,H}&\cong&\circuitsepais{F,H^{\phi}}\\
\circuitsepais{F,G}^{\dagger_{D^{H}}}&\cong&\circuitsepais{F,G}^{\dagger_{\phi(D^{H})}}\\
\circuitsepais{F,G}^{\dagger_{D^{H}}}&\cong&\circuitsepais{F,G}^{\ddagger_{D^{H}}}
\end{eqnarray*}
\end{proposition}


As in \autoref{discussionassocepais}, one considers the three thick graphs $F,G,H$. We are interested in the 1-circuits in $\circuitsepais{F,G\plugepais H}\cup(\circuitsepais{G,H}^{\ddagger{D^{F}}})$. By definition, these are the 1-circuits in one of the following graphs:
$$F^{\dagger_{D^{G}\times D^{H}}}\bicol((G^{\dagger_{D^{H}}}\plug H^{\ddagger_{D^{G}}})^{\ddagger_{D^{F}}})=F^{\dagger_{D^{G}\times D^{H}}}\bicol((G^{\dagger_{D^{H}}})^{\ddagger_{D^{F}}}\plug H^{\ddagger_{D^{F}\times D^{G}}})$$
$$(G^{\dagger_{D^{H}}}\bicol H^{\ddagger_{D^{G}}})^{\ddagger{D^{F}}}=(G^{\dagger_{D^{H}}})^{\ddagger_{D^{F}}}\bicol H^{\ddagger_{D^{F}\times D^{G}}}$$
We can now use the trefoil property to deduce that these sets of 1-circuits are in bijection with the set of 1-circuits in the following graphs:
$$(F^{\dagger_{D^{G}\times D^{H}}}\plug (G^{\ddagger_{D^{F}}})^{\dagger_{D^{H}}})\bicol H^{\ddagger_{D^{F}\times D^{G}}}=(F^{\dagger_{D^{G}}}\plug G^{\ddagger_{D^{F}}})^{\dagger_{D^{H}}}\bicol H^{\ddagger_{D^{F}\times D^{G}}}$$
$$F^{\dagger_{D^{G}\times D^{H}}}\bicol (G^{\ddagger_{D^{F}}})^{\dagger_{D^{H}}}=(F^{\dagger_{D^{G}}}\bicol G^{\ddagger_{D^{F}}})^{\dagger{D^{H}}}$$
This shows that the trefoil property holds for thick graphs.

\begin{proposition}[Geometric Trefoil Property for Thick Graphs]\label{cycliquegeomepais}
If $F$, $G$, $H$ are thick graphs such that $S^{F}\cap S^{G}\cap S^{H}=\emptyset$, then:
\begin{equation*}
\circuitsepais{F,G\plugepais H}\cup\circuitsepais{G,H}^{\dagger_{D^{F}}}\cong\circuitsepais{F\plugepais G, H}\cup\circuitsepais{F,G}^{\dagger_{D^{H}}}
\end{equation*}
\end{proposition}

\begin{corollary}[Geometric Adjunction for Thick Graphs]\label{Adjunctiongeo}
If $F$, $G$, $H$ are thick graphs such that $S^{G}\cap S^{H}=\emptyset$, we have:
\begin{equation*}
\circuitsepais{F,G\cup H}\cong\circuitsepais{F\plugepais G, H}\cup\circuitsepais{F,G}^{\dagger_{D^{H}}}
\end{equation*}
\end{corollary}

\begin{definition}
Let $\Omega$ be a monoid. A (graph) circuit quantifying map $m$ is a function $m:\Omega\rightarrow \realposN\{\infty\}$ which is \emph{tracial}, i.e.\ $m(ab)=m(ba)$ for all $a,b\in \Omega$.
\end{definition}

\begin{definition}
Being given a circuit quantifying map $m$, one can define a measurement of the interaction between thick graphs. For every couple of thick graphs $F,G$, it is defined as:
\begin{equation*}
\meas{F,G}=\sum_{\pi\in\circuitsepais{F,G}} \frac{1}{\card{D^{F}\times D^{G}}}m(\omega(\pi))
\end{equation*}
\end{definition}

\begin{proposition}[Numerical Trefoil Property for Thick Graphs]\label{numtrefoil}
Let $F,G,H$ be thick graphs such that $S^{F}\cap S^{G}\cap S^{H}=\emptyset$. Then:
\begin{equation*}
\meas{F,G\plugepais H}+\meas{G,H}=\meas{H,F\plugepais G}+\meas{F,G}
\end{equation*}
\end{proposition}

\begin{proof}
The proof is a simple calculation using the geometric trefoil property for thick graphs (\autoref{cycliquegeomepais}). We denote by $n^{F}$ (resp. $n^{G}$, $n^{H}$) the cardinality of the dialect $D^{F}$ (resp. $D^{G}$, $D^{H}$).
\begin{align*}
\meas{F,G&\plugepais H}+\meas{G,H}\\
=&~\sum_{\pi\in\circuitsepais{F,G\plugepais H}} \frac{1}{n^{F}n^{G}n^{H}}m(\omega(\pi))+\sum_{\pi\in\circuitsepais{G,H}} \frac{1}{n^{G}n^{H}}m(\omega(\pi))\\
=&~\sum_{\pi\in\circuitsepais{F,G\plugepais H}} \frac{1}{n^{F}n^{G}n^{H}}m(\omega(\pi))+\sum_{\pi\in\circuitsepais{G,H}^{\dagger_{D^{F}}}} \frac{1}{n^{F}n^{G}n^{H}}m(\omega(\pi))\\
=&~\sum_{\pi\in\circuitsepais{F,G\plugepais H}\cup\circuitsepais{G,H}^{\dagger_{D^{F}}}} \frac{1}{n^{F}n^{G}n^{H}}m(\omega(\pi))\\
=&~\sum_{\pi\in\circuitsepais{H,F\plugepais G}\cup\circuitsepais{F,G}^{\dagger_{D^{H}}}} \frac{1}{n^{F}n^{G}n^{H}}m(\omega(\pi))\\
=&~\sum_{\pi\in\circuitsepais{H,F\plugepais G}} \frac{1}{n^{F}n^{G}n^{H}}m(\omega(\pi)) + \sum_{\pi\in\circuitsepais{F,G}^{\dagger_{D^{H}}}}  \frac{1}{n^{F}n^{G}n^{H}}m(\omega(\pi))\\
=&~\sum_{\pi\in\circuitsepais{H,F\plugepais G}} \frac{1}{n^{F}n^{G}n^{H}}m(\omega(\pi)) + \sum_{\pi\in\circuitsepais{F,G}}  \frac{1}{n^{F}n^{G}}m(\omega(\pi))\\
=&~\meas{H,F\plugepais G}+\meas{F,G}&\tag*{\qedhere}
\end{align*}
\end{proof}

\begin{corollary}[Numerical Adjunction for Thick Graphs]\label{numadjunct}
Let $F,G,H$ be thick graphs such that $S^{G}\cap S^{H}=\emptyset$. Then:
\begin{equation*}
\meas{F,G\plugepais H}=\meas{H,F\plugepais G}+\meas{F,G}
\end{equation*}
\end{corollary}

\begin{remark}\noindent\textsc{About the Hidden Convention of the Numerical Measure}\label{secondeconventionepais}\newline
The measurement of interaction we defined hides a convention: each slice of a thick graph $F$ is considered as having a ``weight'' equal to $1/n^{F}$, so that the total weight of the set of all slices have weight $1$. This convention corresponds to the choice of working with a \emph{normalised trace} (such that $tr(1)=1$) on the idiom in Girard's hyperfinite geometry of interaction. It would have been possible to consider another convention which would impose that each slice have a weight equal to $1$ (this would correspond to working with the usual trace on matrices in Girard's hyperfinite geometry of interaction). In this case, the measurement of the interaction between two thick graphs $F,G$ is defined as: 
$$\meassecond{F,G}=\sum_{\pi\in\circuitsepais{F,G}}m(\omega(\pi))$$
The numerical trefoil property is then stated differently: for all thick graphs $F$, $G$, and $H$ such that $S^{F}\cap S^{G}\cap S^{H}=\emptyset$, we have: 
$$\meassecond{F,G\plugepais H}+n^{F}\meassecond{G,H}=\meassecond{H,F\plugepais G}+n^{H}\meassecond{F,G}$$
We stress the apparition of the terms $n^{F}$ and $n^{H}$ in this equality: their apparition corresponds exactly to the apparition of the terms $\unit{F}$ and $\unit{H}$ in the equality stated for the trefoil property for sliced graphs.
\end{remark}

\subsection{Sliced Thick Graphs}

One can of course apply the additive construction presented in our previous paper \cite{seiller-goiadd} in the case of thick graphs. A \emph{sliced thick graph} $G$ of carrier $S^{G}$ is a finite family $\sum_{i\in I^{G}} \alpha^{G}_{i} G_{i}$ where, for all $i\in I^{G}$, $G_{i}$ is a thick graph such that $S^{G_{i}}=S^{G}$, and $\alpha^{G}_{i}\in\realN$. We define the \emph{dialect} of $G$ to be the set $\disjun_{i\in I^{G}} D^{G_{i}}$. We will abusively call a \emph{slice} a couple $(i,d)$ where $i\in I^{G}$ and $d\in D_{G_{i}}$; we will say a graph $G$ is a \emph{single-sliced graph} when $I^{G}=\{i\}$ and $D_{G_{i}}=\{d\}$ are both one-element sets.

One can extend the execution and the measurement of the interaction by applying the thick graphs constructions slice by slice:
\begin{eqnarray*}
(\sum_{i\in I^{F}} \alpha^{F}_{i}F_{i})\plug(\sum_{i\in I^{G}} \alpha^{G}_{i}G_{i})&=&\!\!\!\!\sum_{(i,j)\in I^{F}\times I^{G}}\!\!\!\!\alpha^{F}_{i}\alpha^{G}_{j} F_{i}\plugepais G_{j}\\
\meas{\sum_{i\in I^{F}} \alpha^{F}_{i}F_{i},\sum_{i\in I^{G}} \alpha^{G}_{i}G_{i}}&=&\!\!\!\!\sum_{(i,j)\in I^{F}\times I^{G}}\!\!\!\!\alpha^{F}_{i}\alpha^{G}_{j} \meas{F_{i},G_{j}}
\end{eqnarray*}

\autoref{exemplesliceepais} shows two examples of sliced thick graphs. The graphical convention we will follow for representing sliced and thick graphs corresponds to the graphical convention for sliced graphs, apart from the fact that the graphs contained in the slices are replaced by thick graphs. Thus, two slices are separated by a dashed line, two elements in the dialect of a thick graph (i.e.\ the graph contained in a slice) are separated by a dotted line.

\begin{figure}
\begin{center}
\begin{tikzpicture}
	\node (A1) at (0,0) {$1_{1}$};
	\node (B1) at (0,2) {$1_{2}$};
	\node (C1) at (0,4) {$1_{3}$};
	\node (A2) at (2,0) {$2_{1}$};
	\node (B2) at (2,2) {$2_{2}$};
	\node (C2) at (2,4) {$2_{3}$};
	\node (A3) at (4,0) {$3_{1}$};
	\node (B3) at (4,2) {$3_{2}$};
	\node (C3) at (4,4) {$3_{3}$};

	\draw[->,blue] (A1) -- (B2) {};
	\draw[->,blue] (B2) .. controls (1.5,3) and (2.5,3) .. (B2) {};
	\draw[->,blue] (A2) .. controls (2.5,-1) and (3.5,-1) .. (A3) {};
	\draw[->,blue] (B3) -- (A2) {};
	\draw[->,blue] (C1) .. controls (1,5) and (3,5) .. (C3) {};
	\draw[->,blue] (C3) .. controls (3.5,3) and (2.5,3) .. (C2) {};
	
	\draw[-,dashed,blue] (-1,-1) -- (-1,5) {};
	\draw[-,dashed,blue] (-1,-1) -- (5,-1) {};
	\draw[-,dashed,blue] (5,-1) -- (5,5) {};
	\draw[-,dashed,blue] (5,5) -- (-1,5) {};
	
	\node (F) at (4.6,4.6) {\textcolor{blue}{$\frac{1}{2}F$}};
	\node (F) at (4.6,2.7) {\textcolor{blue}{$3G$}};
	
	\draw[-,dashed,blue] (-1,3) -- (5,3) {};
	\draw[-,dotted,blue] (-1,1) -- (5,1) {};
	
	\node (1a) at (7,1) {$1_{a}$};
	\node (2a) at (9,1) {$2_{a}$};
	\node (1b) at (7,3) {$1_{b}$};
	\node (2b) at (9,3) {$2_{b}$};
	
	\draw[->,red] (1a) .. controls (7.5,2) and (8.5,2) .. (2a) {};
	\draw[->,red] (2a) .. controls (8.5,0) and (9.5,0) .. (2a) {};
	\draw[->,red] (1b) .. controls (7.5,4) and (8.5,4) .. (2b) {};
	\draw[->,red] (1b) .. controls (6.5,2) and (7.5,2) .. (1b) {};

	\draw[dashed,red] (6,0) -- (6,4) {};
	\draw[dashed,red] (6,4) -- (10,4) {};
	\draw[dashed,red] (10,4) -- (10,0) {};
	\draw[dashed,red] (10,0) -- (6,0) {};

	\draw[dashed,red] (6,2) -- (10,2) {};
	
	\node (Fa) at (9.7,1.7) {\textcolor{red}{$F_{a}$}};
	\node (Fb) at (9.7,3.7) {\textcolor{red}{$F_{b}$}};

\end{tikzpicture}
\end{center}
\caption{Examples of sliced thick graphs: $\frac{1}{2} F + 3 G$ and $F_{a}+F_{b}$}\label{exemplesliceepais}
\end{figure}

One should however notice that some sliced thick graphs (for instance the graph $F_{a}+F_{b}$ represented in red in \autoref{exemplesliceepais}) can be considered both as a thick graph -- hence a sliced thick graph with a single slice -- or as a sliced graph with two slices -- hence a sliced thick graph with two slices. Indeed, consider the graphs:
\begin{equation*}
\begin{array}{rcl|rcl|rcl}
&F_{a} &&&F_{b}&&&F_{c}\\ \hline\hline
V^{F_{a}}&=&\{1,2\}&V^{F_{b}}&=&\{1,2\}&V^{F_{c}}&=&\{1,2\}\times\{a,b\} \\
E^{F_{a}}&=&\{f,g\}&E^{F_{b}}&=&\{f,g\}&E^{F_{c}}&=&\{f_{a},f_{b},g_{a},g_{b}\}\\
s^{F_{a}}&=&\left\{\begin{array}{l} f\mapsto 1 \\ g\mapsto 2\end{array}\right.
	&
		s^{F_{b}}&=&\left\{\begin{array}{l} f\mapsto 1 \\ g\mapsto 1\end{array}\right.
			&
				s^{F_{c}}&=&\left\{\begin{array}{l} f_{i}\mapsto s^{F_{i}}(f) \\ g_{i}\mapsto s^{F_{i}}(g)\end{array}\right.\\
t^{F_{a}}&=&\left\{\begin{array}{l} f\mapsto 2 \\ g\mapsto 2\end{array}\right.
	&
		t^{F_{b}}&=&\left\{\begin{array}{l} f\mapsto 2 \\ g\mapsto 1\end{array}\right.
			&
				t^{F_{c}}&=&\left\{\begin{array}{l} f_{i}\mapsto t^{F_{i}}(f) \\ g_{i}\mapsto t^{F_{i}}(g)\end{array}\right.\\
\omega^{F_{a}}&\equiv&1&\omega^{F_{b}}&\equiv&1&\omega^{F_{c}}&\equiv&1
\end{array}
\end{equation*}

One can then define the the two sliced thick graphs $G_{1}=F_{c}$ and $G_{2}=\frac{1}{2}F_{a}+\frac{1}{2}F_{b}$. These two graphs are represented in \autoref{deuxgraphesepais}. They are similar in a very precise sense: one can show that if $H$ is any sliced thick graph, and $m$ is any circuit-quantifying map, then $\meas{G_{1},H}=\meas{G_{2},H}$. We say they are \emph{universally equivalent}. Notice that this explains in a very formal way the remark about the convention on the measurement of interaction \autoref{secondeconventionepais}.

\begin{figure}
\begin{center}
\begin{tikzpicture}

	\node (A1a) at (1,1) {$1_{a}$};
	\node (A2a) at (3,1) {$2_{a}$};
	\node (A1b) at (1,3) {$1_{b}$};
	\node (A2b) at (3,3) {$2_{b}$};
	
	\draw[->,blue] (A1a) .. controls (1.5,2) and (2.5,2) .. (A2a) {};
	\draw[->,blue] (A2a) .. controls (2.5,0) and (3.5,0) .. (A2a) {};
	\draw[->,blue] (A1b) .. controls (1.5,4) and (2.5,4) .. (A2b) {};
	\draw[->,blue] (A1b) .. controls (0.5,2) and (1.5,2) .. (A1b) {};

	\draw[dashed,blue] (0,0) -- (0,4) {};
	\draw[dashed,blue] (0,4) -- (4,4) {};
	\draw[dashed,blue] (4,4) -- (4,0) {};
	\draw[dashed,blue] (4,0) -- (0,0) {};
	
	\draw[dotted,blue] (0,2) -- (4,2) {};

	\node (G1) at (3.7,3.7) {\textcolor{blue}{$F_{c}$}};

	\node (1a) at (7,1) {$1_{a}$};
	\node (2a) at (9,1) {$2_{a}$};
	\node (1b) at (7,3) {$1_{b}$};
	\node (2b) at (9,3) {$2_{b}$};
	
	\draw[->,red] (1a) .. controls (7.5,2) and (8.5,2) .. (2a) {};
	\draw[->,red] (2a) .. controls (8.5,0) and (9.5,0) .. (2a) {};
	\draw[->,red] (1b) .. controls (7.5,4) and (8.5,4) .. (2b) {};
	\draw[->,red] (1b) .. controls (6.5,2) and (7.5,2) .. (1b) {};

	\draw[dashed,red] (6,0) -- (6,4) {};
	\draw[dashed,red] (6,4) -- (10,4) {};
	\draw[dashed,red] (10,4) -- (10,0) {};
	\draw[dashed,red] (10,0) -- (6,0) {};

	\draw[dashed,red] (6,2) -- (10,2) {};
	
	\node (Fa) at (9.6,1.6) {\textcolor{red}{$\frac{1}{2}F_{a}$}};
	\node (Fb) at (9.6,3.6) {\textcolor{red}{$\frac{1}{2}F_{b}$}};

\end{tikzpicture}
\end{center}
\caption{Graphs $G_{1}$ and $G_{2}$}\label{deuxgraphesepais}
\end{figure}

\begin{definition}[Universally equivalent graphs]\label{equivalenceuniv}
Let $F,G$ be two graphs. We say that $F$ and $G$ are \emph{universally equivalent} (for the measurement $\meas{\cdot,\cdot}$) -- which will be denoted by $F\simeq_{u}G$ -- if for all graph $H$:
\begin{equation*}
\meas{F,H}=\meas{G,H}
\end{equation*}
\end{definition}

The next proposition states that if $F'$ is obtained from a graph $F$ by a renaming of edges, then $F\simeq_{u} F'$.
\begin{proposition}\label{equivunivaretes}
Let $F,F'$ be two graphs such that $V^{F}=V^{F'}$, and $\phi$ a bijection $E^{F}\rightarrow E^{F'}$ such that:
\begin{eqnarray*}
&s^{G}\circ\phi=s^{F},~~
t^{G}\circ\phi=t^{F},~~
\omega^{G}\circ\phi=\omega^{F}
\end{eqnarray*}
Then $F\simeq_{u} F'$.
\end{proposition}

\begin{proof}
Indeed, the bijection $\phi$ induces, from the hypotheses in the source and target functions, a bijection between the sets of cycles $\cycles{F, H}$ and $\cycles{G,H}$. The condition on the weight map then insures us that this bijection is $\omega$-invariant, from which we deduce the proposition.
\end{proof}

\begin{proposition}\label{univequivindex}
Let $F,G$ be sliced graphs. If there exists a bijection $\phi:I^{F}\rightarrow I^{G}$ such that $F_{i}=G_{\phi(i)}$ and $\alpha^{F}_{i}=\alpha^{G}_{\phi(i)}$, then $F\simeq_{u} G$.
\end{proposition}

\begin{proof}
By definition:
\begin{eqnarray*}
\meas{G,H}&=&\sum_{(i,j)\in I^{G}\times I^{H}} \alpha^{G}_{i}\alpha^{H}_{j} \meas{G_{i},H_{j}}\\
&=&\sum_{(i,j)\in I^{F}\times I^{H}} \alpha^{G}_{\phi(i)}\alpha^{H}_{j} \meas{G_{\phi(i)},H_{j}}\\
&=&\sum_{(i,j)\in I^{F}\times I^{G}} \alpha^{F}_{i}\alpha^{G}_{j} \meas{F_{i},G_{j}}
\end{eqnarray*}
Thus $F$ and $G$ are universally equivalent.
\end{proof}

\begin{proposition}\label{univequivdialecte}
Let $F,G$ be thick graphs. If there exists a bijection $\phi:D^{F}\rightarrow D^{G}$ such that $G=F^{\phi}$, then $F\simeq_{u} G$.
\end{proposition}

\begin{proof}
Let $F,G$ be thick graphs such that $G=F^{\phi}$ for a bijection $\phi:D^{G}\rightarrow D^{F}$, and $H$ an arbitrary thick graph. Then the bijection $\phi\times\text{Id}: D^{G}\times D^{H}\rightarrow D^{F}\times D^{H}$ satisfies that $G^{\dagger_{D^{H}}}=(F^{\dagger_{D^{H}}})^{\phi\times\text{Id}}$. One can notice that the set of alternating 1-circuits in $F^{\dagger}\bicol H^{\ddagger}$ is the same as the set of alternating 1-circuits in $(F^{\dagger})^{\phi\times\text{Id}}\bicol (H^{\dagger})^{\phi\times\text{Id}}=G^{\dagger}\bicol H^{\ddagger}$. Thus:
\begin{eqnarray*}
\meas{F,H}&=&\sum_{\pi\in\cycles{F,H}} m(\omega(\pi))\\
&=&\sum_{\pi\in\cycles{G,H}} m(\omega(\pi))\\
&=&\meas{G,H}
\end{eqnarray*}
And finally $F$ and $G$ are universally equivalent.
\end{proof}

\begin{proposition}\label{fromETtoEequiv}
Let $F=\sum_{i\in I^{F}} \alpha_{i}^{F} F_{i}$ be a sliced thick graph, and let us define, for all $i\in I^{F}$, $n^{F_{i}}=\card{D^{F_{i}}}$ and $n^{F}=\sum_{i\in I^{F}} n^{F_{i}}$. Suppose that there exists a scalar $\alpha$ such that for all $i\in I^{F}$, $\alpha_{i}^{F}=\alpha \frac{n^{F_{i}}}{n^{F}}$. We then define the sliced thick graph $G$ with a single slice $\alpha G$ of dialect $\disjun D^{F_{i}}=\cup_{i\in I^{F}} D^{F_{i}}\times\{i\}$ and carrier $V^{F}$ by:
\begin{eqnarray*}
V^{G}&=&V^{F}\times \disjun D^{F_{i}}\\
E^{G}&=&\disjun E^{F_{i}} = \cup_{i\in I^{F}} E^{F_{i}}\times\{i\}\\
s^{G}&=&(e,i)\mapsto (s^{F_{i}}(e),i)\\
t^{G}&=&(e,i)\mapsto (t^{F_{i}}(e),i)\\
\omega^{G}&=&(e,i)\mapsto \omega^{F_{i}}(e)
\end{eqnarray*}
Then $F$ and $G$ are universally equivalent.
\end{proposition}

\begin{proof}
Let $H$ be a sliced thick graph. Then:
\begin{eqnarray*}
\meas{F,H}&=&\sum_{i\in I^{H}}\sum_{j\in I^{F}}\alpha^{H}_{i}\alpha^{F}_{j}\meas{F_{i},H_{j}}\\
&=&\sum_{i\in I^{H}}\sum_{j\in I^{F}}\alpha^{H}_{i}\alpha\frac{n^{F_{i}}}{n^{F}}\meas{F_{i},H_{j}}\\
&=&\sum_{i\in I^{H}}\sum_{j\in I^{F}}\alpha^{H}_{i}\alpha\frac{n^{F_{i}}}{n^{F}}\frac{1}{n^{F_{i}}n^{H_{j}}}\sum_{\pi\in\cycles{F_{i},H_{j}}}m(\omega(\pi))\\
&=&\sum_{i\in I^{H}}\sum_{j\in I^{F}}\alpha^{H}_{i}\alpha\frac{1}{n^{F}n^{H_{j}}}\sum_{\pi\in\cycles{F_{i},H_{j}}}m(\omega(\pi))\\
&=&\sum_{i\in I^{H}}\alpha^{H}_{i}\alpha\frac{1}{n^{F}n^{H_{j}}}\sum_{j\in I^{F}}\sum_{\pi\in\cycles{F_{i},H_{j}}}m(\omega(\pi))
\end{eqnarray*}
But one can notice that $\cup_{j\in I^{F}} \cycles{F_{i},H_{j}}=\cycles{G,H}$. We thus get:
\begin{eqnarray*}
\meas{F,H}&=&\sum_{i\in I^{H}}\alpha^{H}_{i}\alpha\frac{1}{n^{F}n^{H_{j}}}\sum_{j\in I^{F}}\sum_{\pi\in\cycles{F_{i},H_{j}}}m(\omega(\pi))\\
&=&\sum_{i\in I^{H}}\alpha^{H}_{i}\alpha\frac{1}{n^{F}n^{H_{j}}}\sum_{\pi\in\cycles{F,H_{j}}}m(\omega(\pi))\\
&=&\meas{\alpha G,H}
\end{eqnarray*}
Finally, we showed that $F$ and $\alpha G$ are universally equivalent.
\end{proof}


One of the consequences of \autoref{equivunivaretes}, \autoref{univequivindex}, and \autoref{univequivdialecte} is that two graphs $F,G$ such that $G$ is obtained from $F$ by a renaming of the sets $E^{F},I^{F},D^{F}$ are universally equivalent. We will therefore work from now on with graphs modulo renaming of these sets.

\subsection{Thick Graphs and Contraction}\label{subsec:contraction}

In this section, we will explain how the introduction of thick graphs allow the definition of contraction by using the fact that edges can go from a slice to another (contrarily to sliced graphs). In the following, we will be working with sliced thick graphs. The way contraction is dealt with by using slice-changing edges is quite simple, and the graph which will implement this transformation is essentially the same as the graph implementing additive contraction (i.e.\ the graph implementing distributivity -- \autoref{distributivity} -- restricted to the location of contexts\footnote{We did not insist on this up to this point in the paper, but the interpretation of proofs is here \emph{locative}, i.e.\ to define properly the interpretation of proofs and formulas, all occurrences of formula are given a distinct \emph{location} -- here a set of vertices; a measurable set when we will consider graphings.}) modified with a change of slices.

The graph we obtain is then the superimposition of two $\de{Fax}$, but where one of them goes from one slice to the other.

\begin{figure}
\centering
\begin{tikzpicture}[x=0.72cm,y=0.72cm]
	\node (A1) at (0,0) {$1_{1}$};
	\node (B1) at (1.5,0) {$2_{1}$};
	\node (C1) at (3,0) {$3_{1}$};
	\node (D1) at (4.5,0) {$4_{1}$};
	\node (E1) at (6,0) {$5_{1}$};
	\node (F1) at (7.5,0) {$6_{1}$};
	\node (G1) at (9,0) {$7_{1}$};
	\node (H1) at (10.5,0) {$8_{1}$};
	\node (I1) at (12,0) {$9_{1}$};
	\node (A2) at (0,2) {$1_{2}$};
	\node (B2) at (1.5,2) {$2_{2}$};
	\node (C2) at (3,2) {$3_{2}$};
	\node (D2) at (4.5,2) {$4_{2}$};
	\node (E2) at (6,2) {$5_{2}$};
	\node (F2) at (7.5,2) {$6_{2}$};
	\node (G2) at (9,2) {$4_{2}$};
	\node (H2) at (10.5,2) {$5_{2}$};
	\node (I2) at (12,2) {$6_{2}$};

	\draw[<->,blue] (A2) -- (D1) {};
	\draw[<->,blue] (B2) -- (E1) {};
	\draw[<->,blue] (C2) -- (F1) {};
	\draw[<->,blue] (A1) .. controls (2,-1) and (10,-1) .. (I1) {};
	\draw[<->,blue] (B1) .. controls (3.5,-0.8) and (8.5,-0.8) .. (H1) {};
	\draw[<->,blue] (C1) .. controls (5,-0.6) and (7,-0.6) .. (G1) {};

	\draw[opacity=0.2,fill=blue] 
		(-1,-1) -- 
			(13,-1) --
				(13,3) --
					(-1,3) --
						(-1,-1) {};
	\draw[dashed,blue] (-1,-1) -- (13,-1) {};
	\draw[dashed,blue] (13,-1) -- (13,3) {};
	\draw[dashed,blue] (13,3) -- (-1,3) {};
	\draw[dashed,blue] (-1,3) -- (-1,-1) {};
	\draw[dotted,blue] (-1,1) -- (13,1) {};
	
\end{tikzpicture}
\caption{The graph of a contraction project}\label{exemplecontr}
\end{figure}

\begin{definition}[Contraction]\label{contraction}
Let $\phi: V^{A}\rightarrow W_{1}$ and $\psi: V^{A}\rightarrow W_{2}$ be two bijections with $V^{A}\cap W_{1}=V^{A}\cap W_{2}=W_{1}\cap W_{2}=\emptyset$. We define the project $\de{Ctr}^{\phi}_{\psi}=(0,\text{\textnormal{Ctr}}^{\phi}_{\psi})$, where the graph $\text{\textnormal{Ctr}}^{\phi}_{\psi}$ is defined by:
\begin{eqnarray*}
V^{\text{\textnormal{Ctr}}^{\phi}_{\psi}}&=&V^{A}\cup W_{1}\cup W_{2}\\
D^{\text{\textnormal{Ctr}}^{\phi}_{\psi}}&=&\{1,2\}\\
E^{\text{\textnormal{Ctr}}^{\phi}_{\psi}}&=&V^{A}\times\{1,2\}\times\{i,o\}\\
s^{\text{\textnormal{Ctr}}^{\phi}_{\psi}}&=&\left\{\begin{array}{rcl} 
										(v,1,o)&\mapsto& (\phi(v),1)\\
										(v,1,i)&\mapsto& (v,1)\\
										(v,2,o)&\mapsto& (\psi(v),1)\\
										(v,2,i)&\mapsto& (v,2)
								\end{array}\right.\\
t^{\text{\textnormal{Ctr}}^{\phi}_{\psi}}&=&\left\{\begin{array}{rcl} 
										(v,1,o)&\mapsto& (v,1)\\
										(v,1,i)&\mapsto& (\phi(v),1)\\
										(v,2,o)&\mapsto& (v,2)\\
										(v,2,i)&\mapsto& (\psi(v),1)
								\end{array}\right.\\
\omega^{\text{\textnormal{Ctr}}^{\phi}_{\psi}}&\equiv&1
\end{eqnarray*}
\end{definition}

\autoref{exemplecontr} illustrates the graph of the project $\de{Ctr}^{\psi}_{\phi}$, where the functions are defined by $\phi: \{1,2,3\}\rightarrow \{4,5,6\}, x\mapsto x+3$ and $\psi:\{1,2,3\}\rightarrow\{7,8,9\}, x\mapsto 10-x$.

\begin{proposition}
Let $\de{a}=(0,A)$ be a project in a behaviour $\cond{A}$, such that $D^{A}\cong\{1\}$. Let $\phi,\psi$ be two delocations  $V^{A}\rightarrow W_{1}$, $V^{A}\rightarrow W_{2}$ of disjoint codomains. Then $\de{Ctr}^{\psi}_{\phi}\plug\de{a}\in\cond{\phi(A)\otimes\psi(A)}$.
\end{proposition}

\begin{proof}
We will denote by $\text{Ctr}$ the graph $\text{Ctr}_{\phi}^{\psi}$ to simplify the notations. We first compute $A\plugepais \text{Ctr}$. We get $A^{\ddagger_{\{1,2\}}}=(V^{A}\times\{1,2\},E^{A}\times\{1,2\},s^{A}\times Id_{\{1,2\}}, t^{A}\times Id_{\{1,2\}}, \omega^{A}\circ\pi)$ where $\pi$ is the projection: $E^{A}\times\{1,2\}\rightarrow E^{A}, (x,i)\mapsto x$. Moreover the graph $\text{Ctr}^{\dagger_{D^{A}}}$ is a variant of the graph $\text{Ctr}$ since $D^{A}\cong\{1\}$. Here is what the plugging of $\text{Ctr}^{\dagger_{D^{A}}}$ with $A^{\ddagger_{\{1,2\}}}$ looks like:
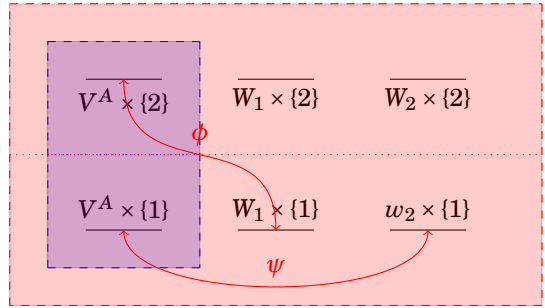
\begin{center}
\begin{tikzpicture}
	\draw (0,2) -- (1,2) node [midway,below] {$V^{A}\times\{2\}$};
	\draw (2,2) -- (3,2) node [midway,below] {$W_{1}\times\{2\}$};
	\draw (4,2) -- (5,2) node [midway,below] {$W_{2}\times\{2\}$};
	\draw (0,0) -- (1,0) node [midway,above] {$V^{A}\times\{1\}$};
	\draw (2,0) -- (3,0) node [midway,above] {$W_{1}\times\{1\}$};
	\draw (4,0) -- (5,0) node [midway,above] {$w_{2}\times\{1\}$};

	\draw[fill=blue,opacity=0.2] 
	(-0.5,-0.5) --
		(-0.5,2.5) --
			(1.5,2.5) --
				(1.5,-0.5) --
					(-0.5,-0.5) {};
	\draw[blue,dashed] 	(-0.5,-0.5) -- (-0.5,2.5) {};
	\draw[blue,dashed] 	(-0.5,2.5) -- (1.5,2.5) {};
	\draw[blue,dashed] 	(1.5,2.5) -- (1.5,-0.5) {};
	\draw[blue,dashed] 	(1.5,-0.5) -- (-0.5,-0.5) {};
	\draw[blue,densely dotted] (-0.5,1) -- (1.5,1) {};
	
	\draw[fill=red,opacity=0.2]
	(-1,-1) --
		(-1,3) --
			(6,3) --
				(6,-1) --
					(-1,-1) {};
	\draw[red,dashed] (-1,-1) -- (-1,3) {};
	\draw[red,dashed] (-1,3) -- (6,3) {};
	\draw[red,dashed] (6,3) -- (6,-1) {};
	\draw[red,dashed] (6,-1) -- (-1, -1) {};
	\draw[red,dotted] (-1,1) -- (6,1) {};
	
	\draw[<->,red] (0.5,2) .. controls (0.5,0.5) and (2.5,1.5) .. (2.5,0) node [midway,above] {$\phi$};
	\draw[<->,red] (0.5,0) .. controls (0.5,-1) and (4.5,-1) .. (4.5,0) node [midway,above] {$\psi$};
\end{tikzpicture}
\end{center}
The result of the execution is therefore a two-sliced graph, i.e.\ a graph of dialect $D^{A}\times\{1,2\}\cong\{1,2\}$, and which contains the graph $\phi(A)\cup\psi(A)$ in the slice numbered $1$ and contains the empty graph in the slice numbered $2$.

We deduce from this that $\de{Ctr}^{\psi}_{\phi}\plug\de{a}$ is universally equivalent (\autoref{equivalenceuniv}) to the project $\frac{1}{2}\de{\phi(a)\otimes\psi(a)}+\frac{1}{2}\de{0}$ from \autoref{fromETtoEequiv}. Since $\de{\phi(a)\otimes\psi(a)}$ belongs to $\cond{\phi(A)\otimes\psi(A)}$, then the project $\frac{1}{2}(\de{\phi(a)\otimes\psi(a)})$ is an element in $\cond{\phi(A)\otimes\psi(A)}$ by the homothety Lemma (\autoref{homothetie}).  Moreover, $\cond{A}$ is a behaviour, hence $\cond{\phi(A)\otimes\psi(A)}$ is also a behaviour and we can deduce that $\frac{1}{2}\de{\phi(a)\otimes\psi(a)}+\frac{1}{2}\de{0}$ is an element in $\cond{\phi(A)\otimes\psi(A)}$.
\end{proof}

\autoref{banchementcontrA}, \autoref{execcontrA} and \autoref{execcontrB} illustrate the plugging and execution of a contraction with two graphs: the first -- $A$ -- having a single slice, and the other -- $B$ -- having two slices (the graphs are shown in \autoref{graphescontraction}). One can see that the hypothesis $D^{A}\equiv\{1\}$ used in the preceding proposition is necessary, and that slice-changing edges allow to implement contraction of graphs with a single slice.

\begin{figure}[b]
\centering
\subfigure[The graph of the project $\de{Ctr}^{\phi}_{\psi}$]{
\begin{tikzpicture}[x=0.72cm,y=0.72cm]
	\node (A1) at (0,0) {$1_{1}$};
	\node (B1) at (1.5,0) {$2_{1}$};
	\node (C1) at (3,0) {$3_{1}$};
	\node (D1) at (4.5,0) {$4_{1}$};
	\node (E1) at (6,0) {$5_{1}$};
	\node (F1) at (7.5,0) {$6_{1}$};
	\node (A2) at (0,2) {$1_{2}$};
	\node (B2) at (1.5,2) {$2_{2}$};
	\node (C2) at (3,2) {$3_{2}$};
	\node (D2) at (4.5,2) {$4_{2}$};
	\node (E2) at (6,2) {$5_{2}$};
	\node (F2) at (7.5,2) {$6_{2}$};

	\draw[<->,blue] (A2) -- (E1) {};
	\draw[<->,blue] (B2) -- (F1) {};
	\draw[<->,blue] (A1) .. controls (1,-1) and (2,-1) .. (C1) {};
	\draw[<->,blue] (B1) .. controls (2.5,-1) and (3.5,-1) .. (D1) {};

	\draw[opacity=0.2,fill=blue] 
		(-1,-1) -- 
			(8.5,-1) --
				(8.5,3) --
					(-1,3) --
						(-1,-1) {};
	\draw[dashed,blue] (-1,-1) -- (8.5,-1) {};
	\draw[dashed,blue] (8.5,-1) -- (8.5,3) {};
	\draw[dashed,blue] (8.5,3) -- (-1,3) {};
	\draw[dashed,blue] (-1,3) -- (-1,-1) {};
	\draw[dotted,blue] (-1,1) -- (8.5,1) {};
	
\end{tikzpicture}}
\subfigure[The graphs $A$ and $B$ of the projects $\de{a}$ and $\de{b}$]{
\begin{tikzpicture}[x=0.72cm,y=0.72cm]
	\node (A1) at (0,0) {$1$};
	\node (B1) at (2,0) {$2$};
	
	\draw[<->,red] (A1) .. controls (0,1) and (2,1) .. (B1) {};

	\draw[fill=red,opacity=0.2]
		(-1,-1) --
			(-1,1.5) --
				(3,1.5) --
					(3,-1) --
						(-1,-1) {};
	\draw[dashed,red] (-1,-1) -- (-1,1.5) {};
	\draw[dashed,red] (-1,1.5) -- (3,1.5) {};
	\draw[dashed,red] (3,1.5) -- (3,-1) {};
	\draw[dashed,red] (3,-1) -- (-1,-1) {};
	\node (nodeA) at (-0.7,1.2) [red] {$A$};

	\node (A21) at (6,-1) {$1_{1}$};
	\node (A22) at (6,1) {$1_{2}$};
	\node (B21) at (8,-1) {$2_{1}$};
	\node (B22) at (8,1) {$2_{2}$};

	\definecolor{darkgreen}{rgb}{0.15,0.6,0.15};
	\draw[<->,darkgreen] (A21) .. controls (6,0) and (8,0) .. (B21) {};
	\draw[<->,darkgreen] (A22) .. controls (5,2) and (7,2) .. (A22) {};
	\draw[<->,darkgreen] (B22) .. controls (7,2) and (9,2) .. (B22) {};

	\draw[fill=darkgreen,opacity=0.2]
		(5,-2) --
			(5,2.5) --
				(9,2.5) --
					(9,-2) --
						(5,-2) {};
	\draw[dashed,darkgreen] (5,-2) -- (5,2.5) {};
	\draw[dashed,darkgreen] (5,2.5) -- (9,2.5) {};
	\draw[dashed,darkgreen] (9,2.5) -- (9,-2) {};
	\draw[dashed,darkgreen] (9,-2) -- (5,-2) {};
	\draw[dotted,darkgreen] (5,0.5) -- (9,0.5) {};
	\node (nodeB) at (8.7,2.2) [darkgreen] {$B$};		
\end{tikzpicture}}
\caption{The graphs of the projects $\de{Ctr}^{\phi}_{\psi}$, $\de{a}$ and $\de{b}$.}\label{graphescontraction}
\end{figure}

\begin{figure}
\centering
\subfigure[Plugging of $\text{\textnormal{Ctr}}^{\phi}_{\psi}$ and $A$]{
\begin{tikzpicture}[x=0.72cm,y=0.72cm]
	\node (A1) at (0,0) {$1_{1}$};
	\node (B1) at (2,0) {$2_{1}$};
	\node (C1) at (4,0) {$3_{1}$};
	\node (D1) at (6,0) {$4_{1}$};
	\node (E1) at (8,0) {$5_{1}$};
	\node (F1) at (10,0) {$6_{1}$};
	\node (A2) at (0,3) {$1_{2}$};
	\node (B2) at (2,3) {$2_{2}$};
	\node (C2) at (4,3) {$3_{2}$};
	\node (D2) at (6,3) {$4_{2}$};
	\node (E2) at (8,3) {$5_{2}$};
	\node (F2) at (10,3) {$6_{2}$};

	\draw[<->,blue] (A2) -- (E1) {};
	\draw[<->,blue] (B2) -- (F1) {};
	\draw[<->,blue] (A1) .. controls (0,-1) and (4,-1) .. (C1) {};
	\draw[<->,blue] (B1) .. controls (2,-1) and (6,-1) .. (D1) {};

	\draw[opacity=0.2,fill=blue] 
		(-1.5,-1.5) -- 
			(11.5,-1.5) --
				(11.5,4.5) --
					(-1.5,4.5) --
						(-1.5,-1.5) {};
	\draw[dashed,blue] (-1.5,-1.5) -- (11.5,-1.5) {};
	\draw[dashed,blue] (11.5,-1.5) -- (11.5,4.5) {};
	\draw[dashed,blue] (11.5,4.5) -- (-1.5,4.5) {};
	\draw[dashed,blue] (-1.5,4.5) -- (-1.5,-1.5) {};
	\draw[dotted,blue] (-1.5,1.5) -- (11.5,1.5) {};
	
	\draw[<->,red] (A1) .. controls (0,1) and (2,1) .. (B1) {};

	\draw[fill=red,opacity=0.2]
		(-1,-1) --
			(-1,1) --
				(3,1) --
					(3,-1) --
						(-1,-1) {};
	\draw[dashed,red] (-1,-1) -- (-1,1) {};
	\draw[dashed,red] (-1,1) -- (3,1) {};
	\draw[dashed,red] (3,1) -- (3,-1) {};
	\draw[dashed,red] (3,-1) -- (-1,-1) {};
	\node (nodeA) at (-0.7,0.7) [red] {$A$};

	\draw[<->,red] (A2) .. controls (0,4) and (2,4) .. (B2) {};

	\draw[fill=red,opacity=0.2]
		(-1,2) --
			(-1,4) --
				(3,4) --
					(3,2) --
						(-1,2) {};
	\draw[dashed,red] (-1,2) -- (-1,4) {};
	\draw[dashed,red] (-1,4) -- (3,4) {};
	\draw[dashed,red] (3,4) -- (3,2) {};
	\draw[dashed,red] (3,2) -- (-1,2) {};
	\node (nodeA) at (-0.7,3.7) [red] {$A$};
\end{tikzpicture}
}
\subfigure[Plugging of $\text{\textnormal{Ctr}}^{\phi}_{\psi}$ and $B$]{
\begin{tikzpicture}[x=0.72cm,y=0.72cm,z=-0.75cm,scale=0.8]
	\node (A11) at (0,0,0) {$1_{1,1}$};
	\node (B11) at (2,0,0) {$2_{1,1}$};
	\node (C11) at (4,0,0) {$3_{1,1}$};
	\node (D11) at (6,0,0) {$4_{1,1}$};
	\node (E11) at (8,0,0) {$5_{1,1}$};
	\node (F11) at (10,0,0) {$6_{1,1}$};
	\node (A21) at (0,3,0) {$1_{2,1}$};
	\node (B21) at (2,3,0) {$2_{2,1}$};
	\node (C21) at (4,3,0) {$3_{2,1}$};
	\node (D21) at (6,3,0) {$4_{2,1}$};
	\node (E21) at (8,3,0) {$5_{2,1}$};
	\node (F21) at (10,3,0) {$6_{2,1}$};

	\node[opacity=0.6] (A12) at (0,0,-4) {$1_{1,2}$};
	\node[opacity=0.6] (B12) at (2,0,-4) {$2_{1,2}$};
	\node[opacity=0.6] (C12) at (4,0,-4) {$3_{1,2}$};
	\node[opacity=0.6] (D12) at (6,0,-4) {$4_{1,2}$};
	\node[opacity=0.6] (E12) at (8,0,-4) {$5_{1,2}$};
	\node[opacity=0.6] (F12) at (10,0,-4) {$6_{1,2}$};
	\node[opacity=0.6] (A22) at (0,3,-4) {$1_{2,2}$};
	\node[opacity=0.6] (B22) at (2,3,-4) {$2_{2,2}$};
	\node[opacity=0.6] (C22) at (4,3,-4) {$3_{2,2}$};
	\node[opacity=0.6] (D22) at (6,3,-4) {$4_{2,2}$};
	\node[opacity=0.6] (E22) at (8,3,-4) {$5_{2,2}$};
	\node[opacity=0.6] (F22) at (10,3,-4) {$6_{2,2}$};

	\draw[<->,blue] (A21) -- (E11) {};
	\draw[<->,blue] (B21) -- (F11) {};
	\draw[<->,blue] (A11) .. controls (0,-1,0) and (4,-1,0) .. (C11) {};
	\draw[<->,blue] (B11) .. controls (2,-1,0) and (6,-1,0) .. (D11) {};

	\draw[dashed,<->,blue] (A22) -- (E12) {};
	\draw[dashed,<->,blue] (B22) -- (F12) {};
	\draw[dashed,<->,blue] (A12) .. controls (0,-1,-4) and (4,-1,-4) .. (C12) {};
	\draw[dashed,<->,blue] (B12) .. controls (2,-1,-4) and (6,-1,-4) .. (D12) {};
	
	\draw[opacity=0.2,fill=blue] 
		(-1.5,-1.5,0) -- 
			(11.5,-1.5,0) --
				(11.5,4.5,0) --
					(-1.5,4.5,0) --
						(-1.5,-1.5,0) {};
	\draw[dashed,blue] (-1.5,-1.5,0) -- (11.5,-1.5,0) {};
	\draw[dashed,blue] (11.5,-1.5,0) -- (11.5,4.5,0) {};
	\draw[dashed,blue] (11.5,4.5,0) -- (-1.5,4.5,0) {};
	\draw[dashed,blue] (-1.5,4.5,0) -- (-1.5,-1.5,0) {};
	\draw[dotted,blue] (-1.5,1.5,0) -- (11.5,1.5,0) {};

	\draw[opacity=0.2,fill=blue] 
		(-1.5,-1.5,-4) -- 
			(11.5,-1.5,-4) --
				(11.5,4.5,-4) --
					(-1.5,4.5,-4) --
						(-1.5,-1.5,-4) {};
	\draw[dashed,blue] (-1.5,-1.5,-4) -- (11.5,-1.5,-4) {};
	\draw[dashed,blue] (11.5,-1.5,-4) -- (11.5,4.5,-4) {};
	\draw[dashed,blue] (11.5,4.5,-4) -- (-1.5,4.5,-4) {};
	\draw[dashed,blue] (-1.5,4.5,-4) -- (-1.5,-1.5,-4) {};
	\draw[dotted,blue] (-1.5,1.5,-4) -- (11.5,1.5,-4) {};

	\definecolor{darkgreen}{rgb}{0.15,0.6,0.15};
	\draw[<->,darkgreen] (A11) .. controls (0,0,-2) and (2,0,-2) .. (B11) {};
	\draw[dashed,<->,darkgreen] (A12) .. controls (-1,0,-6) and (1,0,-6) .. (A12) {};
	\draw[dashed,<->,darkgreen] (B12) .. controls (1,0,-6) and (3,0,-6) .. (B12) {};

	\draw[fill=darkgreen,opacity=0.2]
		(-1,0,2) --
			(-1,0,-6) --
				(3,0,-6) --
					(3,0,2) --
						(-1,0,2) {};
	\draw[dashed,darkgreen] (-1,0,2) -- (-1,0,-6) {};
	\draw[dashed,darkgreen] (-1,0,-6) -- (3,0,-6) {};
	\draw[dashed,darkgreen] (3,0,-6) -- (3,0,2) {};
	\draw[dashed,darkgreen] (3,0,2) -- (-1,0,2) {};	

	\draw[<->,darkgreen] (A21) .. controls (0,3,-2) and (2,3,-2) .. (B21) {};
	\draw[dashed,<->,darkgreen] (A22) .. controls (-1,3,-6) and (1,3,-6) .. (A22) {};
	\draw[dashed,<->,darkgreen] (B22) .. controls (1,3,-6) and (3,3,-6) .. (B22) {};

	\draw[fill=darkgreen,opacity=0.2]
		(-1,3,2) --
			(-1,3,-6) --
				(3,3,-6) --
					(3,3,2) --
						(-1,3,2) {};
	\draw[dashed,darkgreen] (-1,3,2) -- (-1,3,-6) {};
	\draw[dashed,darkgreen] (-1,3,-6) -- (3,3,-6) {};
	\draw[dashed,darkgreen] (3,3,-6) -- (3,3,2) {};
	\draw[dashed,darkgreen] (3,3,2) -- (-1,3,2) {};
	
\end{tikzpicture}
}
\caption{Plugging of $\text{\textnormal{Ctr}}^{\phi}_{\psi}$ with the two graphs $A$ and $B$}\label{banchementcontrA}
\end{figure}

\begin{figure}
\centering
\subfigure[Result of the execution of $\text{\textnormal{Ctr}}^{\phi}_{\psi}$ and $A$]{
\begin{tikzpicture}[x=0.72cm,y=0.72cm]
	\node (C1) at (4,0) {$3_{1}$};
	\node (D1) at (6,0) {$4_{1}$};
	\node (E1) at (8,0) {$5_{1}$};
	\node (F1) at (10,0) {$6_{1}$};
	\node (C2) at (4,2.5) {$3_{2}$};
	\node (D2) at (6,2.5) {$4_{2}$};
	\node (E2) at (8,2.5) {$5_{2}$};
	\node (F2) at (10,2.5) {$6_{2}$};
	
	\draw[<->,violet] (C1) .. controls (4,1) and (6,1) .. (D1) {};
	\draw[<->,violet] (E1) .. controls (8,1) and (10,1) .. (F1) {};

	\draw[fill=violet,opacity=0.2]
		(3,-1) --
			(3,3.5) --
				(11,3.5) --
					(11,-1) --
						(3,-1) {};
	\draw[dashed,violet] (3,-1) -- (3,3.5) {};
	\draw[dashed,violet] (3,3.5) -- (11,3.5) {};
	\draw[dashed,violet] (11,3.5) -- (11,-1) {};
	\draw[dashed,violet] (11,-1) -- (3,-1) {};
	\draw[dotted,violet] (3,1.5) -- (11,1.5) {};
\end{tikzpicture}
}\\
\subfigure[The graph of $\de{\phi(a)\otimes \psi(a)}$]{
\begin{tikzpicture}[x=0.72cm,y=0.72cm]
	\node (C1) at (4,0) {$3_{1}$};
	\node (D1) at (6,0) {$4_{1}$};
	\node (E1) at (8,0) {$5_{1}$};
	\node (F1) at (10,0) {$6_{1}$};
	
	\draw[<->,red] (C1) .. controls (4,1) and (6,1) .. (D1) {};
	\draw[<->,red] (E1) .. controls (8,1) and (10,1) .. (F1) {};

	\draw[fill=red,opacity=0.2]
		(3,-1) --
			(3,1.5) --
				(11,1.5) --
					(11,-1) --
						(3,-1) {};
	\draw[dashed,red] (3,-1) -- (3,1.5) {};
	\draw[dashed,red] (3,1.5) -- (11,1.5) {};
	\draw[dashed,red] (11,1.5) -- (11,-1) {};
	\draw[dashed,red] (11,-1) -- (3,-1) {};
\end{tikzpicture}
}
\caption{Graphs of the projects $\de{Ctr}^{\phi}_{\psi}\plug\de{a}$ and $\de{\phi(a)\otimes\psi(a)}$}\label{execcontrA}
\end{figure}

\begin{figure}
\centering
\subfigure[Result of the execution of $\text{\textnormal{Ctr}}^{\phi}_{\psi}$ and $B$]{
\begin{tikzpicture}[x=0.72cm,y=0.72cm]
	\definecolor{bluegreen}{rgb}{0.1,0.5,0.6}
	
	\node (C1) at (4,0) {$3_{1}$};
	\node (D1) at (6,0) {$4_{1}$};
	\node (E1) at (8,0) {$5_{1}$};
	\node (F1) at (10,0) {$6_{1}$};
	\node (C2) at (4,2.5) {$3_{2}$};
	\node (D2) at (6,2.5) {$4_{2}$};
	\node (E2) at (8,2.5) {$5_{2}$};
	\node (F2) at (10,2.5) {$6_{2}$};
	\node (C3) at (4,4.5) {$3_{3}$};
	\node (D3) at (6,4.5) {$4_{3}$};
	\node (E3) at (8,4.5) {$5_{3}$};
	\node (F3) at (10,4.5) {$6_{3}$};
	\node (C4) at (4,7) {$3_{4}$};
	\node (D4) at (6,7) {$4_{4}$};
	\node (E4) at (8,7) {$5_{4}$};
	\node (F4) at (10,7) {$6_{4}$};
	
	\draw[<->,bluegreen] (C1) .. controls (4,1) and (6,1) .. (D1) {};
	\draw[<->,bluegreen] (E1) .. controls (8,1) and (10,1) .. (F1) {};
	\draw[<->,bluegreen] (C3) .. controls (3,5.5) and (5,5.5) .. (C3) {};	
	\draw[<->,bluegreen] (D3) .. controls (5,5.5) and (7,5.5) .. (D3) {};
	\draw[<->,bluegreen] (E3) .. controls (7,5.5) and (9,5.5) .. (E3) {};
	\draw[<->,bluegreen] (F3) .. controls (9,5.5) and (11,5.5) .. (F3) {};

	\draw[fill=bluegreen,opacity=0.2]
		(3,-1) --
			(3,8) --
				(11,8) --
					(11,-1) --
						(3,-1) {};
	\draw[dashed,bluegreen] (3,-1) -- (3,8) {};
	\draw[dashed,bluegreen] (3,8) -- (11,8) {};
	\draw[dashed,bluegreen] (11,8) -- (11,-1) {};
	\draw[dashed,bluegreen] (11,-1) -- (3,-1) {};
	\draw[dashed,bluegreen] (3,1.5) -- (11,1.5) {};
	\draw[dashed,bluegreen] (3,3.5) -- (11,3.5) {};
	\draw[dashed,bluegreen] (3,6) -- (11,6) {};
\end{tikzpicture}
}\\
\subfigure[Graph of the project $\de{\phi(b)\otimes\psi(b)}$]{
\begin{tikzpicture}[x=0.72cm,y=0.72cm]
	\definecolor{darkgreen}{rgb}{0.15,0.6,0.15};

	\node (C1) at (4,0) {$3_{1}$};
	\node (D1) at (6,0) {$4_{1}$};
	\node (E1) at (8,0) {$5_{1}$};
	\node (F1) at (10,0) {$6_{1}$};
	\node (C2) at (4,2.5) {$3_{2}$};
	\node (D2) at (6,2.5) {$4_{2}$};
	\node (E2) at (8,2.5) {$5_{2}$};
	\node (F2) at (10,2.5) {$6_{2}$};
	\node (C3) at (4,5) {$3_{3}$};
	\node (D3) at (6,5) {$4_{3}$};
	\node (E3) at (8,5) {$5_{3}$};
	\node (F3) at (10,5) {$6_{3}$};
	\node (C4) at (4,7.5) {$3_{4}$};
	\node (D4) at (6,7.5) {$4_{4}$};
	\node (E4) at (8,7.5) {$5_{4}$};
	\node (F4) at (10,7.5) {$6_{4}$};
	
	\draw[<->,darkgreen] (C1) .. controls (4,1) and (6,1) .. (D1) {};
	\draw[<->,darkgreen] (E1) .. controls (8,1) and (10,1) .. (F1) {};
	\draw[<->,darkgreen] (C2) .. controls (4,3.5) and (6,3.5) .. (D2) {};
	\draw[<->,darkgreen] (E2) .. controls (7,3.5) and (9,3.5) .. (E2) {};
	\draw[<->,darkgreen] (F2) .. controls (9,3.5) and (11,3.5) .. (F2) {};
	\draw[<->,darkgreen] (C3) .. controls (3,6) and (5,6) .. (C3) {};
	\draw[<->,darkgreen] (D3) .. controls (5,6) and (7,6) .. (D3) {};
	\draw[<->,darkgreen] (E3) .. controls (8,6) and (10,6) .. (F3) {};
	\draw[<->,darkgreen] (C4) .. controls (3,8.5) and (5,8.5) .. (C4) {};
	\draw[<->,darkgreen] (D4) .. controls (5,8.5) and (7,8.5) .. (D4) {};
	\draw[<->,darkgreen] (E4) .. controls (7,8.5) and (9,8.5) .. (E4) {};
	\draw[<->,darkgreen] (F4) .. controls (9,8.5) and (11,8.5) .. (F4) {};

	\draw[fill=darkgreen,opacity=0.2]
		(3,-1) --
			(3,9) --
				(11,9) --
					(11,-1) --
						(3,-1) {};
	\draw[dashed,darkgreen] (3,-1) -- (3,9) {};
	\draw[dashed,darkgreen] (3,9) -- (11,9) {};
	\draw[dashed,darkgreen] (11,9) -- (11,-1) {};
	\draw[dashed,darkgreen] (11,-1) -- (3,-1) {};
	\draw[dashed,darkgreen] (3,1.5) -- (11,1.5) {};
	\draw[dashed,darkgreen] (3,4) -- (11,4) {};
	\draw[dashed,darkgreen] (3,6.5) -- (11,6.5) {};
\end{tikzpicture}
}
\caption{Graphs of the projects $\de{Ctr}^{\phi}_{\psi}\plug\de{b}$ and $\de{\phi(b)\otimes\psi(b)}$}\label{execcontrB}
\end{figure}

We will use the following direct corollary of \autoref{ethtenscondtens}.
\begin{proposition}\label{propositionethsendcondsend}
If $E$ is a non-empty set of project sharing the same carrier $V^{E}$, $\cond{F}$ is a conduct and $\de{f}$ satisfies that $\forall\de{e}\in E$, $\de{f\plug e}\in\cond{F}$, then $\de{f}\in E^{\pol\pol}\multimap \cond{F}$.
\end{proposition}

This proposition insures us that if $\cond{A}$ is a conduct such that there exists a set $E$ of single-sliced projects with $\cond{A}=E^{\pol\pol}$, then the contraction project $\de{Ctr}^{\psi}_{\phi}$ belongs to the conduct $\cond{A\multimap\phi(A)\otimes \psi(A)}$.

We find here a geometrical explanation to the introduction of exponential connectives. Indeed, in order to use a contraction, we must be sure we are working with single-sliced graphs. We will therefore define, for all behaviour $\cond{A}$, a conduct $\cond{\oc A}$ generated by a set of single-sliced graphs.

One should notice that a conduct $\cond{\oc A}$ generated by a set of single-sliced projects cannot be a behaviour: the projects $(a,\emptyset)$ necessarily belong to the orthogonal of $\cond{\oc A}$. We will therefore introduce \emph{perennial conducts} as those conducts generated by a set of wager-free single-sliced projects. Dually, we introduce the \emph{co-perennial conducts} as the conducts that are the orthogonal of a perennial conduct.

But first, we will need a way to associate a wager-free single-sliced project to any wager-free project. The difficulty in doing so lies in the fact that one should not \enquote{loose} any information in the process, i.e.\ the single-sliced graph $\oc G$ obtained from the thick graph $G$ should contain a copy of $G$. A naive answer would be to say that a thick graph can be thought of as a (non-thick) graph. But a thick graph may have an arbitrarily large dialect, hence two different projects in a conduct of carrier $V$ are in general defined from graphs which have sets of vertices of different cardinality. One solution is to consider infinite sets of vertices and use the usual bijections between $\naturalN+\naturalN$ and $\naturalN$ to deal with the different sizes of dialects. This solution would lead to the exponential connectives usually considered in \goi and game semantics; this is not the choice made here. In particular, working with infinite graphs induces unwanted divergences in the measurement between graphs (this is discussed in the introduction of the author's recent paper on exponentials for full linear logic \cite{seiller-goif}). 

In order to define a different perennisation that do not induce such unwanted divergences, we will use the generalisation of graphs introduced in an earlier paper to deal with second order quantification, namely \emph{graphings}, and introduce the natural notion of \emph{thick graphing} to mimic the generalisation of graphs to thick graphs. Graphings are graphs generalised so that the vertices are no longer points but measurable subsets of a fixed measure space $\measure{X}$. The intuition behind our perennisation (to be defined in the next section) is that the graph(ing) $\oc G$ is obtained from $G$ by \emph{splitting} every vertex into equal-measure subsets indexed by the dialect of $G$. In this way, each thick graph(ing) will be mapped to a non-thick graph(ing) with the same carrier. Although it is true that, as in the previous constructions of exponentials in \goi and game semantics, the objects are formally graphs over infinite sets of vertices, the fact that we can restrict to finite-measure spaces allows us to avoid divergences and, as discussed in the introduction, work with potentially infinite objects rather that actually infinite ones.

\section{Graphings and Exponentials}\label{sec_Graphings}

\subsection{Graphings}

In a previous paper on interaction graphs \cite{seiller-goig}, the author introduced a generalisation of graphs to which the previously described results can extended. This generalisation allows, among other things, for the definition of second order quantification. The main purpose of this generalisation is that a vertex can always be cut in an arbitrary (finite) number of sub-vertices, with the idea that these sub-vertices are smaller (hence vertices have a \emph{size}) and form a partition of the initial vertex (where two sub-vertices have the same size). This particular feature of graphings will be used here to define exponential connectives, i.e.\ define an operation -- called perennisation -- that given any graphing $G$ produces a slice-free graphing $\oc G$.

We stress the fact that these notions could be introduced and dealt with combinatorially, but we choose to use measure-theoretic notions in order to ease the intuitions and some proofs. In fact, a \emph{graphing} -- the notion which is introduced as a generalisation of the notion of graph -- can be thought of and used as a graph. Another important feature of the construction is the fact that it depends on a \emph{microcosm} -- a monoid of measurable maps -- which somehow describes which computational principles are allowed in the model \cite{seiller-lcc14}. For technical reasons discussed in an earlier paper \cite{seiller-goig}, these measurable maps should be non-singular transformations, i.e.\ maps $f:\measure{X}\rightarrow \measure{X}$ which preserves the sets of null measure: $\lambda(f(A))=0$ if and only if $\lambda(A)=0$, and measurable-preserving, the image of measurable sets under $f$ are measurable sets.

\begin{definition}
Let $\measure{X}=(X,\mathcal{B},\lambda)$ be a measured space. We denote by $\measurable{X}$ the set of measurable-preserving non-singular transformations $X\rightarrow X$. A \emph{microcosm} of the measured space $\measure{X}$ is a subset $\mathfrak{m}$ of $\measurable{X}$ which is closed under composition and contains the identity.
\end{definition}

As in the graph construction described above, we will consider a notion of graphing depending on a \emph{weight-monoid} $\Omega$, i.e.\ a monoid $(\Omega,\cdot,1)$ which contains the possible weights of the edges. 

\begin{definition}[Graphings]
Let $\mathfrak{m}$ be a microcosm of a measured space $(X,\mathcal{B},\lambda)$ and $V^{F}$ a measurable subset of $X$. A \emph{$\Omega$-weighted graphing in $\mathfrak{m}$} of carrier $V^{F}$ is a countable family $F=\{(\omega_{e}^{F},\phi_{e}^{F}: S_{e}^{F}\rightarrow T_{e}^{F})\}_{e\in E^{F}}$, where, for all $e\in E^{F}$ (the set of \emph{edges}):
\begin{itemize}[noitemsep,nolistsep]
\item $\omega_{e}^{F}$ is an element of $\Omega$, the \emph{weight} of the edge $e$;
\item $S_{e}^{F}\subset V^{F}$ is a measurable set, the \emph{source} of the edge $e$;
\item $T_{e}^{F}=\phi_{e}^{F}(S_{e}^{F})\subset V^{F}$ is a measurable set, the \emph{target} of the edge $e$;
\item $\phi_{e}^{F}$ is the restriction of an element of $\mathfrak{m}$ to $S_{e}^{F}$, the \emph{realisation} of the edge $e$.
\end{itemize}
\end{definition}

It is natural, as we are working with measure-theoretic notions, to identify two graphings that differ only on a set of null measure. This leads to the definition of an equivalence relation between graphings: that of \emph{almost everywhere equality}. Moreover, since we want vertices to be \emph{decomposable} into any finite number of parts, we want to identify a graphing $G$ with the graphing $G'$ obtained by replacing an edge $e\in E^{F}$ by a finite family of edges $e_{i}\in G'$ ($i=1,\dots,n$) subject to the conditions:
\begin{itemize}[noitemsep,nolistsep]
\item the family $\{S^{G'}_{e_{i}}\}_{i=1}^{n}$ (resp. $\{T^{G'}_{e_{i}}\}_{i=1}^{n}$) is a partition of $S_{e}^{G}$ (resp. $T_{e}^{G}$);
\item for all $i=1,\dots,n$, $\phi_{e_{i}}^{G'}$ is the restriction of $\phi_{e}^{G}$ on $S^{G'}_{e_{i}}$.
\end{itemize}
Such a graphing $G'$ is an example of a \emph{refinement of $G$}, and one can easily generalise the previous conditions to define a general notion of refinement of graphings. \autoref{refinements} gives the most simple example of refinement. To be a bit more precise, we define, in order to ease the proofs, a notion of refinement \emph{up to almost everywhere equality}. We then define a second equivalence relation on graphings by saying that two graphings are equivalent if and only if they have a common refinement (up to almost everywhere equality). 

\begin{figure}
\centering
\begin{tikzpicture}[x=0.7cm,y=0.7cm]
	\draw[-] (0,0) -- (2,0) node [below,very near start] {\scriptsize{$[0,2]$}};
	\draw[-] (3,0) -- (5,0) node [below,very near end] {\scriptsize{$[3,5]$}};
	\node (1) at (1,0) {};
	\node (4) at (4,0) {};
	
	\draw[->,red] (1) .. controls (1,1.5) and (4,1.5) .. (4) node [midway,above] {\scriptsize{$x\mapsto 5-x$}};
	
	\draw[-] (7.5,0) -- (8.5,0) node [below,very near start] {\scriptsize{$[0,1]$}};
	\draw[-] (9,0) -- (10,0) node [below,very near start] {\scriptsize{$[1,2]$}};
	\draw[-] (10.5,0) -- (11.5,0) node [below,very near end] {\scriptsize{$[3,4]$}};
	\draw[-] (12,0) -- (13,0) node [below,very near end] {\scriptsize{$[4,5]$}};
	\node (8) at (8,0) {};
	\node (95) at (9.5,0) {};
	\node (11) at (11,0) {};
	\node (125) at (12.5,0) {};
	
	\draw[->,red] (8) .. controls (8,2) and (12.5,2) .. (125) node [midway,above] {\scriptsize{$x\mapsto 5-x$}};
	\draw[->,red] (95) .. controls (9.5,1) and (11,1) .. (11) node [midway,above] {\scriptsize{$x\mapsto 5-x$}};
	
\end{tikzpicture}
\caption{A graphing and one of its refinements}\label{refinements}
\end{figure}

The objects under study are thus equivalence classes of graphings modulo this equivalence relation. Most of the technical results on graphings contained in our previous paper \cite{seiller-goig} aim at showing that these objects can actually be manipulated as graphs: one can define paths and cycles and these notions are coherent with the quotient by the equivalence relation just mentioned. Indeed, the notions of paths and cycles in a graphings are quite natural, and from two graphings $F,G$ in a microcosm $\mathfrak{m}$ one can define its execution $F\plug G$ which is again a graphing in $\mathfrak{m}$\footnote{As a consequence, a microcosm is a closed world for the execution of programs.}. A more involved argument then shows that the trefoil property holds for a family of measurements $\meas{\cdot,\cdot}$, where $m:\Omega\rightarrow\realposN\cup\{\infty\}$ is any measurable map. These results are obtained as a generalisation of constructions considered in the author's thesis\footnote{In the cited work, the results were stated in the particular case of the microcosm of measure-preserving maps on the real line.}.

\begin{theorem}\label{mainthmgraphings}
Let $\Omega$ be a monoid, $\mathfrak{m}$ a microcosm and $m:\Omega\rightarrow\realposN\cup\{\infty\}$ be a measurable map. The set of $\Omega$-weighted graphings in $\mathfrak{m}$ yields a model, denoted by $\vaguemodel{\Omega}{m}{m}$, of multiplicative-additive linear logic whose orthogonality relation depends on $m$.
\end{theorem}

\subsection{Sliced Thick Graphings}

The sliced graphings are obtained from graphings in the same way we defined sliced thick graphs from directed weighted graphs: we consider formal weighted sums $F=\sum_{i\in I^{F}} \alpha^{F}_{i}F_{i}$ where the $F_{i}$ are graphings of carrier $V^{F_{i}}$. We define the \emph{carrier of $F$} as the measurable set $\cup_{i\in I^{F}} V^{F_{i}}$. The various constructions are then extended as explained above:
\begin{eqnarray*}
(\sum_{i\in I^{F}} \alpha^{F}_{i}F_{i})\plug(\sum_{i\in I^{G}} \alpha^{G}_{i}G_{i})&=&\sum_{(i,j)\in I^{F}\times I^{G}} \alpha_{i}^{F}\alpha^{G}_{j} F_{i}\plug G_{j}\\
\meas{\sum_{i\in I^{F}} \alpha^{F}_{i}F_{i},\sum_{i\in I^{G}} \alpha^{G}_{i}G_{i}}&=&\sum_{(i,j)\in I^{F}\times I^{G}} \alpha_{i}^{F}\alpha^{G}_{j} \meas{F_{i}\plug G_{j}}
\end{eqnarray*}
The trefoil property and the adjunction are then easily obtained through the same computations as in the proofs of \autoref{numtrefoil} and \autoref{numadjunct}.

We will now consider the most general notion of \emph{thick graphing} one can define. As it was the case in the setting of graphs, a thick graphing is a graphing whose carrier has the form $V\times D$. The main difference between graphings and thick graphings really comes from the way two such objects interact.

\begin{definition}
Let $(X,\mathcal{B},\lambda)$ be a measured space and $(D,\mathcal{D},\mu)$ a probability space (a measured space such that $\mu(D)=1$). A thick graphing of carrier $V\in\mathcal{B}$ and dialect $D$ is a graphing on $X\times D$ of carrier $V\times D$.
\end{definition}

\begin{definition}[Dialectal Interaction]
Let $(X,\mathcal{B},\lambda)$ be a measured space and $(D,\mathcal{D},\mu)$, $(E,\mathcal{E},\nu)$ two probability spaces. Let $F,G$ be thick graphings of respective carriers $V^{F},V^{G}\in\mathcal{B}$ and respective dialects $D,E$. We define the graphings $F^{\dagger_{E}}$ and $G^{\ddagger_{D}}$ as the graphings of respective carriers $V^{F},V^{G}$ and dialects $E\times F$: 
\begin{eqnarray*}
F^{\dagger_{E}}&=&\{(\omega^{F}_{e},\phi^{F}_{e}\times\text{Id}_{E}: S_{e}^{F}\times D\times E\rightarrow T_{e}^{F}\times D\times E)\}_{e\in E^{F}}\\
G^{\ddagger_{D}}&=&\{(\omega^{G}_{e},\text{Id}_{X}\times(\tau\circ(\phi^{G}_{e}\times\text{Id}_{D})\circ\tau^{-1}): S_{e}^{G}\times D\times E\rightarrow T_{e}^{G}\times D\times E)\}_{e\in E^{G}}
\end{eqnarray*}
where $\tau$ is the natural symmetry: $E\times D\rightarrow D\times E$.
\end{definition}


\begin{definition}[Execution]
Let $F,G$ be two thick graphings of respective dialects $D^{F},D^{G}$. Their execution is equal to $F^{\dagger_{D^{G}}}\plug G^{\ddagger_{D^{F}}}$.
\end{definition}

The notion of circuit-quantifying map can also be adapted to the setting of graphings, but the definition is much more involved than in the case of graphs. To avoid unnecessary complications, and since this was studied in detail in a previous paper, we refer the reader to this earlier work for a formal definition \cite{seiller-goig}.

\begin{definition}[Measurement]
Let $F,G$ be two thick graphings of respective dialects $D^{F},D^{G}$, and $m$ a (graphing) circuit-quantifying map. The corresponding measurement of the interaction between $F$ and $G$ is equal to $\meas{F^{\dagger_{D^{G}}},G^{\ddagger_{D^{F}}}}$.
\end{definition}

As in the setting of graphs, one can show that all the fundamental properties are preserved when we generalise from graphings to thick graphings.

\begin{proposition}
Let $F,G,H$ be thick graphings such that $V^{F}\cap V^{G}\cap V^{H}$ is of null measure. Then:
\begin{eqnarray*}
F\plug(G\plug H)&=&(F\plug G)\plug H\\
\meas{F,G\plug H}+\meas{G,H}&=&\meas{G,H\plug F}+\meas{H,F}
\end{eqnarray*}
\end{proposition}

In a similar way, the extension from thick graphings to sliced thick graphings should now be quite clear. One extends all operations by ``linearity'' to formal weighted sums of thick graphings, and one obtains, when $F,G,H$ are sliced thick graphings such that $V^{F}\cap V^{G}\cap V^{H}$ is of null measure:
\begin{eqnarray*}
F\plug(G\plug H)&=&(F\plug G)\plug H\\
\meas{F,G\plug H}+\unit{F}\meas{G,H}&=&\meas{G,H\plug F}+\unit{G}\meas{H,F}
\end{eqnarray*}

\subsection{Perennial and Co-perennial conducts}

Since we are working with sliced thick graphings, we can follow the constructions of multiplicative and additive connectives as they are studied in the author's second paper on interaction graphs \cite{seiller-goiadd} and which were quickly recalled in \autoref{recallmall}: replacing graphs by graphings in the sketched construction provides models for Multiplicative-Additive Linear Logic, as insured by the \emph{associativity of execution} \cite{seiller-goig} and the \emph{trefoil property} we just exposed. We nonetheless define formally the notion of project, and refer the reader to previous works on interaction graphs \cite{seiller-goiadd,seiller-goig} for a detailed study of the \MALL construction.

\begin{definition}[Projects]
A \emph{project} is a couple $\de{a}=(a,A)$ together with a support $V^{A}$ where:
\begin{itemize}[noitemsep,nolistsep]
\item $a\in\realN\cup\{\infty\}$ is called the wager;
\item $A$ is a sliced and thick weighted graphing of carrier $V^{A}$, of dialect $D^{A}$ a discrete probability space, and index $I^{A}$ a finite set.
\end{itemize}
\end{definition}

\begin{remark}
We made here the choice to stay close to the hyperfinite geometry of interaction defined by Girard \cite{goi5}. This is why we restrict to discrete probability spaces as dialects, a restriction that corresponds to the restriction to finite von Neumann algebras of type  $\text{I}$ as idioms in Girard's setting. However, the results of the preceding section about execution and measurement, and the definition of the family of circuit-quantifying maps do not rely on this hypothesis. It is therefore possible to consider a more general set of project where the dialects are continuous. It was proved during the revision of this paper that such a generalisation can be used to define more expressive exponential connectives than the one defined in this paper, namely the usual exponentials of linear logic \cite{seiller-goif} (recall that the exponentials defined here are the exponentials of Elementary Linear Logic). 
\end{remark}

A simple adaptation of the definitions recalled in \autoref{recallmall} allows one to define \emph{conducts}, \emph{behaviours}, and prove thick graphings variants of all exposed theorems and propositions. We will therefore here focus on what is new in this work, namely specific notions of conducts arising from the consideration of exponential connectives. Indeed, as explained at the end of \autoref{seccontraction}, we need to consider a particular kind of conducts which are the kind of conducts obtained from the application of the exponential modality to a conduct. Such conducts are unfortunately not behaviours, but they can be shown to satisfy some specific properties that will be important in order to provide models of (elementary) linear logic. Let us now define and study these types of conducts, called \emph{perennial conducts}, and their dual, \emph{co-perennial conducts}.

\begin{definition}[Perennialisation]
A Perennialisation is a function that associates a single-sliced weighted graphing to any sliced and thick weighted graphing.
\end{definition}

\begin{definition}[Exponentials]
Let $\cond{A}$ be a conduct, and $\Omega$ a perennialisation. We define the $\cond{\oc_{\Omega} A}$ as the bi-orthogonal closure of the following set of projects:
\begin{equation*}
\sharp_{\Omega}\cond{A}=\{\oc\de{a}=(0,\Omega(A))~|~\de{a}=(0,A)\in\cond{A}, I^{A}\cong\{1\}\}
\end{equation*}
The dual connective is of course defined as $\cond{\wn_{\Omega} A}=\cond{(\sharp_{\Omega} A^{\pol})^{\pol}}$.
\end{definition}

\begin{definition}
A conduct $\cond{A}$ is a \emph{perennial conduct} when there exists a set $A$ of projects such that:
\begin{enumerate}
\item $\cond{A}=A^{\pol\pol}$;
\item for all $\de{a}=(a,A)\in A$, $a=0$ and $A$ is a single-sliced graphing.
\end{enumerate}
A \emph{co-perennial} conduct is a conduct $\cond{B}=\cond{A}^{\pol}$ where $\cond{A}$ is a perennial conduct.
\end{definition}

\begin{proposition}\label{coperinflation}
A co-perennial conduct $\cond{B}$ satisfies the \emph{inflation property}: for all $\lambda\in\realN$, $\de{b}\in\cond{B}\Rightarrow \de{b+\lambda 0}_{V^{B}}\in\cond{B}$.
\end{proposition}

\begin{proof}
The conduct $\cond{A}=\cond{B}^{\pol}$ being perennial, there exists a set $A$ of single-sliced wager-free projects such that $\cond{A}=A^{\pol\pol}$. If $A$ is non-empty, the result is a direct consequence of \autoref{fellorth}. If $A$ is empty, then $\cond{B}=\cond{A}^{\pol}=A^{\pol}$ is the full behaviour $\cond{T}_{V^{B}}$ which obviously satisfies the inflation property.
\end{proof}

\begin{proposition}
A co-perennial conduct is non-empty.
\end{proposition}

\begin{proof}
Suppose that $\cond{A}^{\pol}$ is a co-perennial conduct of carrier $V^{A}$. Then there exists a set $A$ of single-sliced wager-free projects such that $\cond{A}=\cond{A}^{\pol\pol}$. If $A$ is empty, then $A^{\pol}=\cond{A}^{\pol}$ is the behaviour $\cond{T}_{V^{A}}$. If $\cond{A}$ is non-empty, then one can easily check that for all real number $\lambda\neq 0$, the project $\de{Dai}_{\lambda}=(\lambda,(V^{A},\emptyset))$ is an element of $A^{\pol}=\cond{A}^{\pol}$.
\end{proof}

\begin{corollary}
Let $\cond{A}$ be a perennial conduct. Then $\de{a}=(a,A)\in\cond{A} \Rightarrow a=0$.
\end{corollary}

\begin{proof}
Since $\cond{A}^{\pol}$ is co-perennial, it is a non-empty set of projects with the same carrier which satisfies the inflation property. The result is then obtained by applying \autoref{fellorth}.
\end{proof}

\begin{proposition}\label{coperdemon}
If $\cond{A}$ is a co-perennial conduct, then for all $a\neq 0$, the project $\de{Dai}_{a}=(a,(V^{A},\emptyset))$ is an element of $\cond{A}$.
\end{proposition}

\begin{proof}
We write $B$ the set of single-sliced wager-free projects such that $B^{\pol}=\cond{A}$. Then for all element $\de{b}\in\cond{B}$, we have that $\unit{B}=1$, from which we conclude that $\sca{b}{Dai_{\text{$a$}}}=a\unit{B}=a$. Thus $\de{Dai}_{a}\in B^{\pol}=\cond{A}$ for all $a\neq 0$.
\end{proof}

\begin{proposition}
The tensor product of perennial conducts is a perennial conduct.
\end{proposition}

\begin{proof}
Let $\cond{A,B}$ be perennial conducts. Then there exists two sets of single-sliced wager-free projects $E,F$ such that $\cond{A}=E^{\pol\pol}$ and $\cond{B}=F^{\pol\pol}$. Using \autoref{ethtenscondtens}, we know that $\cond{A\otimes B}=(E\odot F)^{\pol\pol}$. But, by definition, $E\odot F$ contains only projects of the form $\de{e}\otimes \de{f}$, where $\de{e,f}$ are single-sliced and wager-free. Thus $E\odot F$ contains only single-sliced wager-free projects and $\cond{A\otimes B}$ is therefore a perennial conduct.
\end{proof}

\begin{proposition}
If $\cond{A,B}$ are perennial conducts, then $\cond{A\oplus B}$ is a perennial conduct.
\end{proposition}

\begin{proof}
This is a consequence of \autoref{compintoplus}.
\end{proof}

\begin{proposition}
If $\cond{A}$ is a perennial conduct and $\cond{B}$ is a co-perennial conduct, then $\cond{A\multimap B}$ is a co-perennial conduct.
\end{proposition}

\begin{proof}
We recall that $\cond{A\multimap B}=(\cond{A}\otimes\cond{B}^{\pol})^{\pol}$. Since $\cond{A}$ and $\cond{B}^{\pol}$ are perennial conducts, $\cond{A}\otimes\cond{B}^{\pol}$ is a perennial conduct, and therefore $\cond{A\multimap B}$ is a co-perennial conduct. In particular, $\cond{A\multimap B}$ is non-empty and satisfies the inflation property.
\end{proof}

\begin{proposition}\label{tenseurperennecomp}
If $\cond{A}$ is a perennial conduct and $\cond{B}$ is a behaviour, then $\cond{A\otimes B}$ is a behaviour.
\end{proposition}

\begin{proof}
If $\cond{A}=\cond{0}_{V^{A}}$ with $\cond{B}=\cond{0}_{V^{B}}$, then $\cond{A\otimes B}=\cond{0}_{V^{A}\cup V^{B}}$ which is a behaviour.

Let $A$ be the set of single-sliced wager-free projects such that $\cond{A}=A^{\pol\pol}$. We have that $\cond{A\otimes B}=(A\odot \cond{B})^{\pol\pol}$ by \autoref{ethtenscondtens}. If $\cond{B}\neq\cond{0}_{V^{B}}$ and $A\neq 0$, then $A\odot \cond{B}$ is non-empty and contains only wager-free projects. Thus $\cond{(A\otimes B)^{\pol}}$ satisfies the inflation property by \autoref{fellorth}.

Now suppose there exists $\de{f}=(f,F)\in\cond{(A\otimes B)^{\pol}}$ such that $f\neq 0$. Then for all $\de{a}\in A$ and $\de{b}\in \cond{B}$, $\sca{f}{a\otimes b}\neq 0,\infty$. But, since $\de{a}$ is wager-free and $\unit{A}=1$, $\sca{f}{a\otimes b}=f\unit{B}+b\unit{F}+\meas{F,A\cup B}$. We can then define $\mu=(-\meas{F,A\cup B}-b\unit{F})/f-\unit{B}$. Since $\cond{B}$ is a behaviour, $\de{b+\mu 0}\in\cond{B}$. However:
\begin{eqnarray*}
\sca{f}{a\otimes (b+\mu 0)}&=&f(\unit{B}+\mu)+b\unit{F}+\meas{F,A\cup (B+\mu 0)}\\
&=&f(\unit{B}+\mu)+b\unit{F}+\meas{F,A\cup B}\\
&=&f(\unit{B}+(-\meas{F,A\cup B}-b\unit{F})/f-\unit{B})+b\unit{F}+\meas{F,A\cup B}\\
&=&-\meas{F,A\cup B}-b\unit{F}+b\unit{F}+\meas{F,A\cup B}\\
&=&0
\end{eqnarray*}
But this is a contradiction. Therefore the elements in $\cond{(A\otimes B)^{\pol}}$ are wager-free.

If $\cond{(A\otimes B)^{\pol}}$ is non-empty, it is a non-empty conduct containing only wager-free projects and satisfying the inflation property: it is therefore a (proper) behaviour.

The only case left to treat is when $\cond{(A\otimes B)^{\pol}}$ is empty, but then $\cond{A\otimes B}=\cond{T}_{V^{A}\cup V^{B}}$ is clearly a behaviour.
\end{proof}

\begin{corollary}
If $\cond{A}$ is perennial and $\cond{B}$ is a behaviour, then $\cond{A\multimap B}$ is a behaviour.
\end{corollary}

\begin{proof}
We recall that $\cond{A\multimap B}=(\cond{A\otimes B^{\pol}})^{\pol}$. Using the preceding proposition, the conduct $\cond{A\otimes B^{\pol}}$ is a behaviour since $\cond{A}$ is a perennial conduct and $\cond{B}^{\pol}$ is a behaviour. Thus $\cond{A\multimap B}$ is a behaviour since it is the orthogonal of a behaviour.
\end{proof}

\begin{corollary}
If $\cond{A}$ is a behaviour and $\cond{B}$ is a co-perennial conduct, then $\cond{A\multimap B}$ is a behaviour.
\end{corollary}

\begin{proof}
One just has to write $\cond{A\multimap B}=\cond{(A\otimes B^{\pol})^{\pol}}$. Since $\cond{A\otimes B^{\pol}}$ is the tensor product of a behaviour with a perennial conduct, it is a behaviour. The result then follows from the fact that the orthogonal of a behaviour is a behaviour.
\end{proof}

\begin{proposition}
The weakening (on the left) of perennial conducts holds.
\end{proposition}

\begin{proof}
Let $\cond{A,B}$ be conducts, and $\cond{N}$ be a perennial conduct. Choose $\de{f}\in\cond{A\multimap B}$. We will show that $\de{f}\otimes \de{0}_{V^{N}}$ is a project in $\cond{A\otimes N\multimap B}$. For this, we pick $\de{a}\in\cond{A}$ and $\de{n}\in\cond{N}$ (recall that $\de{n}$ is necessarily wager-free). Then for all $\de{b'}\in\cond{B^{\pol}}$,
\begin{eqnarray*}
\lefteqn{\sca{(f\otimes 0)\plug(a\otimes n)}{b'}}\\
&=&\sca{f\otimes 0}{(a\otimes n)\otimes b'}\\
&=&\sca{f\otimes 0}{(a\otimes b')\otimes n}\\
&=&\unit{F}(\unit{A}\unit{B'}n+\unit{N}\unit{A}b'+\unit{N}\unit{B'}a)+\unit{N}\unit{A}\unit{B'}f+\meas{F\cup 0,A\cup B'\cup N}\\
&=&\unit{F}(\unit{N}\unit{A}b'+\unit{N}\unit{B'}a)+\unit{N}\unit{A}\unit{B'}f+\meas{F\cup 0,A\cup B'\cup N}\\
&=&\unit{N}(\unit{F}(\unit{A}b'+\unit{B'}a)+\unit{A}\unit{B'}f)+\unit{N}\meas{F,A\cup B'}\\
&=&\unit{N}\sca{f}{a\otimes b'}
\end{eqnarray*}
Since $\unit{N}\neq 0$, $\sca{(f\otimes 0)\plug(a\otimes n)}{b'}\neq 0,\infty$ if and only if $\sca{f\plug a}{b'}\neq0,\infty$. Thus for all $\de{a\otimes n}\in\cond{A\odot N}$, $\de{(f\otimes 0)\plug (a\otimes n)}\in\cond{B}$. This shows that $\de{f\otimes 0}\in\cond{A\otimes N\multimap B}$ by \autoref{propositionethsendcondsend}.
\end{proof}

\section{Construction of an Exponential Connective on the Real Line}\label{sec_Exp}

We now consider the microcosm $\mathfrak{mi}$ of measure-inflating maps\footnote{A \emph{measure-inflating map} on the real line with Lebesgue measure $\lambda$ is a non-singular measurable-preserving transformation $\phi:\realN\rightarrow\realN$ such that there exists a positive real number $\mu_{\phi}$ with $\lambda(\phi^{-1}(A))=\mu_{\phi}\lambda(A)$. In other terms, $\phi$ \emph{transports the measure} $\lambda$ onto $\mu_{\phi}\lambda$.} on the real line endowed with Lebesgue measure, we fix $\Omega=]0,1]$ endowed with the usual multiplication and we choose any measurable map $m:\Omega\rightarrow \realposN\cup\{\infty\}$ such that $m(1)=\infty$. We showed in a previous work how this framework can interpret multiplicative-additive linear logic with second order quantification\footnote{We actually showed how one can interpret second-order multiplicative-additive linear logic in the model $\vaguemodel{\Omega}{aff}{m}$ where $\mathfrak{aff}\subsetneq\mathfrak{mi}$ is the microcosm of affine transformations on the real line. The result is however valid for any super-microcosm $\mathfrak{n}\supset\mathfrak{aff}$, hence for $\mathfrak{mi}$, since a graphing in $\mathfrak{aff}$ can be considered as a graphing in $\mathfrak{n}$ in a way that is coherent with execution, orthogonality, sums, etc.} \cite{seiller-goig}. We now show how to interpret elementary linear logic exponential connectives in the model $\vaguemodel{\Omega}{mi}{m}$ (defined in \autoref{mainthmgraphings}).


\subsection{A Construction of Exponentials}

We will begin by showing a technical result that will allow us to define measure preserving transformations from bijections of the set of integers. This result will be used to show that functorial promotion can be implemented for our exponential modality.

\begin{definition}
Let $\phi:\naturalN\rightarrow \naturalN$ be a bijection and $b$ an integer $\geqslant 2$. Then $\phi$ induces a  transformation $T^{b}_{\phi}:[0,1]\rightarrow[0,1]$ defined by $\sum_{i\geqslant 0} a_{k}b^{-k}\mapsto \sum_{i\geqslant 0} a_{\phi^{-1}(k)}b^{-k}$.
\end{definition}

\begin{remark}\label{remarkrepbase}
Suppose that $\sum_{i\geqslant 0} a_{i}b^{-i}$ and $\sum_{i\geqslant 0} a'_{i}b^{-i}$ are two distinct representations of a real number $r$. Let us fix $i_{0}$ to be the smallest integer such that $a_{i_{0}}\neq a'_{i_{0}}$. We first notice that the absolute value of the difference between these digits has to be equal to $1$: $\abs{a_{i_{0}}-a'_{i_{0}}}=1$. Indeed, if this was not the case, i.e.\ if $\abs{a_{i_{0}}-a'_{i_{0}}}\geqslant 2$, the distance between $\sum_{i\geqslant 0} a_{i}b^{-i}$ and $\sum_{i\geqslant 0} a'_{i}b^{-i}$ would be greater than $b^{-i_{0}}$, which contradicts the fact that both sums are equal to $r$. Let us now suppose, without loss of generality, that $a_{i_{0}}=a'_{i_{0}}+1$. Then $a_{j}=0$ for all $j>i_{0}$ since if there existed an integer $j>i_{0}$ such that $a_{j}\neq 0$, the distance between the sums $\sum_{i\geqslant 0} a_{i}b^{-i}$ and $\sum_{i\geqslant 0} a'_{i}b^{-i}$ would be greater than $b^{-j}$, which would again be a contradiction. Moreover, $a'_{j}=b-1$ for all $j>i_{0}$: if this was not the case, one could show in a similar way that the difference between the two sums would be strictly greater than $0$. We can deduce from this that only the reals with a finite representation in base $b$ have two distinct representations.

Since the set of such elements is of null measure, the transformation $T_{\phi}$ associated to a bijection $\phi$ of $\naturalN$ is well defined as we can define $T_{\phi}$ only on the set of reals that have a unique representation. We can however choose to deal with this in another way: choosing between the two possible representations, by excluding for instance the representations as sequences that are almost everywhere equal to zero. Then $T_{\phi}$ is defined at all points and bijective. We choose in the following to follow this second approach since it will allow to prove more easily that $T_{\phi}$ is measure-preserving. However, this choice is not relevant for the rest of the construction since both transformations are almost everywhere equal.
\end{remark}

\begin{lemma}
Let $T$ be a transformation of $[0,1]$ such that for all interval $[a,b]$, $\lambda(T([a,b]))=\lambda([a,b])$. Then $T$ is measure-preserving on $[0,1]$.
\end{lemma}

\begin{proof}
A classical result of measure theory states that if $T$ is a transformation of a measured space $(X,\mathcal{B},\lambda)$, that $\mathcal{B}$ is generated by $\mathcal{A}$, and that for all $A\in \mathcal{A}$, $\lambda(T(A))=\lambda(A)$, then $T$ preserves the measure $\lambda$ on $X$. Applying this result with $X=[0,1]$, and $\mathcal{A}$ as the set of intervals $[a,b]\subset[0,1]$, we obtain the result.
\end{proof}

\begin{lemma}
Let $T$ be a bijective transformation of $[0,1]$ that preserves the measure on all interval $I$ of the shape $[\sum_{k=1}^{p} a_{k}b^{-k},\sum_{k=1}^{p} a_{k}b^{-k}+b^{-p}]$. Then $T$ is measure-preserving on $[0,1]$.
\end{lemma}

\begin{proof}
Choose $[a,b]\subset[0,1]$. One can write $[a,b]$ as a union $\cup_{i=0}^{\infty}[a_{i},a_{i+1}]$, where for all $i\geqslant 0$, $a_{i+1}=a_{i}+b^{-k_{i}}$. We then obtain, using the hypotheses of the statement and the $\sigma$-additivity of the measure $\lambda$:
\begin{eqnarray*}
\lambda(T([a,b]))&=&\lambda(T(\cup_{i=0}^{\infty} [a_{i},a_{i+1}[))\\
			  &=&\lambda(\cup_{i=0}^{\infty} T([a_{i},a_{i+1}[))\\
			  &=&\sum_{i=0}^{\infty} \lambda(T([a_{i},a_{i+1}[))\\
			  &=&\sum_{i=0}^{\infty} \lambda([a_{i},a_{i+1}[)\\
			  &=&\lambda(\cup_{i=0}^{\infty} [a_{i},a_{i+1}[)\\
			  &=&\lambda([a,b])
\end{eqnarray*}
We now conclude by using the preceding lemma.
\end{proof}

\begin{theorem}
Let $\phi:\naturalN\rightarrow \naturalN$ be a bijection and $b\geqslant 2$ an integer. Then the transformation $T^{b}_{\phi}$ is measure-preserving.
\end{theorem}

\begin{proof}
We recall first that the transformation $T^{b}_{\phi}$ is indeed bijective (see \autoref{remarkrepbase}).

By using the preceding lemma, it suffices to show that $T^{b}_{\phi}$ preserves the measure on intervals of the shape $I=[a,a+b^{-k}]$ with $a=\sum_{i=0}^{k} a_{i}b^{-i}$. Let us define $N=\max\{\phi(i)~|~0\leqslant i\leqslant k\}$. We then write $[0,1]$ as the union of intervals $A_{i}=[i\times b^{-N},(i+1)\times b^{-N}]$ where $0\leqslant i\leqslant b^{N}-1$.

Then the image if $I$ by $T^{b}_{\phi}$ is equal to the union of the $A_{i}$ for $i\times b^{-N}=\sum_{i=0}^{N} x_{i}b^{-i}$, where $x_{\phi(j)}=a_{j}$ for all $0\leqslant j\leqslant k$. The number of such $A_{j}$ is equal to the number of sequences $\{0,\dots,b-1\}$ of length $N-k$, i.e.\ $b^{N-k}$. Since each $A_{j}$ has a measure equal to $b^{-N}$, the image of $I$ by $T^{b}_{\phi}$ is of measure $b^{-N}b^{N-k}=b^{-k}$, which is equal to the measure of $I$ since $\lambda(I)=b^{-k}$.
\end{proof}

\begin{remark}
The preceding theorem can be easily generalised to any\footnote{i.e.\ not only those bijections that can be obtained by compositions of bijections $\naturalN+\naturalN\rightarrow\naturalN$.} bijection $\naturalN+\dots+\naturalN\rightarrow\naturalN$ (the domain being the disjoint union of $k$ copies of $\naturalN$, $k\in\naturalN$), in the sense that it induces a measure-preserving bijection from $[0,1]^{k}$ onto $[0,1]$. The particular case $\naturalN+\naturalN\rightarrow\naturalN$, $(n,i)\mapsto 2n+i$ defines the well-known measure-preserving bijection between the unit square and the interval $[0,1]$:
\begin{equation*}
(\sum_{i\geqslant 0} a_{i}2^{-i},\sum_{i\geqslant 0} b_{i}2^{-i})\mapsto \sum_{i\geqslant 0} a_{2i}2^{-2i}+b_{2i+1}2^{-2i-1}
\end{equation*}
\end{remark}

Let us now define the bijection:
\begin{equation*}
\psi: \naturalN+\naturalN+\naturalN\rightarrow\naturalN,~~~ (x,i)\mapsto 3x+i
\end{equation*}
We also define the injections $\iota_{i}$ ($i=0,1,2$):
\begin{equation*}
\iota_{i}: \naturalN\rightarrow \naturalN+\naturalN+\naturalN,~~~ x\mapsto (x,i)
\end{equation*}
We will denote by $\psi_{i}$ the composite $\psi\circ\iota_{i}:\naturalN\rightarrow \naturalN$.

\begin{definition}
Let $A\subset \naturalN+\naturalN+\naturalN$ be a finite set. We write $A$ as $A_{0}+A_{1}+A_{2}$, and define, for $i=0,1,2$, $n_{i}$ to be the cardinality of $A_{i}$ if $A_{i}\neq\emptyset$ and $n_{i}=1$ otherwise. We then define a partition of $[0,1]$, denoted by $\mathcal{P}_{A}=\{P_{A}^{i_{1},i_{2},i_{3}}~|~\forall k\in\{0,1,2\}, 0\leqslant i_{k}\leqslant n_{i}-1\}$, by:
\begin{equation*}
P_{A}^{i_{1},i_{2},i_{3}}=\{\sum_{j\geqslant 1} a_{j} 2^{-j}~|~ \forall k\in\{0,1,2\}, \frac{i_{k}}{n_{k}}\leqslant \sum_{j\geqslant 1} a_{\psi_{k}(j)}2^{-j}\leqslant \frac{i_{k}+1}{n_{k}}\}
\end{equation*}
When $A_{k}$ is empty or of cardinality $1$, we will not write the corresponding $i_{k}$ in the triple $(i_{1},i_{2},i_{3})$ since it is necessarily equal to $0$.
\end{definition}

\begin{proposition}
Let us keep the notations of the preceding proposition and let $X=P_{A}^{i_{1},i_{2},i_{3}}$ and $Y=P_{A}^{j_{1},j_{2},j_{3}}$ be two elements of the partition $\mathcal{P}_{A}$. For all $x=\sum_{l\geqslant 1} a_{l}2^{-l}$, we define $T_{i_{1},i_{2},i_{3}}^{j_{1},j_{2},j_{3}}(x)=\sum_{l\geqslant 1} b_{l}2^{-l}$ where the sequence $(b_{i})$ is defined by:
\begin{equation*}
\forall k\in\{0,1,2\}, \sum_{l\geqslant 1} b_{\psi_{k}(l)}2^{-l}=\sum_{l\geqslant 1} a_{\psi_{k}(l)}2^{-l}+j_{k}-i_{k}
\end{equation*}
Then $T_{i_{1},i_{2},i_{3}}^{j_{1},j_{2},j_{3}}: X\rightarrow Y$ is a measure-preserving bijection.
\end{proposition}

\begin{proof}
For $k=0,1,2$, we will denote by $(m^{k}_{j})$ the sequence such that $j_{k}-i_{k}=\sum_{l\geqslant 1} m^{k}_{l}2^{-l}$. We can define the real number $t=\sum_{l\geqslant 1} \sum_{k=0,1,2} m^{k}_{l}2^{-3j+k}$. It is then sufficient to check that $T_{i_{1},i_{2},i_{3}}^{j_{1},j_{2},j_{3}}(x)=x+t$. Since $T_{i_{1},i_{2},i_{3}}^{j_{1},j_{2},j_{3}}$ is a translation translation, it is a measure-preserving bijection.
\end{proof}

\begin{definition}
Let $A\subset \naturalN$ be a finite set endowed with the normalised -- i.e.\ such that $A$ has measure $1$ -- counting measure, and $X\in\mathcal{B}(\realN\times A)$ be a measurable set. We define the measurable subset $\embed{A}\subset\realN\times [0,1]$:
\begin{equation*}
\embed{A}=\{(x,y)~|~ \exists z\in A, (x,z)\in X, y\in P_{A}^{z}\}
\end{equation*}
We will write $\mathcal{P}^{-1}_{A}: [0,1]\rightarrow A$ the map that associates to each $x$ the element $z\in A$ such that $x\in P_{A}^{z}$.
\end{definition}

\begin{proposition}
Let $D^{A}\subset\naturalN$ be a finite set endowed with the normalised counting measure $\mu$ (i.e.\ such that $\mu(A)=1$), $S,T\in\mathcal{B}(\realN\times D^{A})$ be measurable sets, and $\phi: S\rightarrow T$ a measure-preserving transformation. We define $\embed{\phi}: \embed{S}\rightarrow \embed{T}$ by:
\begin{equation*}
\embed{\phi}: (x,y)\mapsto (x',y') ~~~ \phi(x,\mathcal{P}^{-1}_{A}(y))=(x',z),~~ y'=T_{\mathcal{P}^{-1}_{A}(y)}^{z}
\end{equation*}
Then $\embed{\phi}$ is a measure-preserving bijection.
\end{proposition}

\begin{proof}
For all $(a,b)\in D^{A}$ we define the set $S_{a,b}=X\cap\realN\times\{a\}\cap\phi^{-1}(Y\cap\realN\times\{b\})$. The family $(S_{a,b})_{a,b\in D^{A}}$ is a partition of $S$, and the family $(\embed{S_{a,b}})_{a,b\in D^{A}}$ is a partition of $\oc A$. The restriction of $\embed{\phi}$ to $\embed{S_{a,b}}$ can then be defined as the composite $T_{a}\circ\phi_{1}$ with:
\begin{eqnarray*}
\phi_{1}&=&(\pi_{1}\circ\phi)\times\text{Id}\\
T_{a}&=&\text{Id}\times T_{a}^{b}
\end{eqnarray*}
Since the product (resp. the composition) of measure preserving bijections is a measure preserving bijection, the restriction of $\embed{\phi}$ to $X_{a}$ is a measure preserving bijection. Moreover, it is clear that the image of $\embed{S}$ by $\embed{\phi}$ is equal to $\embed{T}$ and we have finished the proof.
\end{proof}

\begin{definition}
Let $A$ be a thick graphing of dialect $D^{A}$ a finite subset of $\naturalN$ endowed with the normalised counting measure. We define the graphing: 
\begin{equation*}
\embed{A}=\{(\omega_{e}^{A},\embed{\phi_{e}^{A}}:\embed{S_{e}^{A}}\rightarrow \embed{T_{e}^{A}})\}_{e\in E^{A}}
\end{equation*}
\end{definition}

\begin{definition}
Let $A$ be a thick graphing of dialect $D^{A}$, and $\Omega: \realN\times[0,1]\rightarrow\realN$ an isomorphism of measured spaces. We define the graphing $\oc_{\Omega} A$ by:
\begin{equation*}
\oc_{\Omega} A=\{(\omega_{e}^{A}, \Omega\circ\embed{\phi_{e}^{A}}\circ\Omega^{-1}: \Omega(\embed{S_{e}^{A}})\rightarrow \Omega(\embed{T_{e}^{A}})\}_{e\in E^{A}}
\end{equation*}
\end{definition}

\begin{definition}
A project $\de{a}$ is \emph{balanced} if $\de{a}=(0,A)$ where $A$ is a thick graphing, i.e.\ $I^{A}$ is a one-element set, for instance $I^{A}=\{1\}$, and $\alpha^{A}_{1}=1$.
\end{definition}

\begin{definition}
Let $\de{a}$ be a balanced project. We define $\oc_{\Omega}\de{a}=(0,\oc_{\Omega} A)$. If $\cond{A}$ is a conduct, we define: 
$$\oc_{\Omega}\cond{A}=\{\oc_{\Omega}\de{a}~|~\de{a}=(0,A)\in\cond{A}, \text{ $\de{a}$ balanced}\}^{\pol\pol}$$
\end{definition}

We will now show that it is possible to implement the functorial promotion. In order to do this, we define the bijections $\tau,\theta: \naturalN+\naturalN+\naturalN\rightarrow \naturalN+\naturalN+\naturalN$:
\begin{eqnarray*}
\tau&:& \left\{\begin{array}{rcl}
	(x,0)&\mapsto&(x,1)\\
	(x,1)&\mapsto&(x,0)\\
	(x,2)&\mapsto&(x,2)
	\end{array}\right.\\
\theta&:& \left\{\begin{array}{rcl}
	(x,0)&\mapsto&(2x,0)\\
	(x,1)&\mapsto&(2x+1,1)\\
	(2x,2)&\mapsto&(x,1)\\
	(2x+1,2)&\mapsto&(x,2)
	\end{array}\right.
\end{eqnarray*}
These bijections induce bijections of $\naturalN$ onto $\naturalN$ through $\psi: (x,i)\mapsto 3x+i$. We will abusively denote by $T_{\tau}=T_{\psi\circ\tau\circ\psi^{-1}}$ and $T_{\theta}=T_{\psi\circ\theta\circ\psi^{-1}}$ the induced measure-preserving transformations $[0,1]\rightarrow[0,1]$.

Pick $\de{a}\in\cond{\sharp\phi(A)}$ and $\de{f}\in\sharp(\cond{A\multimap B})$, where $\phi$ is a delocation. By definition, $\de{a}=(0,\Omega(\embed{A}))$ and $\de{f}=(0,\Omega(\embed{F}))$ where $A,F$ are graphings of respective dialects $D^{A},D^{F}$. We define the graphing $T=\{(1,\Omega(\text{Id}\times T_{\tau})),(1,(\Omega(\text{Id}\times T_{\tau}))^{-1})\}$ of carrier $V^{\phi(A)}\cup V^{A}$, and denote by $t,t^{\ast}$ the two edges in $E^{T}$. We fix $(x,y)$ an element of $V^{B}$ and we will try to understand the action of the path $f_{0}ta_{0}t^{\ast}f_{1}\dots ta_{k-1}t^{\ast}f_{k}$.

We fix the partition $\mathcal{P}_{D^{F}+D^{A}}$ of $[0,1]$, and denote by $(i,j)$ the integers such that $y\in \mathcal{P}_{D^{F}+D^{A}}^{i,j}$. By definition of $\embed{F}$, the map $\embed{\phi^{F}_{f_{0}}}$ sends this element to $(x_{1},y_{1})$ which is an element of $\mathcal{P}_{D^{F}+D^{A}}^{i_{1},j_{1}}$ with $j_{1}=j$. Then, the function $\phi_{t}$ sends this element on $(x_{2},y_{2})$, where $x_{2}=x_{1}$ and $y_{2}$ is an element of $\mathcal{P}_{D^{F}+D^{A}}^{j_{1},i_{1}}$. The function $\embed{\phi_{a_{0}}^{A}}$ then produces an element $(x_{3},y_{3})$ with $y_{3}$ in $\mathcal{P}_{D^{F}+D^{A}}^{j_{2},i_{2}}$ and $i_{2}=i_{1}$. The element produced by $\phi_{t^{\ast}}=\phi_{t}^{-1}$ is then $(x_{4},y_{4})$ where $y_{4}$ is an element of $\mathcal{P}_{D^{F}+D^{A}}^{i_{2},j_{2}}$. One can therefore see how the graphing $T$ simulates the dialectal interaction. The following proposition will show how one can use $T$ to implement functorial promotion.

In order to implement functorial promotion, we will make use of the three bijections we just defined. Though it may seem a complicated, the underlying idea is quite simple. We will be working with three disjoint copies of $\naturalN$, let us say $\naturalN_{i}$ ($i=0,1,2$). When applying promotion, we will encode the information contained in the dialect on the first copy $\naturalN_{0}$ (let us stress here that promotion is defined through a non-surjective map, something that will be essential in the following). Suppose now that we have two graphs obtained from two promotions: all the information they contain is located in their first copy $\naturalN_{0}$. To simulate dialectal information, we need to make these two sets disjoint: this is where the second copy $\naturalN_{1}$ will be used. Hence, we apply to one of these promoted graphs the bijection $\tau$ (in practice we will of course use $\tau$ through the induced transformation $T_{\tau}$) which exchanges $\naturalN_{0}$ and $\naturalN_{1}$. The information coming from the dialects of the two graphs are now disjoint. We then compute the execution of the two graphs to obtain a graph whose information coming from the dialect is encoded on the two copies $\naturalN_{0}$ and $\naturalN_{1}$! In order to be able to see this obtained graph as a graph obtained from a promotion, we need now to move this information so that it is encoded on the first copy $\naturalN_{0}$ only. This is where we use the third copy $\naturalN_{2}$: we use the bijection $\theta$ (once again, we use in practice the induced transformation $T_{\theta}$) in order to contract the two copies $\naturalN_{0}$ and $\naturalN_{1}$ on the first copy $\naturalN_{0}$, while deploying the third copy $\naturalN_{2}$ onto the two copies $\naturalN_{1}$ and $\naturalN_{2}$.

\begin{proposition}\label{promfonct}
One can implement functorial promotion: for all delocations $\phi,\psi$ and conducts $\cond{A,B}$ such that $\cond{\phi(A),A,B,\psi(B)}$ have pairwise disjoint carriers, there exists a project $\de{prom}$ in the conduct
$$\cond{\oc\phi(A)}\otimes\oc(\cond{A\multimap B})\multimap\cond{\oc\psi(B)}$$
\end{proposition}

\begin{proof}
Let $\de{f}\in\cond{A\multimap B}$ be a balanced project, $\phi,\psi$ two delocations of $\cond{A}$ and $\cond{B}$ respectively. We define the graphings $T=\{(1,\Omega(\text{Id}\times T_{\tau})),(1,(\Omega(\text{Id}\times T_{\tau}))^{-1})\}$ of carrier $V^{\phi(A)}\cup V^{A}$ and $P=\{(1,\Omega(\text{Id}\times T_{\theta})),(1,(\Omega(\text{Id}\times T_{\theta}))^{-1})\}$ of carrier $V^{B}\cup V^{\psi(B)}$. We define $\de{t}=(0,T)$ and $\de{p}=(0,P)$, and the project:
$$\de{prom}=(0,T\cup P)=\de{t}\otimes\de{p}$$ 
We will now show that $\de{prom}$ is an element in $\cond{(\oc\phi(A)\otimes\oc(A\multimap B))\multimap \oc\psi(B)}$.

We can suppose, up to choosing refinements of $A$ and $F$, that for all $e\in E^{A}\cup E^{F}$, $(S_{e})_{2}$ and $(T_{e})_{2}$ are one-elements sets\footnote{The sets $S_{e}$ and $T_{e}$ being subsets of a product, w write $(S_{e})_{2}$ (resp. $(T_{e})_{2}$) the result of their projection on the second component.}.

Pick $\de{a}\in\cond{\sharp\phi(A)}$ and $\de{f}\in\sharp(\cond{A\multimap B})$. Then, by definition $\de{a}=(0,\Omega(\embed{A}))$ and $\de{f}=(0,\Omega(\embed{F}))$ where $A,F$ are graphings of dialects $D^{A},D^{F}$. We get that $\de{a\otimes f}\plug\de{prom}=((\de{a}\plug\de{t})\plug\de{f})\plug\de{p}$ from the associativity and commutativity of $\plug$ (recall that $\de{a\otimes f}=\de{a\plug f}$).

We show that $\embed{A}\plug \embed{T}$ is the graphings composed of the $\oc^{\tau} \phi_{a}$ for $a\in E^{A}$, where $\oc^{\tau}\phi_{a}$ is defined by:
\begin{equation*}
\oc^{\tau}\phi_{a}:(x,y)\mapsto (x',y'),~~~~~\phi_{a}(x,\mathcal{P}_{\{0\}+D^{A}}^{-1}(y))=(x',z), ~~~ y'=T_{\mathcal{P}_{\{0\}+D^{A}}^{-1}(y)}^{(z,1)}(y)
\end{equation*}
This is almost straightforward. An element in $\embed{A}\plug\embed{T}$ is a path of the form $t at^{\ast}$. It is therefore the function $\phi_{t}\circ \embed{\phi_{a}}\circ\phi_{t}^{-1}$. By definition,
\begin{equation*}
\embed{\phi_{a}}:(x,y)\mapsto(x',y')~~~~n=\mathcal{P}_{A}^{-1}(y),~~ \phi_{a}(x,n)=(x',k),~ y'=T_{n}^{k}(y)
\end{equation*}
But $\phi_{t}:\text{Id}\times T_{\tau}$ and $T_{\tau}$ is a bijection from $\mathcal{P}_{A}(y)$ to $\mathcal{P}_{\{0\}+A}(1,y)$.

We now describe the graphing $G=(\embed{A}\plug\embed{T})\plug\embed{F}$. It is composed of the paths of the shape $\rho=f_{0}(t a_{0}t^{\ast})f_{1}(t a_{1}t^{\ast})f_{2}\dots f_{n-1}(t a_{n-1}t^{\ast})f_{n}$. The associated function is therefore:
\begin{equation*}
\phi_{\rho}=\embed{\phi_{f_{0}}}(\oc^{\tau}\phi_{a_{0}})\embed{\phi_{f_{1}}}\dots\embed{\phi_{f_{n-1}}}(\oc^{\tau}\phi_{a_{n-1}})\embed{\phi_{f_{n}}}
\end{equation*}
Let $\pi=f_{0}a_{0}f_{1}\dots f_{n-1}a_{n-1}f_{n}$ be the corresponding path in $F\plug A$. The function $\phi_{\pi}$ has, by definition, as domain and codomain measurable subsets of $\realN\times D^{F}\times D^{A}$. We define, for such a function, the function $\co\phi_{\pi}$ by:
\begin{eqnarray*}
&\co\phi_{\pi}: (x,y)\mapsto(x',y')\\
&(n,m)=\mathcal{P}_{D^{F}+D^{A}}^{-1}(y),~~\phi_{\pi}(x,n,m)=(x',k,l),~~y'=T_{(n,m)}^{(k,l)}(y)
\end{eqnarray*}
One can then check that $\co\phi_{\pi}=\phi_{\rho}$.

Finally, $G\plug \embed{P}$ is the graphing composed of paths that have the shape $p\rho p^{\ast}$ where $\rho$ is a path in $G$. But $\phi_{p}=\text{Id}\times T_{\theta}$ applies a bijection, for all couple $(k,l)\in D^{F}\times D^{A}$, from the set $\mathcal{P}_{D^{F}+D^{A}}^{k,l}$ to the set $\mathcal{P}_{\theta(D^{F}+ D^{A})}^{\theta(k,l)}$ where:
\begin{equation*}
\theta(D^{F}+D^{A})=\{\theta(f,a)~|~f\in D^{F},a\in D^{A}\}
\end{equation*}
We deduce that:
\begin{eqnarray*}
&\phi_{p\rho p^{\ast}}: (x,y)\mapsto (x',y')\\
&n=\theta(k,l)=\mathcal{P}_{\theta(D^{F}+D^{A})}^{-1}(y)~~~\phi_{\pi}(x,k,l)=(x',k',l')~~~y'=T_{n}^{\theta(k',l')}(y)
\end{eqnarray*}
Modulo the bijection $\mu: D^{F}\times D^{A}\rightarrow\theta(D^{F}+D^{A})\subset\naturalN$, we get that $G\plug\embed{P}$ is the delocation (along $\psi$) of the graphing $\oc (F\plug A)$.

Therefore, for all $\de{a},\de{f}$ in $\cond{\sharp A},\cond{\sharp (A\multimap B)}$ respectively there exists a project $\de{b}$ in $\cond{\sharp \psi(B)}$ such that $\de{prom}\plug(\de{a}\otimes\de{f})=\de{b}$. We showed that for all $\de{g}\in\cond{\sharp{A}\odot\sharp(A\multimap B)}$, one has $\de{prom}\plug\de{g}\in\cond{\oc B}$, and thus $\de{prom}$ is an element in $(\sharp A\odot\sharp(A\multimap B))^{\pol\pol}\multimap \cond{B}$ by \autoref{propositionethsendcondsend}. But $(\sharp A\odot\sharp(A\multimap B))^{\pol\pol}=\cond{\oc A\otimes\oc(A\multimap B)}$ by \autoref{ethtenscondtens}.
\end{proof}

\begin{figure}
\centering
\subfigure[Global Picture]{
\centering
\begin{tikzpicture}[x=1.1cm,scale=0.9]
	\draw[-] (0,0) -- (2,0) node [midway,below] {$\oc \phi(A)$};
			\draw[-,fill,red,opacity=0.4] (0,4) -- (2,4) node [midway,above] {$\phi(V^{A})\times[0,1]$} -- (2,2) -- (0,2) -- (0,4) {};
	\draw[-] (3,0) -- (7,0) node [midway,below] {$\oc F$};
		\draw[-,blue] (3,0.1) -- (4.9,0.1) node [midway,below] {\small{$\oc V^{A}$}};
		\draw[-,blue] (5.1,0.1) -- (7,0.1) node [midway,below] {\small{$\oc V^{B}$}};
			\draw[-,fill,red,opacity=0.4] (3,4) -- (7,4) node [midway,above] {$(V^{A}\cup V^{B})\times[0,1]$} -- (7,2) -- (3,2) -- (3,4) {};
	\draw[-] (8,0) -- (10,0) node [midway,below] {$\oc \psi(B)$};
			\draw[-,fill,red,opacity=0.4] (8,4) -- (10,4) node [midway,above] {$\psi(V^{B})\times[0,1]$} -- (10,2) -- (8,2) -- (8,4) {};
			
	\draw[<->] (1,0) -- (1,2) node [midway,left] {$\Omega$};
	\draw[<->] (5,0) -- (5,2) node [midway,left] {$\Omega$};
	\draw[<->] (9,0) -- (9,2) node [midway,left] {$\Omega$};
	
	\draw[<->] (2,3) -- (3,3) node [midway,above] {\small{$\text{Id}\times T_{\tau}$}};
	\draw[<->] (7,3) -- (8,3) node [midway,above] {\small{$\text{Id}\times T_{\theta}$}};
\end{tikzpicture}
}
\subfigure[Action of $T_{\tau}$]{
\begin{tikzpicture}[x=0.9cm,y=0.9cm]
	\draw[-,fill,blue,opacity=0.4] (-3,5) -- (-1,5) node [midway,above] {$\phi(V^{A})\times D^{A}$} -- (-1,3) -- (-3,3) -- (-3,5) {};
	\draw[-,fill,blue,opacity=0.4] (1,5) -- (3,5) node [midway,above] {$\phi(V^{A})\times D^{A}$} -- (3,3) -- (1,3) -- (1,5) {};

	\draw[-,fill,red,opacity=0.4] (-3,2) -- (-1,2) -- (-1,0) -- (-3,0)  node [midway,below] {$\phi(V^{A})\times[0,1]$} -- (-3,2) {};
	\draw[-,fill,red,opacity=0.4] (1,2) -- (3,2) -- (3,0) -- (1,0) node [midway,below] {$\phi(V^{A})\times[0,1]$} -- (1,2) {};	
	
	\draw[->] (-2,3) -- (-2,2) node [midway,left] {$\mathcal{P}_{D^{A}}$};
	\draw[->] (2,3) -- (2,2) node [midway,left] {$\mathcal{P}_{\{0\}+D^{A}}$};
	\draw[->] (-1,1) -- (1,1) node [midway,below] {$\text{Id}\times T_{\tau}$};
	\draw[->] (-1,4) -- (1,4) node [midway,below] {$\text{Id}\times \tau$};
	
\end{tikzpicture}
}
\subfigure[Action of $T_{\theta}$]{
\begin{tikzpicture}[x=0.9cm,y=0.9cm]
	\draw[-,fill,blue,opacity=0.4] (-3,5) -- (-1,5) node [midway,above] {$\psi(V^{B})\times D^{F}\times D^{A}$} -- (-1,3) -- (-3,3) -- (-3,5) {};
	\draw[-,fill,blue,opacity=0.4] (1,5) -- (3,5) node [midway,above] {$\psi(V^{B})\times \theta(D^{F}+D^{A})$} -- (3,3) -- (1,3) -- (1,5) {};

	\draw[-,fill,red,opacity=0.4] (-3,2) -- (-1,2) -- (-1,0) -- (-3,0)  node [midway,below] {$\phi(V^{A})\times[0,1]$} -- (-3,2) {};
	\draw[-,fill,red,opacity=0.4] (1,2) -- (3,2) -- (3,0) -- (1,0) node [midway,below] {$\phi(V^{A})\times[0,1]$} -- (1,2) {};	
	
	\draw[->] (-2,3) -- (-2,2) node [midway,left] {$\mathcal{P}_{D^{F}+D^{A}}$};
	\draw[->] (2,3) -- (2,2) node [midway,left] {$\mathcal{P}_{\theta(D^{F}+D^{A})}$};
	\draw[->] (-1,1) -- (1,1) node [midway,below] {$\text{Id}\times T_{\theta}$};
	\draw[->] (-1,4) -- (1,4) node [midway,below] {$\text{Id}\times \theta$};
	
\end{tikzpicture}
}
\caption{Functorial Promotion}
\end{figure}
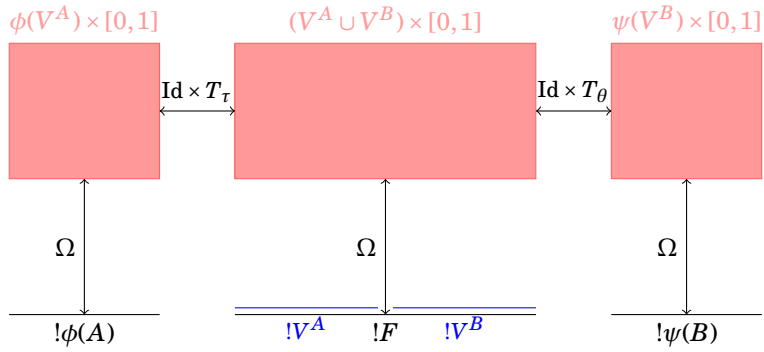
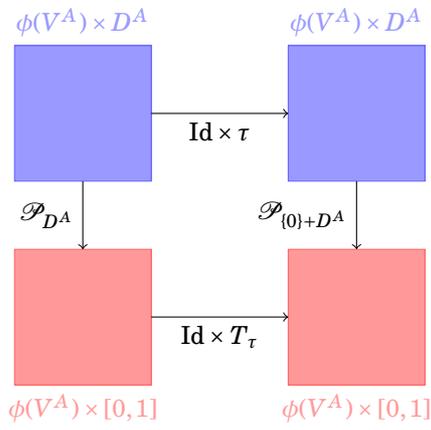
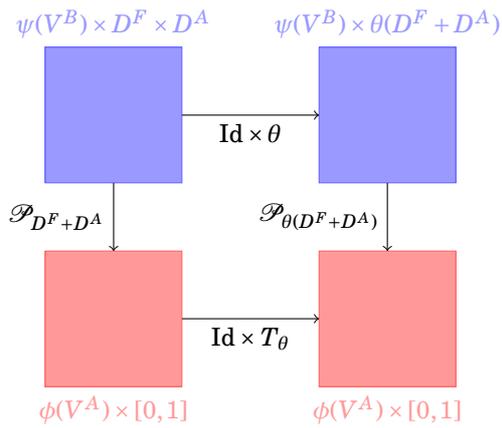

In the setting of its hyperfinite geometry of interaction \cite{goi5}, Girard shows how one can obtain the exponentials isomorphism as an equality between the conducts $\oc(\cond{A\with B})$ and $\cond{\oc A\otimes\oc B}$. Things are however quite different here. Indeed, if the introduction of behaviours in place of Girard's negative/positive conducts is very interesting when one is interested in the additive connectives, this leads to a (small) complication when dealing with exponentials. The first thing to notice is that the proof of the implication $\oc A\otimes\oc B\multimap \oc (A\with B)$ in a sequent calculus with functorial promotion and without dereliction and digging rules cannot be written if the weakening rule is restrained to the formulas of the form $\wn A$:

\begin{prooftree}
\AxiomC{}
\RightLabel{\scriptsize{ax}}
\UnaryInfC{$\vdash A,A^{\pol}$}
\RightLabel{\scriptsize{weak}}
\UnaryInfC{$\vdash A,B^{\pol},A^{\pol}$}
\AxiomC{}
\RightLabel{\scriptsize{ax}}
\UnaryInfC{$\vdash B,B^{\pol}$}
\RightLabel{\scriptsize{weak}}
\UnaryInfC{$\vdash B,B^{\pol},A^{\pol}$}
\RightLabel{\scriptsize{$\with$}}
\BinaryInfC{$\vdash B^{\pol},A^{\pol},A\with B$}
\RightLabel{\scriptsize{$\oc$}}
\UnaryInfC{$\vdash \wn B^{\pol},\wn A^{\pol},\oc(A\with B)$}
\doubleLine
\UnaryInfC{$\vdash \oc A\otimes\oc B\multimap \oc (A\with B)$}
\end{prooftree}

In Girard's setting, weakening is available for all positive conducts (the conducts on which one can apply the $\wn$ modality), something which is coherent with the fact that the inclusion $\cond{\oc A\otimes\oc B\subset \oc (A\with B)}$ is satisfied. In our setting, however, weakening is never available for behaviours and we think the latter inclusion is therefore not satisfied. This question stays however open.

Concerning the converse inclusion, it does not seem clear at first that it is satisfied in our setting either. This issue comes from the contraction rule. Indeed, since the latter does not seem to be satisfied in full generality (see \autoref{remcontraction}),  one could think the inclusion $\cond{\oc (A\with B)\subset \oc A\otimes\oc B}$ is not satisfied either. We will show however in \autoref{polarise}, through the introduction of alternative ``additive connectives'', that it does hold (a result that will not be used until the last section).

\begin{proposition}
The conduct $\cond{1}$ is a perennial conduct, equal to $\oc\cond{T}$.
\end{proposition}

\begin{proof}
By definition, $\cond{1}=\{(0,\emptyset)\}^{\pol\pol}$ is a perennial conduct. Moreover, the balanced projects in $\cond{T}$ are the projects of the shape $\de{t}_{D}=(0,\emptyset)$ with dialects $D\subset\naturalN$. Each of these satisfy $\oc\de{t}_{D}=(0,\emptyset)$. Thus $\sharp\cond{T}=\{(0,\emptyset)\}$ and $\cond{\oc T}=\cond{1}$.
\end{proof}

\begin{corollary}
The conduct $\cond{\bot}$ is a co-perennial conduct, equal to $\wn\cond{0}$.
\end{corollary}

\begin{proof}
This is straightforward: 
\begin{equation*}
\cond{\bot}=\cond{1}^{\pol}=\cond{(\oc T)^{\pol}}=\cond{(\sharp T)^{\pol\pol\pol}}=\cond{(\sharp T)^{\pol}}=\cond{(\sharp 0^{\pol})^{\pol}}=\cond{\wn 0}\tag*{\qedhere}
\end{equation*}
\end{proof}

\subsection{Impossibility results for linear logic}

We now finish this section with a few negative results showing that the obtained model is not a model of linear logic (with \emph{full} exponential connectives) in an obvious manner. These results support the hypothesis that the generalisation to continuous dialects considered in \cite{seiller-goif} is necessary for modelling linear logic exponentials.

\begin{theorem}\label{nodereliction}
Let $\cond{A}$ be a behaviour. No project $\de{der}$ in $\cond{\oc_{\Omega} A\multimap A}$ is such that $\de{der\plug \oc_{\Omega} a}=\de{a}$.
\end{theorem}

\begin{proof}
The proof is quite easy. Suppose there exists such a project $\de{der}=(d,D)$. First, let us notice that $d=0$ since $\cond{\oc_{\Omega} A\multimap A}$ is a behaviour. By definition for any balanced project $\de{a}=(0,A)$ in $\cond{A}$, it maps $(0,\oc_{\Omega} A)$ to $(0,A)$, i.e.\ $(\meas{D,\oc_{\Omega} A},D\plug \oc_{\Omega} A)=(0,A)$. That implies  $\meas{D,\oc_{\Omega} A}=0$, but more importantly, it implies that $D\plug \oc_{\Omega} A$ is equal to $A$. Since $\oc_{\Omega} A $ is single-sliced, the execution $D\plug \oc_{\Omega} A$ has the same dialect as $D$. By taking a project $\de{a}=(0,A)$ with a dialect different from this one, we obtain a contradiction.
\end{proof}

Let us now discuss the possible existence of digging. If we had an adequate interpretation of digging, one should be able to provide a single project $\de{dig}_{V}$ such that for all conduct $\cond{A}$ of carrier $V$, $\de{dig}_{V}\in\cond{\oc_{\Omega}\oc_{\Omega} A\multimap \oc_{\Omega} A}$. Indeed, this condition corresponds to the requirement of \emph{naturality} in categorical models. We now show that this is equivalent to a technical result on Borel automorphisms.

\begin{lemma}
There exists such a project $\de{dig}_{V}$ such that $\de{dig\plug \oc_{\Omega}\oc_{\Omega} a}=\de{\oc_{\Omega} a}$ for all $\de{a}$ in $\cond{\oc_{\Omega}\oc_{\Omega} A\multimap \oc_{\Omega} A}$ if and only if there is a Borel automorphism $B: \realN\times[0,1]\times[0,1]\rightarrow \realN\times[0,1]$ such that for all $\phi:  \realN\times[0,1]\rightarrow \realN\times[0,1]$ measurable, $B\circ\phi\circ B^{-1}=\phi$.
\end{lemma}

\begin{proof}
Since it maps single-sliced projects to single-sliced projects, the project $\de{dig}$ should be single-sliced. Let us write $\de{dig}=(0,D)$; we first show that $D$ can be written as a pair of two edges: $B: V^{\oc_{\Omega}A}\rightarrow V^{\oc_{\Omega}\oc_{\Omega}A}$ and $C: V^{\oc_{\Omega}\oc_{\Omega}A}\rightarrow V^{\oc_{\Omega}A}$. 

The first thing to notice is that $D$ should not have any edge with both the source and target in $V$ (resp. in $V^{!}$) because it maps wager-free projects to wager-free projects: if such an edge existed, then it is not hard to build a project $\de{a}$ such that $\de{\oc_{\Omega} a}$ (resp. $\de{\oc_{\Omega}\oc_{\Omega} a}$) possesses a single edge such that $\sca{\oc_{\Omega}a}{dig}\neq 0$. Then since it \enquote{maps} a single edged graphing (e.g. the identity) to a single edge graphing, $D$ should be \emph{deterministic}, i.e.\ the set $\{x\in V\cup V^{\oc}~|~\exists e\neq e'\in E^{D}, x\in S^{D}_{e}\cap S_{e'}^{D}\}$ is of measure zero. These two facts show that $D$ can be written as a a pair of edges $B: V^{\oc_{\Omega}A}\rightarrow V^{\oc_{\Omega}\oc_{\Omega}A}$ and $C: V^{\oc_{\Omega}\oc_{\Omega}A}\rightarrow V^{\oc_{\Omega}A}$. It should be clear that $B,C$ are necessarily non-singular measurable-preserving maps since $D$ \enquote{maps} characteristic functions of measurable sets to characteristic functions of measurable sets with equal measure.

Let us write $\identity[V]$ the project $(0,\identity[V])$ for all measurable set $V$. Then on one hand $D\plug \identity[V^{\oc_{\Omega}}]$ is equal to $B\circ C$ by definition of the execution, and on the other hand it is equal to $\identity[V^{\oc_{\Omega}\oc_{\Omega}}]$. Thus $B\circ C=\identity[V^{\oc_{\Omega}\oc_{\Omega}}]$; similarly $C\circ B$ is equal to $\identity[V^{\oc_{\Omega}}]$. Therefore $C=B^{-1}$, and by the way we have shown that $B$ is a Borel automorphism. Summing up, we have shown that $D\plug \phi$, where $\phi$ is a single edged graphing, computes the conjugation of $\phi$ w.r.t. a Borel automorphism $B$. 

Now, remark that $\oc_{\Omega}\oc_{\Omega} \phi=\Omega\circ (\oc_{\Omega}\phi \times\identity)\circ \Omega^{-1}=\Omega\circ ((\Omega\circ \embed{\phi}\circ \omega^{-1}) \times\identity)\circ \Omega^{-1}$. Up to conjugation w.r.t. $\Omega$, we can show that $B$ is such that for all measurable map $\phi:\realN\times[0,1]\rightarrow\realN\times[0,1]$, $B\circ\phi\circ B^{-1}=\phi\times\identity$. This shows the first implication.

The converse implication is straightforward: if such a Borel automorphism exists, then one can build a graphing $D$ consisting of a pair of edges realised by $B$ and $B^{-1}$ (conjugated by $\Omega$) such that $D\plug \phi=\Omega\circ (\phi \times\identity)\circ \Omega^{-1}$, i.e.\ $D\plug \oc_{\Omega}\phi=\oc_{\Omega}\oc_{\Omega}\phi$.
\end{proof}

Using the latter lemma, we can now prove that no satisfactory interpretation of digging can be found in the model we consider.

\begin{lemma}
Let $\measure{X}$ be a Borel space. There are no Borel automorphisms $B:  \measure{X}\times[0,1]\times[0,1]\rightarrow \measure{X}\times[0,1]$ such that for all $\phi:  \measure{X}\times[0,1]\rightarrow \measure{X}\times[0,1]$ measurable, $B\circ\phi\circ B^{-1}=\phi$.
\end{lemma}

\begin{proof}
Suppose such a Borel automorphism exists, and recall that $\measurable{\measure{Y}}$ denotes the set of all measurable maps $\measure{Y}\rightarrow\measure{Y}$. Let us notice that a Borel automorphism induces an isomorphism\footnote{Let us notice that this argument can be modified to work with the more classical groups of Borel endomorphisms on the spaces.} between the monoid of measurable maps $\measurable{\measure{X}\times[0,1]\times[0,1]}$ and the monoid of measurable maps $\measurable{\measure{X}\times[0,1]\times[0,1]}$. Indeed, the conjugation by $B$, i.e.\ $\phi\mapsto B\circ\phi\circ B^{-1}$ defines such an isomorphism. In particular, it should be onto, which contradicts the hypothesis that $B\circ\phi\circ B^{-1}=\phi$ since not all measurable function in $\measurable{\measure{X}\times[0,1]\times[0,1]}$ is of the form $\psi\times\identity$ with $\psi\in\measurable{\measure{X}\times[0,1]}$.
\end{proof}

As a corollary of the two previous lemmas, we obtain a negative result concerning the possibility of interpreting digging.

\begin{theorem}\label{nodigging}
There do not exists a project $\de{dig}_{V}$ such that $\de{dig\plug \oc_{\Omega}\oc_{\Omega} a}=\de{\oc_{\Omega} a}$ for all $\de{a}$ in $\cond{\oc_{\Omega}\oc_{\Omega} A\multimap \oc_{\Omega} A}$.
\end{theorem}

\section{Soundness for behaviours}\label{ellbehav}

\subsection{Sequent Calculus}

To deal with the three kinds of conducts we are working with (behaviours, perennial and co-perennial conducts), we introduce three types of formulas.

\begin{definition}
We define three types of formulas, (B)ehaviors,  (N)egative -- perennial, and (P)ositive -- co-perennial, inductively defined by the following grammar:
\begin{eqnarray*}
B &:=& \cond{T}~|~\cond{0}~|~X ~|~ X^{\pol} ~|~ B\otimes B ~|~ B\parr B~|~ B\oplus B~|~ B\with B~|~\forall X~B~|~\exists X~B~|~ N\otimes B~|~ P\parr B\\
N &:=& \cond{1}~|~P^{\pol}~|~\oc B ~|~ N\otimes N ~|~ N\oplus N\\
P &:=& \cond{\bot}~|~N^{\pol}~|~\wn B ~|~P\parr P~|~ P\with P
\end{eqnarray*}
We will denote by $\freevar{\Gamma}$ the set of free variables in $\Gamma$, where $\Gamma$ is a sequence of formulas (of any type).
\end{definition}

\begin{definition}
We define \emph{pre-sequents} $\Delta \pdash \Gamma; \Theta$ where $\Delta,\Theta$ contain negative (perennial) formulas, $\Theta$ containing at most one formula, and $\Gamma$ contains only behaviours.
\end{definition}

\label{remcontraction} \autoref{contraction} supposes that we are working with behaviours, and cannot be used to interpret contraction in full generality. It is however possible to show in a similar way that contraction can be interpreted when the sequent contains at least one behaviour (this is the next proposition). This restriction of the context is necessary: without behaviours in the sequent one cannot interpret the contraction since the inflation property is essential for showing that $(1/2)\phi(\oc\de{a})\otimes\psi(\oc\de{a})+(1/2)\de{0}$ is an element of $\cond{\phi(\oc A)\otimes\psi(\oc A)}$. 

\begin{proposition}\label{contractiongen}
Let $\cond{A}$ be a conduct and $\phi,\psi$ be disjoint delocations of $\oc V^{A}$. Let $\cond{C}$ be a behaviour and $\theta$ a delocation disjoint from $\phi$ and $\psi$. Then the project $\de{Ctr}_{\phi\cup\theta}^{\psi}$ is an element of the behaviour:
$$\cond{(\oc A\otimes C)\multimap (\phi(\oc A)\otimes \psi(\oc A)\otimes\theta(C))}$$
\end{proposition}

\begin{proof}
The proof follows the proof of \autoref{contraction}. We show in a similar manner that the project $\de{Ctr}_{\phi\cup\theta}^{\psi}\plug(\de{a\otimes c})$ is universally equivalent to:
$$\frac{1}{2}\phi(\oc\de{a})\otimes\psi(\oc\de{a})\otimes\theta(\cond{C})+\frac{1}{2}\de{0}$$
Since $\cond{\oc A}$ is a perennial conduct and $\cond{C}$ is a behaviour, $\cond{(\phi(\oc A)\otimes \psi(\oc A)\otimes\theta(C))}$ is a behaviour. Thus $\de{Ctr}_{\phi\cup\theta}^{\psi}\plug(\de{a\otimes c})$ is an element in $\cond{(\phi(\oc A)\otimes \psi(\oc A)\otimes\theta(C))}$. Finally we showed that the project $\de{Ctr}_{\phi\cup\theta}^{\psi}$ is an element of $\cond{(\oc A\otimes C)\multimap (\phi(\oc A)\otimes \psi(\oc A)\otimes\theta(C))}$, and that the latter is a behaviour.
\end{proof}

In a similar way, the proof of distributivity relies on the property that $\cond{A+B}\subset\cond{A\with B}$ which is satisfied for behaviours but not in general. It is therefore necessary to restrict to pre-sequents that contain at least one behaviour in order to interpret the $\with$ rule. Indeed, we can think of a pre-sequent $\Delta\pdash \Gamma;\Theta$ as the conduct\footnote{This will actually be the exact definition of its interpretation.}
$$\left(\bigparr_{N\in \Delta}N^{\pol}\right)\parr\left(\bigparr_{B\in\Gamma} B\right)\parr\left(\bigparr_{N\in\Theta} N\right)$$
Such a conduct is a behaviour when the set $\Gamma$ is non-empty and the set $\Theta$ is empty, but it is neither a perennial conduct nor a co-perennial conduct when $\Gamma=\emptyset$. We will therefore restrict to pre-sequents such that $\Gamma\neq\emptyset$ and $\Theta=\emptyset$.

\begin{definition}[Sequents]
A sequent $\Delta\vdash\Gamma;$ is a pre-sequent $\Delta\pdash\Gamma;\Theta$ such that $\Gamma$ is non-empty and $\Theta$ is empty. 
\end{definition}

\begin{definition}[The Sequent Calculus $\textnormal{ELL}_{\textnormal{comp}}$]
A proof in the sequent calculus $\textnormal{ELL}_{\textnormal{comp}}$ is a derivation tree constructed from the derivation rules shown in \autoref{ellcomp} page \pageref{ellcomp}.
\end{definition}

\begin{figure}
\centering
\subfigure[Identity Group]{
\framebox{
\centering
\begin{tabular}{cc}
\begin{minipage}{4.4cm}
\begin{prooftree}
\AxiomC{}
\RightLabel{\scriptsize{ax}}
\UnaryInfC{$\vdash C^{\pol},C;$}
\end{prooftree}
\end{minipage}
&
\begin{minipage}{5.95cm}
\begin{prooftree}
\AxiomC{$\Delta_{1} \vdash \Gamma_{1},C;$}
\AxiomC{$\Delta_{2} \vdash \Gamma_{2},C^{\pol};$}
\RightLabel{\scriptsize{cut}}
\BinaryInfC{$\Delta_{1},\Delta_{2}\vdash \Gamma_{1},\Gamma_{2};$}
\end{prooftree}
\end{minipage}
\end{tabular}
}
}
\subfigure[Multiplicative Group]{
\framebox{
\begin{tabular}{c}
\begin{tabular}{cc}
\begin{minipage}{5cm}
\begin{prooftree}
\AxiomC{$\Delta_{1}\vdash \Gamma_{1},C_{1};$}
\AxiomC{$\Delta_{2}\vdash \Gamma_{2},C_{2};$}
\RightLabel{\scriptsize{$\otimes$}}
\BinaryInfC{$\Delta_{1},\Delta_{2}\vdash\Gamma_{1},\Gamma_{2},C_{1}\otimes C_{2};$}
\end{prooftree}
\end{minipage}
&
\begin{minipage}{4.95cm}
\begin{prooftree}
\AxiomC{$\Delta \vdash \Gamma,C_{1},C_{2};$}
\RightLabel{\scriptsize{$\parr$}}
\UnaryInfC{$\Delta\vdash \Gamma,C_{1}\parr C_{2};$}
\end{prooftree}
\end{minipage}
\\~\\
\begin{minipage}{4cm}
\begin{prooftree}
\AxiomC{$\Delta, N_{1},N_{2} \vdash \Gamma;$}
\RightLabel{\scriptsize{$\otimes^{\mathrm{pol}}_{g}$}}
\UnaryInfC{$\Delta,N_{1}\otimes N_{2}\vdash \Gamma;$}
\end{prooftree}
\end{minipage}
&
\begin{minipage}{4cm}
\begin{prooftree}
\AxiomC{$\Delta,P^{\pol} \vdash \Gamma,C;$}
\RightLabel{\scriptsize{$\parr^{\mathrm{mix}}$}}
\UnaryInfC{$\Delta \vdash \Gamma,P\parr C ;$}
\end{prooftree}
\end{minipage}
\\~\\
\begin{minipage}{2.83cm}
\begin{prooftree}
\AxiomC{$\Delta\vdash \Gamma,C;$}
\RightLabel{\scriptsize{$\cond{1}_{d}$}}
\UnaryInfC{$\Delta\vdash \Gamma,C\otimes \cond{1};$}
\end{prooftree}
\end{minipage}
&
\begin{minipage}{2.83cm}
\begin{prooftree}
\AxiomC{$\Delta \vdash \Gamma;$}
\RightLabel{\scriptsize{$\cond{1}_{g}$}}
\UnaryInfC{$\Delta, \cond{1} \vdash \Gamma;$}
\end{prooftree}
\end{minipage}
\end{tabular}
\end{tabular}
}
}
\subfigure[Additive Group]{
\framebox{
\begin{tabular}{cc}
\begin{minipage}{5cm}
\begin{prooftree}
\AxiomC{$\Delta\vdash\Gamma,C_{i};$}
\RightLabel{\scriptsize{$\oplus_{i}$}}
\UnaryInfC{$\Delta\vdash \Gamma,C_{1}\oplus C_{2};$}
\end{prooftree}
\end{minipage}
&
\begin{minipage}{5.4cm}
\begin{prooftree}
\AxiomC{$\Delta\vdash \Gamma, C_{1};$}
\AxiomC{$\Delta\vdash \Gamma, C_{2};$}
\RightLabel{\scriptsize{$\with$}}
\BinaryInfC{$\Delta\vdash \Gamma, C_{1}\with C_{2};$}
\end{prooftree}
\end{minipage}
\\~\\
\begin{minipage}{4cm}
\begin{prooftree}
\AxiomC{}
\RightLabel{\scriptsize{$\top$}}
\UnaryInfC{$\Delta\vdash \Gamma, \top;$}
\end{prooftree}
\end{minipage}
&
\begin{minipage}{4cm}
\centering
No rules for $0$.
\end{minipage}
\end{tabular}
}
}
\subfigure[Exponential Group]{
\framebox{
\begin{tabular}{c}
\begin{minipage}{5.4cm}
\begin{prooftree}
\AxiomC{$\Delta_{1}\vdash \Gamma_{1}, C_{1};$}
\AxiomC{$\Delta_{2} \vdash \Gamma_{2},C_{2};$}
\RightLabel{\scriptsize{$\oc$}}
\BinaryInfC{$\Delta_{1},\oc \Delta_{2}, \oc\Gamma_{2}^{\pol} \vdash \Gamma_{1},C_{1}\otimes \oc C_{2} ;$}
\end{prooftree}
\end{minipage}
\\
\begin{tabular}{cc}
\begin{minipage}{5cm}
\begin{prooftree}
\AxiomC{$\Delta,\oc A,\oc A\vdash \Gamma;$}
\RightLabel{\scriptsize{ctr ($\Gamma\neq\emptyset$)}}
\UnaryInfC{$\Delta,\oc A\vdash \Gamma;$}
\end{prooftree}
\end{minipage}
&
\begin{minipage}{5cm}
\begin{prooftree}
\AxiomC{$\Delta\vdash\Gamma;$}
\RightLabel{\scriptsize{weak}}
\UnaryInfC{$\Delta,N\vdash \Gamma;$}
\end{prooftree}
\end{minipage}
\end{tabular}
\end{tabular}
}
}
\subfigure[Quantifier Group]{
\framebox{
\begin{tabular}{c}
\begin{tabular}{cc}
\begin{minipage}{5cm}
\begin{prooftree}
\AxiomC{$\vdash \Gamma,C;$}
\AxiomC{$X\not\in \freevar{\Gamma}$}
\RightLabel{\scriptsize{$\forall$}}
\BinaryInfC{$\vdash\Gamma,\forall X~ C;$}
\end{prooftree}
\end{minipage}
&
\begin{minipage}{4.95cm}
\begin{prooftree}
\AxiomC{$\vdash \Gamma,C[A/X];$}
\RightLabel{\scriptsize{$\exists$}}
\UnaryInfC{$\vdash \Gamma,\exists X~C;$}
\end{prooftree}
\end{minipage}
\end{tabular}
\end{tabular}
}
}
\caption{Rules for the sequent calculus $\textnormal{ELL}_{\textnormal{comp}}$}\label{ellcomp}
\end{figure}

\subsection{Truth}

The notion of success is the natural generalisation of the corresponding notion on graphs \cite{seiller-goim,seiller-goiadd}. The graphing of a successful project will therefore be a disjoint union of ``transpositions''. Such a graphing can be represented as a graph with a set of vertices that could be infinite, but since we are working with equivalence classes of graphings one can always find a simpler representation: a graphing with exactly two edges.

\begin{definition}
A project $\de{a}=(a,A)$ is \emph{successful} when it is balanced, $a=0$ and $A$ is a disjoint union of transpositions:
\begin{itemize}[noitemsep,nolistsep]
\item for all $e\in E^{A}$, $\omega^{A}_{e}=1$;
\item for all $e\in E^{A}$, $\exists e^{\ast}\in E^{A}$ such that $\phi^{A}_{e^{\ast}}=(\phi_{e}^{A})^{-1}$ -- in particular $S_{e}^{A}=T_{e^{\ast}}^{A}$ and $T_{e}^{A}=S_{e^{\ast}}^{A}$;
\item for all $e,f\in E^{A}$ with $f\not\in\{e,e^{\ast}\}$, $S^{A}_{e}\cap S^{A}_{f}$ and $T^{A}_{e}\cap T^{A}_{f}$ are of null measure;
\end{itemize}
A conduct $\cond{A}$ is \emph{true} when it contains a successful project.
\end{definition}

%

The following results were shown in our previous paper \cite{seiller-goig}. They ensure that the given definition of truth is coherent.

\begin{proposition}[Consistency]
The conducts $\cond{A}$ and $\cond{A}^{\pol}$ cannot be simultaneously true.
\end{proposition}

\begin{proof}
We suppose that $\de{a}=(0,A)$ and $\de{b}=(0,B)$ are successful project in the conducts $\cond{A}$ and $\cond{A}^{\pol}$ respectively. Then:
\begin{equation*}
\sca{a}{b}=\meas{A,B}
\end{equation*}
If there exists a cycle whose support is of strictly positive measure between $A$ and $B$, then $\meas{A,B}=\infty$. Otherwise, $\meas{A,B}=0$. In both cases we obtained a contradiction since $\de{a}$ and $\de{b}$ cannot be orthogonal.
\end{proof}

\begin{proposition}[Compositionnality]
If $\cond{A}$ and $\cond{A \multimap B}$ are true, then $\cond{B}$ is true.
\end{proposition}

\begin{proof}
Let $\de{a}\in\cond{A}$ and $\de{f}\in\cond{A\multimap B}$ be successful projects. Then:
\begin{itemize}[noitemsep,nolistsep]
\item If $\sca{a}{f}=\infty$, the conduct $\cond{B}$ is equal to $\cond{T}_{V^{B}}$, which is a true conduct since it contains $(0,\emptyset)$;
\item Otherwise $\sca{a}{f}=0$ (this is shown in the same manner as in the preceding proof) and it is sufficient to show that $F\plug A$ is a disjoint union of transpositions. But this is straightforward: to each path there corresponds an opposite path and the weights of the paths are all equal to $1$, the conditions on the source and target sets $S_{\pi}$ and $T_{\pi}$ are then easily checked.
\end{itemize}
Finally, if $\cond{A}$ and $\cond{A\multimap B}$ are true, then $\cond{B}$ is true.
\end{proof}

\subsection{Interpretation of proofs}

To prove soundness, we will follow the proof technique used in our previous papers \cite{seiller-goim,seiller-goiadd,seiller-goig}. We will first define a localised sequent calculus and show a result of full soundness for it. The soundness result for the non-localised calculus is then obtained by noticing that one can always \emph{localise} a derivation. We will consider here that the variables are defined with the carrier equal to an interval in $\realN$ of the form $[i,i+1[$.

\begin{definition}
We fix a set $\mathcal{V}=\{X_{i}(j)\}_{i,j\in\naturalN\times\integerN}$ of \emph{localised variables}. For $i\in\naturalN$, the set $X_{i}=\{X_{i}(j)\}_{j\in\integerN}$ will be called the \emph{variable name $X_{i}$}, and an element of $X_{i}$ will be called a \emph{variable of name $X_{i}$}.
\end{definition}
For $i,j\in\naturalN\times\integerN$ we define the \emph{location} $\sharp X_{i}(j)$ of the variable $X_{i}(j)$ as the set $$\{x\in\realN~|~ 2^{i}(2j+1)\leqslant x< 2^{i}(2j+1)+1\}$$
Let us notice that this particular encoding could be replaced with any other $\naturalN\times\integerN$-indexed partition of (a subset of) the real line as measurable subsets of equal measure\footnote{It would also work with the more general case where the partition is such that for all $i$, $j$ and $j'$, the subsets corresponding to $(i,j)$ and to $(i,j')$ have the same measure.}. The reader can convince himself that a suitable definition of encoding of variables can be defined in order to insure that all the results of this section hold in a more general setting.

\begin{definition}[Formulas of $\textnormal{locELL}_{\textnormal{comp}}$]
We inductively define the formulas of \emph{localised polarised elementary linear logic} $\textnormal{locELL}_{\textnormal{comp}}$ as well as their \emph{locations} as follows:
\begin{itemize}[noitemsep,nolistsep]
\item \textbf{Behaviours}:
\begin{itemize}[noitemsep,nolistsep]
\item A variable $X_{i}(j)$ of name $X_{i}$ is a behaviour whose location is defined as $\sharp X_{i}(j)$;
\item If $X_{i}(j)$ is a variable of name $X_{i}$, then $(X_{i}(j))^{\pol}$ is a behaviour whose location is $\sharp X_{i}(j)$.
\item The constants $\cond{T}_{\sharp \Gamma}$ are behaviours whose location is defined as $\sharp\Gamma$;
\item The constants $\cond{0}_{\sharp\Gamma}$ are behaviours whose location is defined as $\sharp\Gamma$.
\item If $A,B$ are behaviours with respective locations $X,Y$ such that $X\cap Y=\emptyset$, then $A\otimes B$ (resp. $A\parr B$, resp. $A\with B$, resp. $A\oplus B$) is a behaviour whose location is $X\cup Y$;
\item If $X_{i}$ is a variable name, and $A(X_{i})$ is a behaviour of location $\sharp A$, then $\forall X_{i}~A(X_{i})$ and $\exists X_{i}~A(X_{i})$ are behaviours of location $\sharp A$.
\item If $A$ is a perennial conduct with location $X$ and $B$ is a behaviour whose location is $Y$ such that $X\cap Y=\emptyset$, then $A\otimes B$ is a behaviour with location $X\cup Y$;
\item If $A$ is a co-perennial conduct whose location is $X$ and $B$ is a behaviour with location $Y$ such that $X\cap Y=\emptyset$, then $A\parr B$ is a behaviour and its location is $X\cup Y$;
\end{itemize}
\item \textbf{Perennial conducts}:
\begin{itemize}[noitemsep,nolistsep]
\item The constant $\cond{1}$ is a perennial conduct and its location is $\emptyset$;
\item If $A$ is a behaviour or a perennial conduct and its location is $X$, then $\oc A$ is a perennial conduct and its location is $\Omega(X\times[0,1])$;
\item If $A,B$ are perennial conducts with respective locations $X,Y$ such that $X\cap Y=\emptyset$, then $A\otimes B$ (resp. $A\oplus B$) is a perennial conduct whose location is $X\cup Y$;
\end{itemize}
\item \textbf{Co-perennial conducts}:
\begin{itemize}[noitemsep,nolistsep]
\item The constant $\cond{\bot}$ is a co-perennial conduct;
\item If $A$ is a behaviour or a co-perennial conduct and its location is $X$, then $\wn A$ is a co-perennial conduct whose location is $\Omega(X\times[0,1])$;
\item If $A,B$ are co-perennial conducts with respective locations $X,Y$ such that $X\cap Y=\emptyset$, then $A\parr B$ (resp. $A\with B$) is a co-perennial conduct whose location is $X\cup Y$;
\end{itemize}
\end{itemize}
If $A$ is a formula, we will denote by $\sharp A$ the location of $A$. A sequent $\Delta\vdash \Gamma;$ of $\textnormal{locELL}_{\textnormal{comp}}$ must satisfy the following conditions:
\begin{itemize}[noitemsep,nolistsep]
\item the formulas of $\Gamma\cup\Delta$ have pairwise disjoint locations;
\item the formulas of $\Delta$ are all perennial conducts;
\item $\Gamma$ is non-empty and contains only behaviours.
\end{itemize}
\end{definition}

\begin{definition}[Interpretations]
An \emph{interpretation basis} is a function $\Phi$ which associates to each variable name $X_{i}$ a behaviour of carrier $[0,1[$.
\end{definition}

\begin{definition}[Interpretation of $\textnormal{locELL}_{\textnormal{comp}}$ formulas]
Let $\Phi$ be an interpretation basis. We define the interpretation $I_{\Phi}(F)$ along $\Phi$ of a formula $F$ inductively:
\begin{itemize}[noitemsep,nolistsep]
\item If $F=X_{i}(j)$, then $I_{\Phi}(F)$ is the delocation (i.e.\ a behaviour) of $\Phi(X_{i})$ defined by the function $x\mapsto 2^{i}(2j+1)+x$;
\item If $F=(X_{i}(j))^{\pol}$, we define the behaviour $I_{\Phi}(F)=(I_{\Phi}(X_{i}(j)))^{\pol}$;
\item If $F=\cond{T}_{\sharp\Gamma}$ (resp. $F=\cond{0}_{\sharp\Gamma}$), we define $I_{\Phi}(F)$ as the behaviour $\cond{T}_{\sharp\Gamma}$ (resp. $\cond{0}_{\sharp\Gamma}$);
\item If $F=\cond{1}$ (resp. $F=\cond{\bot}$), we define $I_{\Phi}(F)$ as the behaviour $\cond{1}$ (resp. $\cond{\bot}$);
\item If $F=A\otimes B$, we define the conduct $I_{\Phi}(F)=I_{\Phi}(A)\otimes I_{\Phi}(B)$;
\item If $F=A\parr B$, we define the conduct $I_{\Phi}(F)=I_{\Phi}(A)\parr I_{\Phi}(B)$;
\item If $F=A\oplus B$, we define the conduct $I_{\Phi}(F)=I_{\Phi}(A)\oplus I_{\Phi}(B)$;
\item If $F=A\with B$, we define the conduct $I_{\Phi}(F)=I_{\Phi}(A)\with I_{\Phi}(B)$;
\item If $F=\forall X_{i} A(X_{i})$, we define the conduct $I_{\Phi}(F)=\cond{\forall X_{i}} I_{\Phi}(A(X_{i}))$;
\item If $F=\exists X_{i} A(X_{i})$, we define the conduct $I_{\Phi}(F)=\cond{\exists X_{i}} I_{\Phi}(A(X_{i}))$.
\item If $F=\oc A$ (resp. $\wn A$), we define the conduct $I_{\Phi}(F)=\oc I_{\Phi}(A)$ (resp. $\wn I_{\Phi}(A)$).
\end{itemize}
Moreover, a sequent $\Delta\vdash \Gamma;$ will be interpreted as the $\parr$ of formulas in $\Gamma$ and negations of formulas in $\Delta$, which will be written $\bigparr\Delta^{\pol}\parr\bigparr \Gamma$. This formulas can also be written in the equivalent form $\bigotimes\Delta\multimap (\bigparr\Gamma)$.
\end{definition}

\begin{definition}[Interpretation of $\textnormal{locELL}_{\textnormal{comp}}$ proofs]\label{interpretationpreuvesellcomp}
Let $\Phi$ be an interpretation basis. We define the interpretation $I_{\Phi}(\pi)$ -- a project -- of a proof $\pi$ inductively:
\begin{itemize}[noitemsep,nolistsep]
\item if $\pi$ is a single axiom rule introducing the sequent $\vdash (X_{i}(j))^{\pol},X_{i}(j')$, we define $I_{\Phi}(\pi)$ as the project $\de{Fax}$ defined by the translation $x \mapsto 2^{i}(2j'-2j)+x$;
\item if $\pi$ is composed of a single rule $\cond{T}_{\sharp \Gamma}$, we define $I_{\Phi}(\pi)=\de{0}_{\sharp\Gamma}$;
\item if $\pi$ is obtained from $\pi'$ by using a $\parr$ rule, a $\parr^{\mathrm{mix}}$ rule, a $\otimes_{g}^{\mathrm{pol}}$ rule, or a $\cond{1}$ rule, then $I_{\Phi}(\pi)=I_{\Phi}(\pi')$;
\item if $\pi$ is obtained from $\pi_{1}$ and $\pi_{2}$ by performing a $\otimes$ rule, we define $I_{\Phi}(\pi)=I_{\Phi}(\pi_{1})\otimes I_{\Phi}(\pi')$;
\item if $\pi$ is obtained from $\pi'$ using a $\text{weak}$ rule or a $\oplus_{i}$ rule introducing a formula of location $V$, we define $I_{\Phi}(\pi)=I_{\Phi}(\pi')\otimes\de{0}_{V}$;
\item if $\pi$ of conclusion $\vdash \Gamma, A_{0}\with A_{1}$ is obtained from $\pi_{0}$ and $\pi_{1}$ using a $\with$ rule, we define the interpretation of $\pi$ in the same way it was defined in our previous paper \cite{seiller-goiadd};
\item If $\pi$ is obtained from a $\forall$ rule applied to a derivation $\pi'$, we define $I_{\Phi}(\pi)=I_{\Phi}(\pi')$;
\item If $\pi$ is obtained from a $\exists$ rule applied to a derivation $\pi'$ replacing the formula $\cond{A}$ by the variable name $X_{i}$, we define $I_{\Phi}(\pi)=I_{\Phi}(\pi')\plug (\bigotimes [e^{-1}(j)\leftrightarrow X_{i}(j)])$, using the notations of our previous paper \cite{seiller-goig};
\item if $\pi$ is obtained from $\pi_{1}$ and $\pi_{2}$ through the use of a promotion rule $\oc$, we think of this rule as the following ``derivation of pre-sequents'':

\begin{prooftree}
\AxiomC{$\vdots^{\pi_{1}}$}
\noLine
\UnaryInfC{$\Delta_{1}\vdash \Gamma_{1},C_{1};$}
\AxiomC{$\vdots^{\pi_{2}}$}
\noLine
\UnaryInfC{$\Delta_{2} \vdash \Gamma_{2},C_{2};$}
\RightLabel{\scriptsize{$\oc$}}
\UnaryInfC{$\oc\Delta_{2},\oc\Gamma_{2}^{\pol}\pdash ;\oc C_{2}$}
\RightLabel{\scriptsize{$\otimes^{\mathrm{mix}}$}}
\BinaryInfC{$\oc \Delta_{2},\Delta_{1}, \oc\Gamma_{2}^{\pol} \vdash \Gamma_{1},C_{1}\otimes \oc C_{2} ;$}
\end{prooftree}
As a consequence, we first define a delocation of $\oc I_{\Phi}(\pi)$ to which we apply the implementation of the functorial promotion. Indeed, the interpretation of
$$\bigparr\Delta^{\pol}\parr\bigparr\Gamma$$
can be written as a sequence of implications. The exponential of a well-chosen delocation is then represented as:
\begin{equation*}
\cond{\oc (\phi_{1}(A_{1})\multimap (\phi_{2}(A_{2})\multimap \dots (\phi_{n}(A_{n})\multimap \phi_{n+1}(A_{n+1}))\dots))}
\end{equation*}
Applying $n$ instances of the project implementing the functorial promotion to the interpretation of $\pi$, we obtain a project $\de{p}$ in:
\begin{equation*}
\cond{\oc (\phi_{1}(A_{1}))\multimap \oc(\phi_{2}(A_{2}))\multimap \dots \oc(\phi_{n}(A_{n}))\multimap \oc(\phi_{n+1}(A_{n+1}))}
\end{equation*}
which is the same conduct as the one interpreting the conclusion of the promotion ``rule'' in the ``derivation of pre-sequents'' we have shown. Now we are left with taking the tensor product of the interpretation of $\pi_{2}$ with the project $\de{p}$ to obtain the interpretation of the $\oc$ rule;
\item if $\pi$ is obtained from $\pi$ using a contraction rule $ctr$, we write the conduct interpreting the premise of the rule as $\cond{(\oc A\otimes \oc A)\multimap D}$. We then define a delocation of the latter in order to obtain $\cond{(\phi(\oc A)\otimes \psi(\oc A))\multimap D}$, and take its execution with $\de{ctr}$ in $\cond{\oc A\multimap (\oc A\otimes \oc A)}$;
\item if $\pi$ is obtained from $\pi_{1}$ and $\pi_{2}$ by applying a $\text{cut}$ rule or a $\text{cut}^{\mathrm{pol}}$ rule, we define $I_{\Phi}(\pi)=I_{\Phi}(\pi_{1})\exec I_{\Phi}(\pi_{2})$.
\end{itemize}
\end{definition}

\begin{theorem}[$\textnormal{locELL}_{\textnormal{comp}}$ soundness]
Let $\Phi$ be an interpretation basis. Let $\pi$ be a derivation in $\textnormal{locELL}_{\textnormal{comp}}$ of conclusion $\Delta\vdash\Gamma;$. Then $I_{\Phi}(\pi)$ is a successful project in $I_{\Phi}(\Delta\vdash\Gamma;)$.
\end{theorem}

\begin{proof}
The proof is a simple consequence of of the proposition and theorems proved before hand. Indeed, the case of the rules of multiplicative additive linear logic was already treated in our previous papers \cite{seiller-goim, seiller-goiadd}. The only rules we are left with are the rules dealing with exponential connectives and the rules about the multiplicative units. But the implementation of the functorial promotion (\autoref{promfonct}) uses a successful project do not put any restriction on the type of conducts we are working with, and the contraction project (\autoref{contraction} and \autoref{contractiongen}) is successful. Concerning the multiplicative units, the rules that introduce them do not change the interpretations.
\end{proof}

As it was remarked in our previous papers, one can choose an enumeration of the occurrences of variables in order to ``localise'' any formula $A$ and any proof $\pi$ of $\textnormal{ELL}_{\textnormal{comp}}$: we then obtain formulas $A^{e}$ and proofs $\pi^{e}$ of $\textnormal{locELL}_{\textnormal{comp}}$. The following theorem is therefore a direct consequence of the preceding one.

\begin{theorem}[Full $\textnormal{ELL}_{\textnormal{comp}}$ Soundness]
Let $\Phi$ be an interpretation basis, $\pi$ an $\textnormal{ELL}_{\textnormal{comp}}$ proof of conclusion $\Delta\vdash \Gamma;$ and $e$ an enumeration of the occurrences of variables in the axioms in $\pi$. Then $I_{\Phi}(\pi^{e})$ is a successful project in $I_{\Phi}(\Delta^{e}\vdash \Gamma^{e};)$.
\end{theorem}

As explained in the introduction, this sequent calculus is not ideal for an \ELL-like system. Indeed, while this system is subject to a restriction similar in spirit to that of Elementary Linear Logic (i.e.\ no digging and dereliction rules), it has a intuitionnistic flavour that is somehow orthogonal to what makes \ELL an interesting system. Indeed, the modified promotion rule is in fact the rule for intuitionnistic implication, disguised under linear logic apparatus. The system $\textnormal{ELL}_{\textnormal{comp}}$ is therefore a restriction of the usual \ELL sequent calculus to a set of \enquote{intuitionnistic} proofs. However the main property of \ELL, namely that it provides a characterisation of elementary time function \cite{danosjoinet}, uses non-intuitionnistic (in this sense) proofs in an essential way, as pointed out to the author by Damiano Mazza. It is therefore important to obtain a less restricted sequent calculus. This is done in the next section, through the introduction of a notion of polarised conducts that generalise the notions of perennial/coperennial conducts.

\section{Contraction and Soundness for Polarised Conducts}\label{polarise}

\subsection{Definitions and Properties}

In this section, we consider a variation on the definition of additive connectives, which is constructed from the definition of the formal sum $\de{a+b}$ of projects. Let us first try to explain the difference between the usual additives $\with$ and $\oplus$ considered until now and the new additives $\tilde{\with}$ and $\tilde{\oplus}$ defined in this section. The conduct $\cond{A\with B}$ contains all the tests that are necessary for the set $\{\de{a'}\otimes\de{0}~|~\de{a'}\in\cond{A}^{\pol}\}\cup\{\de{b'}\otimes\de{0}~|~\de{b'}\in\cond{B}^{\pol}\}$ to generate the conduct $\cond{A\oplus B}$, something for which the set $\de{a+b}$ is not sufficient. For the variant of additives considered in this section, it is the contrary that happens: the conduct $\cond{A\tilde{\with} B}$ is generated by the projects of the form $\de{a+b}$, but it is therefore necessary to add to the conduct $\cond{A\tilde{\oplus} B}$ all the needed tests.

\begin{definition}
Let $\cond{A,B}$ be conducts of disjoint carriers. We define $\cond{A\tilde{\with} B}=\cond{(A+B)^{\pol\pol}}$. Dually, we define $\cond{A\tilde{\oplus} B}=\cond{(A^{\pol}\tilde{\with} B^{\pol})^{\pol}}$.
\end{definition}

These connectives will be useful for showing that the inclusion $\oc\cond{(A\tilde{\with}B)}\subset \cond{\oc A\otimes\oc B}$ holds when $\cond{A,B}$ are behaviours. We will first dwell on some properties of these connectives before showing this inclusion. Notice that if one of the two conducts $\cond{A,B}$ is empty, then $\cond{A\tilde{\with} B}$ is empty. Therefore, the behaviour $\cond{0}_{\emptyset}$ is a kind of absorbing element for $\tilde{\with}$. But the latter connective also has a neutral element, namely the neutral element $\cond{1}$ of the tensor product! Notice that the fact that $\tilde{\with}$ and $\otimes$ share the same unit appeared in Girard's construction\footnote{Our construction \cite{seiller-goiadd} differs slightly from Girard's, which explains why our additives don't share the same unit as the multiplicatives.} of geometry of interaction in the hyperfinite factor \cite{goi5}.

Notice that at the level of denotational semantics, this connective is almost the same as the usual $\with$ (apart from units). The differences between them are erased in the quotient operation.

\begin{proposition}
Distributivity for $\tilde{\with}$ and $\tilde{\oplus}$ is satisfied for behaviours.
\end{proposition}

\begin{proof}
Using the same project than in the proof of \autoref{distributivity}, the proof consists in a simple computation.
\end{proof}

\begin{proposition}
Let $\cond{A,B}$ be behaviours. Then $$\{\de{a}\otimes\de{0}_{V^{B}}~|~\de{a}\in\cond{A}\}\cup\{\de{b}\otimes\de{0}_{V^{A}}~|~\de{b}\in\cond{B}\}\subset\cond{A\tilde{\oplus} B}$$
\end{proposition}

\begin{proof}
We will show only one of the inclusions, the other one can be obtained by symmetry. Choose $\de{f+g}\in\cond{A^{\pol}+B^{\pol}}$ and $\de{a}\in\cond{A}$. Then:
\begin{eqnarray*}
\sca{f+g}{a\otimes 0}&=&\sca{f}{a\otimes 0}+\sca{g}{a\otimes 0}\\
&=&\sca{f}{a}
\end{eqnarray*}
Using the fact that $\de{g}$ and $\de{a}$ have null wagers.
\end{proof}

Recall (this notion is defined and studied in our second paper \cite{seiller-goiadd}) that a behaviour $\cond{A}$ is \emph{proper} if both $\cond{A}$ and its orthogonal $\cond{A^{\pol}}$ are non-empty. Proper behaviour can be characterised as those conducts $\cond{A}$ such that:
\begin{itemize}
\item $(a,A)\in\cond{A}$ implies that $a=0$;
\item for all $\de{a}\in\cond{A}$ and $\lambda\in\realN$, the project $\de{a}+\lambda\de{0}\in\cond{A}$;
\item $\cond{A}$ is non-empty.
\end{itemize}

\begin{proposition}
Let $\cond{A,B}$ be proper behaviours. Then every element in $\cond{A\tilde{\oplus}B}$ is observationally equivalent to an element in $\{\de{a}\otimes\de{0}_{V^{B}}~|~\de{a}\in\cond{A}\}\cup\{\de{b}\otimes\de{0}_{V^{A}}~|~\de{b}\in\cond{B}\}\subset\cond{A\tilde{\oplus} B}$.
\end{proposition}

\begin{proof}
Let $\de{c}\in\cond{A\tilde{\oplus} B}$. Since $\cond{(A^{\pol}+B^{\pol})^{\pol}}=\cond{A\tilde{\oplus} B}$, we know that $\de{c}\poll\de{a+b}$ for all $\de{a+b}\in\cond{A^{\pol}+B^{\pol}}$. By the homothety lemma (\autoref{homothetie}), we obtain, for all $\lambda,\mu$ non-zero real numbers $0$:
\begin{eqnarray*}
\sca{c}{\lambda a+\mu b}=\lambda\sca{c}{a}+\mu\sca{c}{b}\neq0,\infty
\end{eqnarray*}
We deduce that one expression among $\sca{c}{a}$ and $\sca{c}{b}$ is equal to $0$. Suppose, without loss of generality, that it is $\sca{c}{a}$. Then $\sca{c}{a'}=0$ for all $\de{a'}\in\cond{A^{\pol}}$. Thus $\sca{b}{c}\neq0,\infty$ for all $\de{b}\in\cond{B^{\pol}}$. But $\sca{b\otimes 0}{c}=\sca{b}{c\plug 0}$. We finally have that $\de{c\plug 0}\in\cond{B^{\pol}}$ and $\de{c\plug 0}\cong_{\cond{A\tilde{\oplus} B}}\de{c}$.
\end{proof}

\begin{proposition}
Let $\cond{A,B}$ be proper behaviours. Then $\cond{A\tilde{\with} B}$ is a proper behaviour.
\end{proposition}

\begin{proof}
By definition, $\cond{A\tilde{\with} B}=\cond{(A+B)}^{\pol\pol}$, where $\cond{A,B}$ are non empty and contain only single-sliced wager-free projects. Thus $\cond{A+B}$ is non empty and contains only single-sliced wager-free projects. As a consequence, $\cond{(A+B)^{\pol}}$ satisfies the inflation property. Moreover, since $\cond{A}$ has the inflation property, $\cond{A+B}$ has the inflation property: if $\de{a}+\de{b}\in\cond{A+B}$, we have that\ $\de{a+b+\lambda 0}=\de{(a+\lambda 0)+b}$. Thus $\cond{(A+B)^{\pol}}$ contains only wager-free projects. Moreover, $\cond{(A+B)^{\pol}}=\cond{A^{\pol}\tilde{\oplus}B^{\pol}}$ and it is therefore non-empty by the preceding proposition (because $\cond{A}^{\pol},\cond{B^{\pol}}$ are non empty). Then $\cond{(A+B)^{\pol}}$ is a proper behaviour, which allows us to conclude.
\end{proof}

\begin{proposition}\label{expiso}
Let $\cond{A,B}$ be behaviours. Then $\oc\cond{(A\tilde{\with}B)}\subset \cond{\oc A\otimes\oc B}$.
\end{proposition}

\begin{proof}
If one of the behaviours among $\cond{A,B}$ is empty, $\cond{\oc(A\tilde{\with} B)}=\cond{0}=\cond{\oc A\otimes\oc B}$. We will now suppose that $\cond{A,B}$ are both non empty.

Choose $\de{f}=(0,F)$ a single-sliced wager-free project. We have that $\de{f}'=n_{F}/(n_{F}+n_{G})\de{f}\in\cond{A}$ if and only if $\de{f}\in\cond{A}$ from the homothety lemma (\autoref{homothetie}). Moreover, since $\cond{A}$ is a behaviour, $\de{f'}\in\cond{A}$ is equivalent\footnote{The implication $\de{a}\in\cond{A}\Rightarrow \de{a+\lambda 0}\in\cond{A}$ comes from the definition of behaviours, its reciprocal is shown by noticing that $\de{a+\lambda 0-\lambda 0}$ is equivalent to $\de{a}$.} to $\de{f''}=\de{f'}+\sum_{i\leqslant n_{G}}(1/(n_{F}+n_{G}))\de{0}\in\cond{A}$. Since the weighted thick and sliced graphing $\frac{n_{F}}{n_{F}+n_{G}}F+\sum_{i=1}^{n_{G}}\frac{1}{n_{F}+n_{G}}\emptyset$ is universally equivalent to (\autoref{equivalenceuniv}) a single-sliced weighted thick and sliced graphing $F'$, we obtain finally that the project $(0,F')$ is an element of $\cond{A}$ if and only if $\de{f}\in\cond{A}$. We define in a similar way, being given a project $\de{g}$, a weighted graphing with a single slice $G'$ such that $(0,G')\in\cond{B}$ if and only if $\de{g}\in\cond{B}$.

We are now left to show that $\oc (0,F')\otimes\oc (0,G')=\oc(\de{f+g})$. By definition, the graphing of $\oc (0,F')\otimes \oc(0,G')$ is equal to $\oc_{\Omega} F' \disjun \oc_{\Omega} G'$. By definition again, the graphing of $\oc(\de{f+g})$ is equal to $\oc_{\Omega}(F\disjun G)=\oc_{\Omega} F^{\iota_{1}}\disjun \oc_{\Omega} G^{\iota_{2}}$, where $\iota_{1}$ (resp. $\iota_{2}$) denotes the injection of $D^{F}$ (resp. $D^{G}$) into $D^{F}\disjun D^{G}$. We now are left to notice that $\oc_{\Omega} F^{\iota_{1}}=\oc_{\Omega} F'$ since $F^{\iota_{1}}$ and $F'$ are variants one of the other. Similarly, $\oc_{\Omega} G^{\iota_{2}}=\oc_{\Omega} G'$. Finally, we have that $\sharp\cond{(A+B)}\subset\cond{\sharp A\odot\sharp B}$ which is enough to conclude.
\end{proof}

\begin{lemma}
Let $\cond{A}$ be a conduct, and $\phi,\psi$ disjoint delocations. There exists a successful project in the conduct $$\cond{A\multimap \phi(A)\tilde{\with}\psi(A)}$$
\end{lemma}

\begin{proof}
We define $\de{c}=\de{Fax}_{\phi}\otimes\de{0}_{\psi(V^{A})}+\de{Fax}_{\psi}\otimes\de{0}_{\phi(V^{A})}$. Then for all $\de{a}\in\cond{A}$:
\begin{equation*}
\de{c}\plug\de{a}=\phi(\de{a})\otimes\de{0}_{\psi(V^{A})}+\psi(\de{a})\otimes\de{0}_{\phi(V^{A})}
\end{equation*}
Thus $\de{c}\in\cond{A\multimap \phi(A)\tilde{\with}\psi(A)}$. Moreover, $\de{c}$ is obviously successful.
\end{proof}

\begin{proposition}\label{contractionpresquegen}
Let $\cond{A}$ be a behaviour, and $\phi,\psi$ be disjoint delocations. There exists a successful project in the conduct $$\cond{\wn \phi(A)\parr\wn \psi(A)\multimap\wn A}$$
\end{proposition}

\begin{proof}
If $\de{f}\in\cond{\wn \phi(A)\parr \wn \psi(A)}$, then we have $\de{f}\in\wn\cond{(\phi(A)\tilde{\oplus} \psi(A))}$ by \autoref{expiso}. Moreover, we have a successful project $\de{c}$ in $\cond{A^{\pol}}\multimap\cond{\phi(A^{\pol})\tilde{\with} \psi(A^{\pol})}$ using the preceding lemma. Using the successful project implementing functorial promotion we obtain a successful project $\de{c'}\in\cond{\oc A^{\pol}\multimap \oc (\phi(A^{\pol})\tilde{\with} \psi(A^{\pol}))}$. Thus $\de{c'}$ is a successful project in $\wn\cond{\phi(A)\tilde{\oplus} \psi(A)}\multimap \cond{\wn A}$. Finally, we obtain, by composition, that $\de{f}\plug\de{c'}$ is a successful project in $\cond{\wn A}$.
\end{proof}

We also state here an important property of the models, although the proof makes use of the sequent calculus defined later on. We chose to make it appear here as it is a property of the model, and not a consequence of the sequent calculus: one can find a (much more involved) direct proof of this that do not make use of the sequent calculus, for instance by defining by hand the project interpreting the derivation considered and showing that it has the right type and is successful. 

\begin{corollary}
Let $\cond{A,B}$ be behaviours, and $\phi,\psi$ be respective delocations of $\cond{A}$ and $\cond{B}$. There exists a successful project in the conduct $$\cond{\oc(A\with B)\multimap \oc \phi(A)\otimes \oc \psi(B)}$$
\end{corollary}

\begin{proof}
It is obtained as the interpretation of the following derivation (well formed in the sequent calculus we define later on):

\begin{prooftree}
\AxiomC{}
\RightLabel{\scriptsize{ax}}
\UnaryInfC{$\pdash A,A^{\pol};$}
\RightLabel{\scriptsize{$\oplus_{d,2}$}}
\UnaryInfC{$\pdash A^{\pol}\oplus B^{\pol},A$}
\RightLabel{\scriptsize{$\oc$}}
\UnaryInfC{$\oc (A\with B)\pdash ;\oc A$}
\AxiomC{}
\RightLabel{\scriptsize{ax}}
\UnaryInfC{$\pdash B,B^{\pol};$}
\RightLabel{\scriptsize{$\oplus_{d,1}$}}
\UnaryInfC{$\pdash A^{\pol}\oplus B^{\pol},B$}
\RightLabel{\scriptsize{$\oc$}}
\UnaryInfC{$\oc (A\with B)\pdash ;\oc B$}
\RightLabel{\scriptsize{$\otimes^{\mathrm{pol}}$}}
\BinaryInfC{$\oc (A\with B),\oc (A\with B)\pdash ;\oc A\otimes\oc B$}
\RightLabel{\scriptsize{ctr}}
\UnaryInfC{$\oc (A\with B)\pdash ;\oc A\otimes\oc B$}
\end{prooftree}
The fact that it is successful is a consequence of the soundness theorem (\autoref{adeqforteellpol}).
\end{proof}

\subsection{Polarised conducts}


The notions of perennial and co-perennial conducts are not completely satisfactory. In particular, we are not able to show that an implication $\cond{A\multimap B}$ is either perennial or co-perennial when $\cond{A}$ is a perennial conduct (resp. co-perennial) and $\cond{B}$ is a co-perennial conduct (resp. perennial). This is an important issue when one considers the sequent calculus: the promotion rule has to be associated with a rule involving behaviours in order to in the setting of behaviours (using \autoref{tenseurperennecomp}). Indeed, a sequent $\vdash \wn \Gamma,\oc A$ would be interpreted by a conduct which is neither perennial nor co-perennial in general. The sequents considered are for this reason restricted to pre-sequent containing behaviours.

We will define now the notions of negative and positive conducts. The idea is to relax the notion of perennial conduct in order to obtain a notion \emph{negative conduct}. The main interest of this approach is that positive/negative conducts will share the important properties of perennial/co-perennial conducts while interacting in a better way with connectives. In particular, we will be able to interpret the usual functorial promotion (not associated to a $\otimes$ rule), and we will be able to use the contraction rule without all the restrictions we had in the previous section. 

\begin{definition}[Polarised Conducts]
A positive conduct $\cond{P}$ is a conduct satisfying the inflation property and containing all daemons:
\begin{itemize}[noitemsep,nolistsep]
\item $\de{p}\in\cond{P}\Rightarrow\de{p+\lambda 0}\in\cond{P}$;
\item $\forall \lambda\in\realN-\{0\},~\de{Dai}_{\lambda}=(\lambda,(V^{P},\emptyset))\in\cond{P}$.
\end{itemize}
A conduct $\cond{N}$ is negative when its orthogonal $\cond{N}^{\pol}$ is a positive conduct.
\end{definition}

\begin{proposition}
A perennial conduct is negative. A co-perennial conduct is positive.
\end{proposition}

\begin{proof}
We already showed that the perennial conducts satisfy the inflation property (\autoref{coperinflation}) and contain daemons (\autoref{coperdemon}).
\end{proof}

\begin{proposition}
A conduct $\cond{A}$ is negative if and only if:
\begin{itemize}[noitemsep,nolistsep]
\item $\cond{A}$ contains only wager-free projects;
\item $\de{a}\in\cond{A}\Rightarrow \unit{A}\neq 0$.
\end{itemize}
\end{proposition}

\begin{proof}
If $\cond{A}^{\pol}$ is a positive conduct, then it is non-empty and satisfies the inflation property, thus $\cond{A}$ contains only wager-free projects by \autoref{wagerfreeorth}. As a consequence, if $\de{a}\in\cond{A}$, we have that $\sca{a}{Dai_{\text{$\lambda$}}}=\lambda\unit{A}$ thus the condition $\sca{a}{Dai}\neq 0$ implies that $\unit{A}\neq 0$.\\
Conversely, if $\cond{A}$ satisfies that stated properties, we distinguish two cases. If $\cond{A}$ is empty, then is it clear that $\cond{A}^{\pol}$ is a positive conduct. Otherwise, $\cond{A}$ is a non-empty conduct containing only wager-free projects, thus $\cond{A}^{\pol}$ satisfies the inflation property (\autoref{fellorth}). Moreover, $\sca{a}{Dai}=\unit{A}\lambda\neq 0$ as a consequence of the second condition and therefore $\de{Dai}\in\cond{A}^{\pol}$. Finally, $\cond{A}^{\pol}$ is a positive conduct, which implies that $\cond{A}$ is a negative conduct.
\end{proof}

The polarised conducts do not interact very well with the connectives $\tilde{\with}$ and $\tilde{\oplus}$. Indeed, if $\cond{A,B}$ are negative conducts, the conduct $\cond{A\tilde{\with}B}$ is generated by a set of wager-free projects, but it does not satisfy the second property needed to be a negative conduct. Similarly, if $\cond{A,B}$ are positive conducts, then $\cond{A\tilde{\with}B}$ will obviously have the inflation property, but it will contain the project $\de{Dai}_{0}$ (which implies that any element $\de{c}$ in its orthogonal is such that $\unit{C}=0$). We are also not able to characterise in any way the conduct $\cond{A\tilde{\with} B}$ when $\cond{A}$ is a positive conduct and $\cond{B}$ is a negative conduct, except that it is has the inflation property. However, the notions of positive and negative conducts interacts in a nice way with the connectives $\otimes,\with,\parr,\oplus$.


\begin{proposition}
The tensor product of negative conducts is a negative conduct. The $\with$ of negative conducts is a negative conduct. The $\oplus$ of negative conducts is a negative conduct.
\end{proposition}

\begin{proof}
We know that $\cond{A\otimes B}=\emptyset$ if one of the two conducts $\cond{A}$ and $\cond{B}$ is empty, which leaves us to treat the non-empty case. In this case, $\cond{A\otimes B}=(\cond{A}\odot\cond{B})^{\pol\pol}$ is the bi-orthogonal of a non-empty set of wager-free projects. Thus $(\cond{A\otimes B})^{\pol}$ satisfies the inflation property. Moreover $\sca{a\otimes b}{Dai}=\unit{B}\unit{A}\lambda$ which is different from zero since $\unit{A},\unit{B}$ both are different from zero. Thus $\de{Dai}\in(\cond{A\otimes B})^{\pol}$, which shows that $\cond{A\otimes B}$ is a negative conduct since $\cond{(A\otimes B)^{\pol}}$ is a positive conduct.\\
The set $\cond{A}^{\pol}{\uparrow_{{B}}}$ contains all daemons $\de{Dai}_{\lambda}\otimes \de{0}=\de{Dai}_{\lambda}$, and $\de{Dai}\in\cond{A}^{\pol}$. It has the inflation property since $\de{(b+\lambda 0)\otimes 0}=\de{b}\otimes \de{0}+\lambda\de{0}$. Thus $((\cond{A}^{\pol}){\uparrow_{{B}}})^{\pol}$ is a negative conduct. Similarly, $((\cond{B}^{\pol}){\uparrow_{{B}}})^{\pol}$ is a negative conduct, and their intersection is a negative conduct since the properties defining negative conducts are are preserved by intersection. As a consequence, $\cond{A\with B}$ is a negative conduct.\\
In the case of $\oplus$, we will use the fact that $\cond{A\oplus B}=(\cond{A}{\uparrow_{{B}}}\cup\cond{B}{\uparrow_{{A}}})^{\pol}$. If $\de{a}\in\cond{A}$, $\de{a\otimes 0}=\de{b}$ has a null wager and $\unit{B}=\unit{A}\neq 0$. If $\cond{A}$ is empty, $(\cond{A}{\uparrow_{{B}}})^{\pol}$ is a positive conduct. If $\cond{A}$ is non-empty, then \autoref{fellorth} allows us to state that $(\cond{A}{\uparrow_{{B}}})^{\pol}$ has the inflation property. Moreover, the fact that all elements in $\de{a\otimes 0}=\de{b}$ satisfy $\unit{B}\neq0$ implies that $\de{Dai}_{\lambda}\in(\cond{A}{\uparrow_{{B}}})^{\pol}$ for all $\lambda\neq 0$. Therefore, $(\cond{A}{\uparrow_{{B}}})^{\pol}$ is a positive conduct. As a consequence, $\cond{A}{\uparrow_{{B}}}$ is a negative conduct. We show in a similar way that $\cond{B}{\uparrow_{{A}}}$ is a negative conduct. We can deduce from this that $\cond{A}{\uparrow_{{B}}}\cup\cond{B}{\uparrow_{{A}}}$ contains only projects $\de{c}$ with zero wager and such that $\unit{C}\neq 0$. Finally, we showed that $\cond{A\oplus B}$ is a negative conduct.
\end{proof}

\begin{corollary}
The $\parr$ of positive conducts is a positive conduct, the $\with$ of positive conducts is a positive conduct, and the $\oplus$ of positive conducts is a positive conduct.
\end{corollary}

\begin{proposition}
Let $\cond{A}$ be a positive conduct and $\cond{B}$ be a negative conduct. Then $\cond{A\otimes B}$ is a positive conduct.
\end{proposition}

\begin{proof}
Pick $\de{f}\in\cond{(A\otimes B)^{\pol}}=\cond{B\multimap A^{\pol}}$. Then for all $\de{b}\in\cond{B}$, $\de{f\plug b}=(\unit{B} f+\unit{F} b,F\plug B)$ is an element of $\cond{A}^{\pol}$. Since $\cond{A}^{\pol}$ is a negative conduct, we have that $\unit{F}\unit{B}\neq 0$ and $\unit{B}f+\unit{F}b=0$. Thus $\unit{F}\neq 0$. Moreover, $\cond{B}$ is a negative conduct, therefore $\unit{B}\neq 0$ and $b=0$. The condition $\unit{B}f+\unit{F}b=0$ then becomes $\unit{B}f=0$, i.e.\ $f=0$.

Thus $\cond{(A\otimes B)^{\pol}}$ is a negative conduct, which implies that $\cond{A\otimes B}$ is a positive conduct.
\end{proof}

\begin{corollary}
If $\cond{A}$ is a positive conduct and $\cond{B}$ is a positive conduct, $\cond{A\multimap B}=\cond{(A\otimes B^{\pol})^{\pol}}$ is a positive conduct.
\end{corollary}

\begin{corollary}
If $\cond{A,B}$ are negative conducts, then $\cond{A\multimap B}$ is a negative conduct.
\end{corollary}

\begin{proof}
We know that $\cond{A\multimap B}=\cond{(A\otimes B^{\pol})^{\pol}}$. We also just showed that $\cond{A\otimes B^{\pol}}$ is a positive conduct, thus $\cond{A\multimap B}$ is a negative conduct.
\end{proof}

\begin{proposition}
The tensor product of a negative conduct and a behaviour is a behaviour.
\end{proposition}

\begin{proof}
Let $\cond{A}$ be a negative conduct and $\cond{B}$ be a behaviour. If either $\cond{A}$ or $\cond{B}$ is empty (or both), $\cond{(A\otimes B)^{\pol}}$ equals $\cond{T}_{V^{A}\cup V^{B}}$ and we are done. We now suppose that $\cond{A}$ and $\cond{B}$ are both non empty.

Since $\cond{A,B}$ contain only wager-free projects, the set $\{\de{a\otimes b}~|~\de{a}\in\cond{A},\de{b}\in\cond{B}\}$ contains only wager-free projects. Thus $\cond{(A\otimes B)^{\pol}}$ has the inflation property: this is a consequence of \autoref{fellorth}. Suppose now that there exists $\de{f}\in\cond{(A\otimes B)^{\pol}}$ such that $f\neq 0$. Choose $\de{a}\in\cond{A}$ and $\de{b}\in\cond{B}$. Then $\sca{f}{a\otimes b}=f\unit{B}\unit{A}+\meas{F,A\plug B}$. Since $\unit{A}\neq 0$, we can define $\mu=-\meas{F,A\cup B}/(\unit{A}f)$, and $\de{b+\mu 0}\in\cond{B}$ since $\cond{B}$ has the inflation property. We then have:
\begin{eqnarray*}
\sca{f}{a\otimes(b+\mu0)}&=&f\unit{A}(\unit{B}+\mu)+\meas{F,A\plug(B+\mu 0)}\\
&=&f\unit{A}\frac{-\meas{F,A\cup B}}{\unit{A}f}+\meas{F,A\plug B}\\
&=&0
\end{eqnarray*}
This is a contradiction, since $\de{f}\in\cond{(A\otimes B)^{\pol}}$. Thus $f=0$.

Finally, we have shown that $\cond{(A\otimes B)^{\pol}}$ has the inflation property and contains only wager-free projects.
\end{proof}

\begin{corollary}
If $\cond{A}$ is a negative conduct and $\cond{B}$ is a behaviour, $\cond{A\multimap B}$ is a behaviour.
\end{corollary}

\begin{figure}
\centering
\subfigure[Tensor]{
\begin{tabular}{c||c|c}
$\otimes$ & N & P\\
\hline\hline
N & N & P\\
P & P & ?
\end{tabular}
}
\subfigure[Parr]{
\begin{tabular}{c||c|c}
$\parr$ & N & P\\
\hline\hline
N & ? & N\\
P & N & P
\end{tabular}
}
\subfigure[With(1)]{
\begin{tabular}{c||c|c}
$\with$ & N & P\\
\hline\hline
N & N & ?\\
P & ? & P
\end{tabular}
}
\subfigure[Plus(1)]{
\begin{tabular}{c||c|c}
$\oplus$ & N & P\\
\hline\hline
N & N & ?\\
P & ? & P
\end{tabular}
}
\caption{Connectives and Polarisation}
\end{figure}

\begin{proposition}
The weakening (on the left) of negative conducts holds.
\end{proposition}

\begin{proof}
Let $\cond{A,B}$ be conducts, $\cond{N}$ be a negative conduct, and pick $\de{f}\in\cond{A\multimap B}$. We will show that $\de{f}\otimes \de{0}_{V^{N}}$ is an element of $\cond{A\otimes N\multimap B}$. For this, we pick $\de{a}\in\cond{A}$ and $\de{n}\in\cond{N}$. Then for all $\de{b'}\in\cond{B^{\pol}}$,
\begin{eqnarray*}
\lefteqn{\sca{(f\otimes 0)\plug(a\otimes n)}{b'}}\\
&=&\sca{f\otimes 0}{(a\otimes n)\otimes b}\\
&=&\sca{f\otimes 0}{(a\otimes b')\otimes n}\\
&=&\unit{F}(\unit{A}\unit{B'}n+\unit{N}\unit{A}b'+\unit{N}\unit{B'}a)+\unit{N}\unit{A}\unit{B'}f+\meas{F\cup 0,A\cup B'\cup N}\\
&=&\unit{F}(\unit{N}\unit{A}b'+\unit{N}\unit{B'}a)+\unit{N}\unit{A}\unit{B'}f+\meas{F\cup 0,A\cup B'\cup N}\\
&=&\unit{N}(\unit{F}(\unit{A}b'+\unit{B'}a)+\unit{A}\unit{B'}f)+\unit{N}\meas{F,A\cup B'}\\
&=&\unit{N}\sca{f}{a\otimes b'}
\end{eqnarray*}
Since $\unit{N}\neq 0$, $\sca{(f\otimes 0)\plug(a\otimes n)}{b'}\neq 0,\infty$ if and only if  $\sca{f\plug a}{b'}\neq0,\infty$. Therefore, for all $\de{a\otimes n}\in\cond{A\odot N}$, $\de{(f\otimes 0)\plug (a\otimes n)}\in\cond{B}$. This shows that $\de{f\otimes 0}$ is an element of $\cond{A\otimes N\multimap B}$ by \autoref{propositionethsendcondsend}.
\end{proof}

\subsection{Sequent Calculus and Soundness}

We now describe a sequent calculus which is much closer to the usual sequent calculus for Elementary Linear Logic. We introduce once again three types of formulas: (B)ehaviors, (P)ositive, (N)egative. The sequents we will be working with will be the equivalent to the notion of pre-sequent introduced earlier.

\begin{definition}
We once again define three types of formulas -- (B)ehavior, (P)ositive, (N)egative -- by the following grammar:
\begin{eqnarray*}
B&:=& X~|~X^{\pol}~|~\cond{0}~|~\cond{T}~|~B\otimes B~|~B\parr B~|~B\oplus B~|~B\with B~|~\forall X~B~|~\exists X~B~|~B\otimes N~|~B\parr P\\
N&:=& \cond{1}~|~\oc B~|~\oc N~|~N\otimes N~|~ N\with N~|~N\oplus N~|~N\parr P\\
P&:=&  \cond{\bot}~|~\wn B~|~\wn P~|~P\parr P~|~ P\with P~|~P\oplus P~|~N\otimes P
\end{eqnarray*}
\end{definition}

\begin{definition}
A sequent $\Delta\pdash \Gamma;\Theta$ is such that $\Delta,\Theta$ contain only negative formulas, $\Theta$ containing at most one formula and $\Gamma$ containing only behaviours.
\end{definition}

\begin{definition}[The System $\textnormal{ELL}_{\textnormal{pol}}$]
A proof in the system $\textnormal{ELL}_{\textnormal{pol}}$ is a derivation tree constructed from the derivation rules shown in \autoref{ellpol}.
\end{definition}

\begin{figure}
\centering
\subfigure[Identity Group]{
\framebox{
\begin{tabular}{c}
\begin{minipage}{5.1cm}
\begin{prooftree}
\AxiomC{}
\RightLabel{\scriptsize{ax}}
\UnaryInfC{$\pdash B^{\pol},B;$}
\end{prooftree}
\end{minipage}
\\~\\
\begin{tabular}{cc}
\begin{minipage}{5.8cm}
\begin{prooftree}
\AxiomC{\!\!$\Delta_{1} \pdash \Gamma_{1};N$\!\!}
\AxiomC{\!\!$\Delta_{2},N \pdash \Gamma_{2};\Theta$\!\!}
\RightLabel{\scriptsize{cut$^{\mathrm{pol}}$}}
\BinaryInfC{$\Delta_{1},\Delta_{2}\pdash \Gamma_{1},\Gamma_{2};\Theta$}
\end{prooftree}
\end{minipage}
&
\begin{minipage}{5.4cm}
\begin{prooftree}
\AxiomC{\!\!$\Delta_{1} \pdash \Gamma_{1},B;\Theta$\!\!}
\AxiomC{\!\!$\Delta_{2} \pdash \Gamma_{2},B^{\pol};$\!\!}
\RightLabel{\scriptsize{cut}}
\BinaryInfC{$\Delta_{1},\Delta_{2}\pdash \Gamma_{1},\Gamma_{2};\Theta$}
\end{prooftree}
\end{minipage}
\end{tabular}
\end{tabular}
}
}
\subfigure[Multiplicative Group]{
\framebox{
\begin{tabular}{cc}
\begin{minipage}{5.6cm}
\begin{prooftree}
\AxiomC{\!\!$\Delta_{1}\pdash \Gamma_{1},B_{1};\Theta$\!\!}
\AxiomC{\!\!$\Delta_{2}\pdash \Gamma_{2},B_{2};$\!\!}
\RightLabel{\scriptsize{$\otimes$}}
\BinaryInfC{$\Delta_{1},\Delta_{2}\pdash\Gamma_{1},\Gamma_{2},B_{1}\otimes B_{2};\Theta$}
\end{prooftree}
\end{minipage}
&
\begin{minipage}{5.95cm}
\begin{prooftree}
\AxiomC{$\Delta \pdash \Gamma,B_{1},B_{2};\Theta$}
\RightLabel{\scriptsize{$\parr$}}
\UnaryInfC{$\Delta\pdash \Gamma,B_{1}\parr B_{2};\Theta$}
\end{prooftree}
\end{minipage}
\\~\\
\begin{minipage}{5.6cm}
\begin{prooftree}
\AxiomC{$\Delta, N_{1},N_{2} \pdash \Gamma;\Theta$}
\RightLabel{\scriptsize{$\otimes^{\mathrm{pol}}_{g}$}}
\UnaryInfC{$\Delta,N_{1}\otimes N_{2}\pdash \Gamma;\Theta$}
\end{prooftree}
\end{minipage}
&
\begin{minipage}{5.95cm}
\begin{prooftree}
\AxiomC{\!$\Delta_{1} \pdash \Gamma_{1};N_{1}$\!}
\AxiomC{\!$\Delta_{2} \pdash \Gamma_{2};N_{2}$\!}
\RightLabel{\scriptsize{$\otimes^{\mathrm{pol}}_{d}$}}
\BinaryInfC{$\Delta_{1},\Delta_{2}\pdash \Gamma_{1},\Gamma_{2};N_{1}\otimes N_{2}$}
\end{prooftree}
\end{minipage}
\\~\\
\begin{minipage}{5.6cm}
\begin{prooftree}
\AxiomC{$\Delta, P_{1}^{\pol} \pdash \Gamma;N_{2}$}
\RightLabel{\scriptsize{$\parr^{\mathrm{pol}}_{d}$}}
\UnaryInfC{$\Delta \pdash \Gamma ; P_{1}\parr N_{2}$}
\end{prooftree}
\end{minipage}
&
\begin{minipage}{5.95cm}
\begin{prooftree}
\AxiomC{\!\!$\Delta_{1} \pdash\Gamma_{1} ;P^{\pol}_{1}$\!\!}
\AxiomC{\!\!$\Delta_{2},N_{2} \pdash\Gamma_{2} ;\Theta$\!\!}
\RightLabel{\scriptsize{$\parr^{\mathrm{pol}}_{g}$}}
\BinaryInfC{$\Delta_{1},\Delta_{2},P_{1}\parr N_{2}\pdash \Gamma_{1},\Gamma_{2} ;\Theta$}
\end{prooftree}
\end{minipage}
\\~\\
\begin{minipage}{5.6cm}
\begin{prooftree}
\AxiomC{$\Delta,P^{\pol} \pdash \Gamma,B;\Theta$}
\RightLabel{\scriptsize{$\parr^{\mathrm{mix}}$}}
\UnaryInfC{$\Delta \pdash \Gamma,P\parr B ; \Theta$}
\end{prooftree}
\end{minipage}
&
\begin{minipage}{5.95cm}
\begin{prooftree}
\AxiomC{\!$\Delta_{1} \pdash \Gamma_{1};N$\!}
\AxiomC{\!$\Delta_{2} \pdash \Gamma_{2},B;\Theta$\!}
\RightLabel{\scriptsize{$\otimes^{\mathrm{mix}}$}}
\BinaryInfC{$\Delta_{1},\Delta_{2} \pdash \Gamma_{1},\Gamma_{2},N\otimes B ; \Theta$}
\end{prooftree}
\end{minipage}
\\~\\
\begin{minipage}{5.6cm}
\begin{prooftree}
\AxiomC{}
\RightLabel{\scriptsize{$\cond{1}_{d}$}}
\UnaryInfC{$\pdash ;\cond{1}$}
\end{prooftree}
\end{minipage}
&
\begin{minipage}{5.95cm}
\begin{prooftree}
\AxiomC{$\Delta \pdash \Gamma;\Theta$}
\RightLabel{\scriptsize{$\cond{1}_{g}$}}
\UnaryInfC{$\Delta, \cond{1} \pdash \Gamma;\Theta$}
\end{prooftree}
\end{minipage}
\end{tabular}
}
}
\subfigure[Additive Group]{
\framebox{
\begin{tabular}{cc}
\begin{minipage}{5.6cm}
\begin{prooftree}
\AxiomC{$\Delta\pdash\Gamma,B_{i};\Theta$}
\RightLabel{\scriptsize{$\oplus_{i}$}}
\UnaryInfC{$\Delta\pdash \Gamma,B_{1}\oplus B_{2};\Theta$}
\end{prooftree}
\end{minipage}
&
\begin{minipage}{5.95cm}
\begin{prooftree}
\AxiomC{$\Delta\pdash \Gamma, B_{1};\Theta$}
\AxiomC{$\Delta\pdash \Gamma, B_{2};\Theta$}
\RightLabel{\scriptsize{$\with$}}
\BinaryInfC{$\Delta\pdash \Gamma, B_{1}\with B_{2};\Theta$}
\end{prooftree}
\end{minipage}
\\~\\
\begin{minipage}{5.6cm}
\begin{prooftree}
\AxiomC{}
\RightLabel{\scriptsize{$\top$}}
\UnaryInfC{$\Delta\pdash \Gamma, \top;\Theta$}
\end{prooftree}
\end{minipage}
&
\begin{minipage}{5.95cm}
\centering
No rules for $0$.
\end{minipage}
\end{tabular}
}
}
\subfigure[Exponential Group]{
\framebox{
\begin{tabular}{cc}
\begin{minipage}{5.6cm}
\begin{prooftree}
\AxiomC{$\Delta\pdash \Gamma;N$}
\RightLabel{\scriptsize{$\oc^{\mathrm{pol}}$}}
\UnaryInfC{$\oc \Delta,\oc \Gamma^{\pol}\pdash ;\oc N$}
\end{prooftree}
\end{minipage}
&
\begin{minipage}{5.95cm}
\begin{prooftree}
\AxiomC{$\Delta \pdash \Gamma,B;$}
\RightLabel{\scriptsize{$\oc$}}
\UnaryInfC{$\oc \Delta,\oc \Gamma^{\pol}\pdash ;\oc B$}
\end{prooftree}
\end{minipage}
\\~\\
\begin{minipage}{5.6cm}
\begin{prooftree}
\AxiomC{$\Delta,\oc B,\oc B\pdash \Gamma;\Theta$}
\RightLabel{\scriptsize{ctr (B\text{ behaviour})}}
\UnaryInfC{$\Delta,\oc B\pdash \Gamma;\Theta$}
\end{prooftree}
\end{minipage}
&
\begin{minipage}{5.95cm}
\begin{prooftree}
\AxiomC{$\Delta\pdash\Gamma;\Theta$}
\RightLabel{\scriptsize{weak}}
\UnaryInfC{$\Delta,N\pdash \Gamma;\Theta$}
\end{prooftree}
\end{minipage}
\end{tabular}
}
}
\subfigure[Quantifier Group]{
\framebox{
\begin{tabular}{c}
\begin{tabular}{cc}
\begin{minipage}{5.6cm}
\begin{prooftree}
\AxiomC{$\Delta\pdash \Gamma,C;\Theta$}
\AxiomC{$X\not\in \freevar{\Gamma,\Delta,\Theta}$}
\RightLabel{\scriptsize{$\forall$}}
\BinaryInfC{$\Delta\pdash\Gamma,\forall X~ C;\Theta$}
\end{prooftree}
\end{minipage}
&
\begin{minipage}{5.6cm}
\begin{prooftree}
\AxiomC{$\Delta\pdash \Gamma,C[A/X];\Theta$}
\RightLabel{\scriptsize{$\exists$}}
\UnaryInfC{$\Delta\pdash \Gamma,\exists X~C;\Theta$}
\end{prooftree}
\end{minipage}
\end{tabular}
\end{tabular}
}
}
\caption{Sequent Calculus $\textnormal{ELL}_{\textnormal{pol}}$}\label{ellpol}
\end{figure}

\begin{remark}
Even though one can consider the conduct $\cond{A\with B}$ when $\cond{A,B}$ are negative conducts, no rule of the sequent calculus $\textnormal{ELL}_{\textnormal{pol}}$ allows one to construct such a formula. The reason for that is simple: since in this case the set $\cond{A+B}$ is not necessarily included in the conduct $\cond{A\with B}$, one cannot interpret the rule in general (since distributivity does not necessarily holds). The latter can be interpreted when the context contains at least one behaviour, but imposing such a condition on the rule could lead to difficulties when considering the cut-elimination procedure (in case of commutations). We therefore whose to work with a system in which one introduces additive connectives only between behaviours. Notice however that a formula built with an additive connective between negative sub-formulas can still be introduced by a weakening rule.
\end{remark}

It is quite clear that the obtained system is a subsystem of a decorated version of linear logic. As such it inherits linear logic cut-elimination procedure and its properties. In particular, let us recall that proof nets for linear logic were shown to have the strong normalisation property \cite{PaganiSN}. As a consequence, linear logic sequent calculus (and the subsystem considered here) is not only weakly normalising, but satisfies some stronger property. Indeed, thinking about proof nets as equivalence classes of sequent calculus proofs modulo commutation rules, we deduce that any non-normalising sequence of reduction in the sequent calculus gives rise to an eventually constant sequence of reduction on the proof net side. i.e.\ a non-terminating sequence of reduction is a finite sequence of reduction followed by a sequence of commutation rules only. For the purpose of stating the next theorem, let us call this property \emph{strong normalisation up to commutations}.

\begin{proposition}
The system $\textnormal{ELL}_{\textnormal{pol}}$ possesses a cut-elimination procedure which is strongly normalising up to commutations.
\end{proposition}

We now define the interpretation of the formulas and proofs of the localised sequent calculus in the model $\vaguemodel{\Omega}{mi}{m}$.

\begin{definition}
We fix $\mathcal{V}=\{X_{i}(j)\}_{i,j\in\naturalN\times\integerN}$ a set of \emph{localised variables}. For $i\in\naturalN$, the set $X_{i}=\{X_{i}(j)\}_{j\in\integerN}$ will be referred to as \emph{the name of the variable $X_{i}$}, and an element of $X_{i}$ will be referred to as a \emph{variable of name $X_{i}$}.
\end{definition}
For $i,j\in\naturalN\times\integerN$ we define the \emph{location} $\sharp X_{i}(j)$ of the variable $X_{i}(j)$ as the set $$\{x\in\realN~|~ 2^{i}(2j+1)\leqslant x< 2^{i}(2j+1)+1\}$$

\begin{definition}[Formulas of $\textnormal{locELL}_{\textnormal{pol}}$]
We inductively define the formulas of $\textnormal{locELL}_{\textnormal{pol}}$ together with their \emph{locations} as follows:
\begin{itemize}[noitemsep,nolistsep]
\item \textbf{Behaviours}:
\begin{itemize}[noitemsep,nolistsep]
\item A variable $X_{i}(j)$ of name $X_{i}$ is a behaviour whose location is defined as $\sharp X_{i}(j)$;
\item If $X_{i}(j)$ is a variable of name $X_{i}$, then $(X_{i}(j))^{\pol}$ is a behaviour of location $\sharp X_{i}(j)$.
\item The constants $\cond{T}_{\sharp \Gamma}$ are behaviours of location $\sharp\Gamma$;
\item The constants $\cond{0}_{\sharp\Gamma}$ are behaviours of location $\sharp\Gamma$.
\item If $A,B$ are behaviours of respective locations $X,Y$ such that $X\cap Y=\emptyset$, then $A\otimes B$ (resp. $A\parr B$, resp. $A\with B$, resp. $A\oplus B$) is a behaviour of location $X\cup Y$;
\item If $X_{i}$ is a variable name, and $A(X_{i})$ is a behaviour of location $\sharp A$, then $\forall X_{i}~A(X_{i})$ and $\exists X_{i}~A(X_{i})$ are behaviours of location $\sharp A$.
\item If $A$ is a negative conduct of location $X$ and $B$ is a behaviour of location $Y$ such that $X\cap Y=\emptyset$, then $A\otimes B$is a behaviour of location $X\cup Y$;
\item If $A$ is a positive conduct of location $X$ and $B$ is a behaviour of location $Y$ such that $X\cap Y=\emptyset$, then $A\parr B$ is a behaviour of location $X\cup Y$;
\end{itemize}
\item \textbf{Negative Conducts}:
\begin{itemize}[noitemsep,nolistsep]
\item The constant $\cond{1}$ is a negative conduct;
\item If $A$ is a behaviour or a negative conduct of location $X$, then $\oc A$ is a negative conduct of location $\Omega(X\times[0,1])$;
\item If $A,B$ are negative conducts of locations $X,Y$ such that $X\cap Y=\emptyset$, then $A\otimes B$ (resp. $A\oplus B$, resp. $A\with B$) is a negative conduct of location $X\cup Y$;
\item If $A$ is a negative conduct of location $X$ and $B$  is a positive conduct of location $Y$, $A\parr B$ is a negative conduct of location $X\cup Y$.
\end{itemize}
\item \textbf{Positive Conducts}:
\begin{itemize}[noitemsep,nolistsep]
\item The constant $\cond{\bot}$ is a positive conduct;
\item If $A$ is a behaviour or a positive conduct of location $X$, then $\wn A$ is a positive conduct of location $\Omega(X\times[0,1])$;
\item If $A,B$ are positive conducts of locations  $X,Y$ such that $X\cap Y=\emptyset$, then $A\parr B$ (resp. $A\with B$, resp. $A\oplus B$) is a positive conduct of location $X\cup Y$;
\item If $A$ is a negative conduct of location $X$ and $B$ is a positive conduct of location $Y$, $A\otimes B$ is a positive conduct of location $X\cup Y$.
\end{itemize}
\end{itemize}
If $A$ is a formula, we will denote by $\sharp A$ its location. We also define sequents $\Delta\pdash \Gamma;\Theta$ of $\textnormal{locELL}_{\textnormal{pol}}$ when:
\begin{itemize}[noitemsep,nolistsep]
\item formulas in $\Gamma\cup\Delta\cup\Theta$ have pairwise disjoint locations;
\item formulas in $\Delta$ and $\Theta$ are negative conducts;
\item there is at most one formula in $\Theta$;
\item $\Gamma$ contains only behaviours.
\end{itemize}
\end{definition}

\begin{definition}[Interpretations]
We define an \emph{interpretation basis} as a function $\Phi$  which maps every variable name $X_{i}$ to a behaviour of carrier $[0,1[$.
\end{definition}

\begin{definition}[Interpretation of $\textnormal{locELL}_{\textnormal{pol}}$ formulas]
Let $\Phi$ be an interpretation basis. We define the interpretation $I_{\Phi}(F)$ along $\Phi$ of a formula $F$ inductively:
\begin{itemize}[noitemsep,nolistsep]
\item If $F=X_{i}(j)$, then $I_{\Phi}(F)$ is the delocation (i.e.\ a behaviour) of $\Phi(X_{i})$ along the function $x\mapsto 2^{i}(2j+1)+x$;
\item If $F=(X_{i}(j))^{\pol}$, we define the behaviour $I_{\Phi}(F)=(I_{\Phi}(X_{i}(j)))^{\pol}$;
\item If $F=\cond{T}_{\sharp\Gamma}$ (resp. $F=\cond{0}_{\sharp\Gamma}$), we define $I_{\Phi}(F)$ as the behaviour $\cond{T}_{\sharp\Gamma}$ (resp. $\cond{0}_{\sharp\Gamma}$);
\item If $F=\cond{1}$ (resp. $F=\cond{\bot}$), we define $I_{\Phi}(F)$ as the behaviour $\cond{1}$ (resp. $\cond{\bot}$);
\item If $F=A\otimes B$, we define the conduct $I_{\Phi}(F)=I_{\Phi}(A)\otimes I_{\Phi}(B)$;
\item If $F=A\parr B$, we define the conduct $I_{\Phi}(F)=I_{\Phi}(A)\parr I_{\Phi}(B)$;
\item If $F=A\oplus B$, we define the conduct $I_{\Phi}(F)=I_{\Phi}(A)\oplus I_{\Phi}(B)$;
\item If $F=A\with B$, we define the conduct $I_{\Phi}(F)=I_{\Phi}(A)\with I_{\Phi}(B)$;
\item If $F=\forall X_{i} A(X_{i})$, we define the conduct $I_{\Phi}(F)=\cond{\forall X_{i}} I_{\Phi}(A(X_{i}))$;
\item If $F=\exists X_{i} A(X_{i})$, we define the conduct $I_{\Phi}(F)=\cond{\exists X_{i}} I_{\Phi}(A(X_{i}))$.
\item If $F=\oc A$ (resp. $\wn A$), we define the conduct $I_{\Phi}(F)=\oc I_{\Phi}(A)$ (resp. $\wn I_{\Phi}(A)$).
\end{itemize}
Moreover a sequent $\Delta\vdash \Gamma;\Theta$ will be interpreted as the $\parr$ of the formulas in $\Gamma$ and $\Theta$ and the negations of formulas in $\Delta$, which we will write $\bigparr\Delta^{\pol}\parr\bigparr \Gamma\parr\bigparr\Theta$. We will also represent this formula by the equivalent formula $\bigotimes\Delta\multimap (\bigparr\Gamma\parr\bigparr\Theta)$.
\end{definition}

\begin{definition}[Interpretation of $\textnormal{locELL}_{\textnormal{pol}}$ proofs]\label{interpretationpreuvesellpol}
Let $\Phi$ be an interpretation basis. We define the interpretation $I_{\Phi}(\pi)$ -- a project -- of a proof $\pi$ inductively:
\begin{itemize}[noitemsep,nolistsep]
\item if $\pi$ consists in an axiom rule introducing $\vdash (X_{i}(j))^{\pol},X_{i}(j')$, we define $I_{\Phi}(\pi)$ as the project $\de{Fax}$ defined by the translation $x \mapsto 2^{i}(2j'-2j)+x$;
\item if $\pi$ consists solely in a $\cond{T}_{\sharp \Gamma}$ rule, we define $I_{\Phi}(\pi)=\de{0}_{\sharp\Gamma}$;
\item if $\pi$ consists solely in a $\cond{1}_{d}$ rule, we define $I_{\Phi}(\pi)=\de{0}_{\emptyset}$;
\item if $\pi$ is obtained from $\pi'$ by a $\parr$ rule, a $\otimes_{g}^{\mathrm{pol}}$ rule, a $\parr_{d}^{\mathrm{pol}}$ rule, a $\parr^{\mathrm{mix}}$ rule, or a $\cond{1}_{g}$ rule, then $I_{\Phi}(\pi)=I_{\Phi}(\pi')$;
\item if $\pi$ is obtained from $\pi_{1}$ and $\pi_{2}$ by applying a $\otimes$ rule, a $\otimes^{\mathrm{pol}}_{d}$ rule,  a $\parr_{g}^{\mathrm{pol}}$ rule or a $\otimes^{\mathrm{mix}}$ rule, we define $I_{\Phi}(\pi)=I_{\Phi}(\pi_{1})\otimes I_{\Phi}(\pi')$;
\item if $\pi$ is obtained from $\pi'$ by a $\text{weak}$ rule or a $\oplus_{i}$ rule introducing a formula of location $V$, we define $I_{\Phi}(\pi)=I_{\Phi}(\pi')\otimes\de{0}_{V}$;
\item if $\pi$ of conclusion $\vdash \Gamma, A_{0}\with A_{1}$ is obtained from $\pi_{0}$ and $\pi_{1}$ by applying a $\with$ rule, we define the interpretation of $\pi$ as it was done in our earlier paper \cite{seiller-goiadd}: ;
\item If $\pi$ is obtained from a $\forall$ rule applied to a derivation $\pi'$, we define $I_{\Phi}(\pi)=I_{\Phi}(\pi')$;
\item If $\pi$ is obtained from a $\exists$ rule applied to a derivation $\pi'$ replacing the formula $\cond{A}$ by the variable name $X_{i}$, we define $I_{\Phi}(\pi)=I_{\Phi}(\pi')\plug (\bigotimes [e^{-1}(j)\leftrightarrow X_{i}(j)])$, using the notations of our previous paper \cite{seiller-goig} for the \emph{measure-inflating faxes} $[e^{-1}(j)\leftrightarrow X_{i}(j)]$ where $e$ is an enumeration of the occurrences of $\cond{A}$ in $\pi'$;
\item if $\pi$ is obtained from $\pi'$ by applying a promotion rule $\oc$ or $\oc^{\mathrm{pol}}$, we apply the implementation of the functorial promotion rule to the project $\oc I_{\Phi}(\pi')$ $n-1$ times, where $n$ is the number of formulas in the sequent;
\item if $\pi$ is obtained from $\pi$ by applying a contraction rule $ctr$, we define the interpretation of $\pi$ as the execution between the interpretation of $\pi'$ and the project implementing contraction described in \autoref{contractiongen};
\item if $\pi$ is obtained from $\pi_{1}$ and $\pi_{2}$ by applying a $cut$ rule or a $cut^{\mathrm{pol}}$ rule, we define $I_{\Phi}(\pi)=I_{\Phi}(\pi_{1})\exec I_{\Phi}(\pi_{2})$.
\end{itemize}
\end{definition}

Once again, one can choose an enumeration $e$ of the occurrences of variables in order to ``localise'' any formula $A$ and any proof $\pi$ of $\textnormal{ELL}_{\textnormal{pol}}$: and define formulas $A^{e}$ and proofs $\pi^{e}$ of $\textnormal{locELL}_{\textnormal{pol}}$. One easily shows a soundness result for the localised calculus $\textnormal{locELL}_{\textnormal{pol}}$ which implies the following result.

\begin{theorem}\label{adeqforteellpol}
Let $\Phi$ be an interpretation basis, $\pi$ a proof of $\textnormal{ELL}_{\textnormal{pol}}$ of conclusion $\Delta\pdash \Gamma;\Theta$, and $e$ an enumeration of the occurrences of variables in the axioms of $\pi$. Then $I_{\Phi}(\pi^{e})$ is a successful project in $I_{\Phi}(\Delta^{e}\vdash \Gamma^{e};\Theta^{e})$.
\end{theorem}

\section{Conclusion and Perspectives}

In this paper, we extended the setting of Interaction Graphs in order to deal with all connectives of linear logic. We showed how one can obtain a soundness result for two versions of Elementary Linear Logic. The first system, which is conceived so that the interpretation of sequents are behaviours, seems to lack expressivity and it may appear that elementary functions cannot be typed in this system. The second system is however very close to usual \ELL sequent calculus and, even though one we do not provide a formal proof of it here, should lead to a characterisation of elementary functions from natural numbers to natural numbers, as it is the case with traditional Elementary Linear Logic \cite{danosjoinet}. 

Though the generalisation from graphs to graphings may seem a big effort, we believe the resulting framework to be extremely interesting. We should stress that with little work on the definition of exponentials, one should be able to show that interpretations of proofs can be described by finite means. Indeed, the only operation that seems to turn an interval into an infinite number of intervals is the promotion rule. One should however be able to show that, up to a suitable delocation, the promotion of a project defined on a finite number of rational intervals is defined on a finite number of rational intervals. 


More generally, now that this framework has been defined and that we have shown its interest by providing a construction for \emph{elementary exponentials}, we believe the definition and study of other exponential connectives may be a work of great interest. First, these new exponentials would co-exist with each other, making it possible to study their interactions. Secondly, even if the definition of exponentials for full linear logic may be a complicated task, the definition of low-complexity exponentials may be of great interest.

Finally, we explained in our previous paper how the systematic construction of models of linear logic based on graphings \cite{seiller-goig} give rise to a hierarchy of models mirroring subtle distinctions concerning computational principles. In particular, it gives rise to a hierarchy of models characterising complexity classes \cite{seiller-lcc14, seiller-goinda} by adapting results obtained using operator theory \cite{seiller-conl,seiller-lsp,seiller-masas}. The present work could lead to characterisations of larger complexity classes such as \textbf{Ptime} or \textbf{Exptime} predicates and/or functions, following the work of Baillot \cite{baillot}.


\bibliographystyle{alpha}
\bibliography{thomas}

\end{document}